\documentclass[12pt]{article}
\usepackage{amsmath}
\usepackage{graphicx}
\usepackage{enumerate}
\usepackage[numbers]{natbib}
\usepackage{url} % not crucial - just used below for the URL 

\newcommand{\blind}{1}

\usepackage{authblk}
\usepackage{amsthm}
\usepackage{booktabs}
\usepackage{epsfig}
\usepackage{amsfonts}
\usepackage{color}
\usepackage{amssymb}
\usepackage{multirow}
\usepackage{algorithm}
\usepackage{comment}

\usepackage{caption}
\usepackage{subcaption}

\newtheorem{theorem}{Theorem}

\newtheorem{remark}{Remark}
\newtheorem{lemma}{Lemma}

\usepackage[colorlinks,bookmarksopen,bookmarksnumbered,citecolor=blue,urlcolor=blue,linkcolor=blue]{hyperref}

\newcommand{\argmax}{\ensuremath{\operatornamewithlimits{arg\,max}}}

\usepackage{xr} %% for external cross-reference 
\renewcommand{\figurename}{Fig.}

\MakeRobust{\ref}% avoid expanding it when in a textual label

\makeatletter
\newcommand{\labeltext}[2]{%
	\@bsphack
	\csname phantomsection\endcsname % in case hyperref is used
	\def\@currentlabel{#1}{\label{#2}}%
	\@esphack
}
\makeatother
\begin{document}

\def\spacingset#1{\renewcommand{\baselinestretch}%
{#1}\small\normalsize} \spacingset{1}

%%%%%%%%%%%%%%%%%%%%%%%%%%%%%%%%%%%%%%%%%%%%%%%%%%%%%%%%%%%%%%%%%%%%%%%%%%%%%%

\if1\blind
{
  \title{\bf LinDA: linear models for differential abundance analysis of microbiome compositional data}
  		\date{}
  \author[1, 2, 3]{Huijuan Zhou (zhouhuijuan@mail.shufe.edu.cn)}
  \author[3]{Kejun He (kejunhe@ruc.edu.cn)}
  \author[4,*]{Jun Chen (chen.jun2@mayo.edu)}
  \author[2,*]{Xianyang Zhang (zhangxiany@stat.tamu.edu)}
  \affil[1]{Shanghai University of Finance and Economics}
  \affil[2]{Texas A\&M University}
  \affil[3]{Renmin University of China}
  \affil[4]{Mayo Clinic}

  \affil[*]{Corresponding author}
  \maketitle
} \fi

\if0\blind
{
  \bigskip
  \bigskip
  \bigskip
  \begin{center}
    {\LARGE\bf LinDA: linear models for differential abundance analysis of microbiome compositional data}
\end{center}
  \medskip
} \fi

\bigskip
\begin{abstract}
%One fundamental statistical task in microbiome data analysis is differential abundance analysis, which aims to identify microbial taxa whose abundance covaries with a variable of interest. 
Differential abundance analysis is at the core of statistical analysis of microbiome data. The compositional nature of microbiome sequencing data makes false positive control challenging. Here, we show that the compositional effects can be addressed by a simple, yet highly flexible and scalable, approach. The proposed method, LinDA, only requires fitting linear regression models on the centered log-ratio transformed data, and correcting the bias due to compositional effects. We show that LinDA enjoys asymptotic FDR control and can be extended to mixed-effect models for correlated microbiome data. Using simulations and real examples, we demonstrate the effectiveness of LinDA.
%We implemented the proposed method in the R package \texttt{LinDA} (\texttt{https://github.com/zhouhj1994/LinDA}).
\end{abstract}

\noindent%
\vfill

\newpage
\spacingset{1.5} % DON'T change the spacing!
\section{Background}\label{sec:intro}
The role of the human microbiome in health and disease has been intensively studied over the past few years, see, e.g., \cite{Fan:2021cf, Valdes:2018fx}, for several reviews. 
Potentially pathogenic or probiotic microorganisms can be identified by analyzing  their abundances in a microbial ecosystem (e.g., the human gut) with respect to some covariate of interest such as disease status. 
Current prevailing technologies for studying the human microbiome use metagenomic sequencing, where either the DNA of a taxonomically informative gene (e.g. 16S rRNA) or all the genomic DNA in the microbial genome is sequenced. 
After obtaining the raw sequencing reads, the reads can be clustered into  operational taxonomic units (OTUs), denoised into amplicon sequence variants (ASVs), or mapped to a microbial reference database (taxa) using existing bioinformatics pipelines such as UPARSE, DADA2, and MetaPhlAn \cite{Edgar:2013, Callahan:2016, Segata:2012}. For simplicity, we use the term taxon (pl. taxa) to represent any taxonomic unit (OTU/ASV/taxon) from a bioinformatics pipeline. 
Therefore, after bioinformatics processing, we have an abundance table recording the frequencies of detected taxa in the samples, together with a meta data table capturing the sample-level information. Differential abundance analysis is then carried out based on the abundance and meta data table.

%differential abundance analysis are based on operational taxonomic units (OTUs), which are microbial genomic sequences clustered by sequence similarity. OTUs are then typically mapped to a taxonomic reference database to link the sequences to the microbial species \citep{Wang:2007}. For this reason, we often use the term taxon to refer to the target of differential abundance analysis. The OTU counts are typically organized into large scale matrices with rows representing taxa, columns representing samples, and the entries denoting the sequences' counts.

Ideally, we want to measure the absolute abundance of the microorganisms, i.e., the number of microorganisms per unit area/volume at the microbial ecosystem, and differential abundance analysis is  performed on the absolute abundance data. 
However, the data  from a sequencing experiment only captures  the relative abundance (compositional) information since the total sequence read count, also known as sequencing depth or library size,  does not reflect the  total microbial load in the specimen due to the complex chemistry involved in sequencing \cite{Gloor:2017,Tsilimigras:2016}.  Although there are several experimental techniques such as qPCR, spike-in and flow cytometry that can be used to achieve absolute abundance measurement, the severe limitations of these techniques prevent their wide adoption \cite{Morton:2019}.  Therefore, the prevailing sequencing protocol is still only able to measure the relative abundances. Drawing inferences about the changes on the unknown absolute abundance based on the measured relative abundance data is challenging due to missing the total microbial load information. The increase or decrease in the abundance of some taxa with respect to a covariate of interest  automatically results in changes in the relative abundances of all other taxa, a statistical phenomenon known as compositional effects. Therefore, using the standard statistical techniques such as two-sample t-test, Wilcoxon rank sum test, and linear regression analysis ignoring the compositional nature of the data could lead to a large number of false discoveries. 
%We consider an artificial example for illustration. Suppose we have two samples with three detected taxa. The absolute taxa abundances for the two samples are $(10, 20, 70)$ and $(30,20,70)$. Thus, only the first taxon is differentially abundant.  Now suppose, after sequencing (ignoring the sampling variability),  the read counts for the two samples are $(100, 200, 700)$  and $(3, 2, 7)$, where the first sample is more deeply sequenced.  Since the total read sum is an experimental artifact, we  normalize the data into relative abundances by dividing by the library size, and the corresponding relative abundances for the two samples become $(0.1,0.2,0.7)$ and $(0.25,0.167,0.583)$.   Hence, all three taxa appear differentially abundant while the  truth is that only the first is differential.  

%Based on the relative abundance data alone, it is impossible to tell whether it is the first taxon that is differential or all the taxa are differential for the previous example.
 For the differential abundance problem to be well defined, one has to make assumptions. One major assumption is that the differential signal is sparse, i.e., only a small proportion of taxa are associated with the covariate of interest. Although many studies have supported the sparse signal assumption, there are also studies support dense signal hypotheses, where a large number of taxa are differential with small effect sizes \cite{Xiao:20181,Xiao:20182}.  Therefore, the validity of a method and the definition of true or false positive depends on the specific assumption one is willing to accept. Here 
%we do not claim that our model is ``correct": 
 our goal is to provide a statistical tool that could be potentially useful for pinpointing top candidate taxa for further biological validation.

%The abundance of each OTU in the specimen  (e.g., a fecal sample) does not reflect the microbial ecosystem's absolute taxa abundances from where the specimen was derived. It is more reasonable to draw inferences regarding the relative abundance of a taxon in the ecosystem using its relative abundance in the specimen (the counts divided by the total number of counts or library size in each sample). However, the increase (or decrease) of some taxa relative abundances in response to a physiological change automatically results in changes in the relative abundances of all other taxa. Therefore, using the standard statistical methods such as Pearson correlation coefficient,  ANOVA, Kruskal--Wallis test, linear regression analysis, and so on and ignoring the compositionality might cause false discoveries. We consider an example for demonstration. Suppose we have two groups of samples. The absolute taxa abundances of the samples are equal to $(10, 20, 70)$ in one group and equal $(30,20,70)$ in the other group, namely, only the first taxon is differentially abundant in the ecosystems of the two groups. The corresponding relative abundances are equal to $(0.1,0.2,0.7)$ and $(0.25,0.167,0.583)$ obtained by dividing the counts by the respective library sizes. For a large enough sample size, all of the taxa will be tested as differentially abundant according to the relative abundances.

To address compositional effects in differential analysis, one popular approach is robust normalization.  It involves calculating a normalizing factor (scale factor), which is robust to a small number of differential taxa and could well capture the sequencing effort for the non-differential part.  Therefore, dividing by such a normalizing factor will bring the abundance of the non-differential taxa to the same scale while retaining the differences  for those differential ones.
Assuming the number of differential taxa is small, different strategies have been used to calculate a robust normalizing factor including TMM, RLE, CSS, and GMPR \cite{Robinson:2010,Anders:2010,Paulson:2013,Chen:2018}.  We list these methods in Additional file \ref{supp1}: Table S1. In contrast,  the naive total sum scaling (TSS) normalization, which divides the counts by the library size, is not a robust normalization method \cite{Chen:2018}.

%It transforms the data by eliminating the effects of specific artifactual influences so that the normalized values across different datasets are comparable. 
%The use of relative abundances is a common normalization technique, called total-sum scaling, eliminating the effect of varying library sizes across samples. However, as we have demonstrated, total-sum scaling does not correctly handle the compositional effect in differential abundance analysis. An oracle normalized values of the data in the previous example are $(0.1/0.9,0.2/0.9,0.7/0.9)$ and $(0.25/0.75,0.167/0.75,0.583/0.75)$, where we divide the relative abundances by a scale factor, the sum of the proportions of the non-differentially abundant taxa. In other words, we multiply the library size by a scale factor to obtain an effective library size. In reality, however, we would not know which taxa are unchanged and hence to be used as references. One idea is to exploit the relatively invariant parts of the count distribution, for which many methods have been proposed in the literature. We list some of these methods in Table \ref{table-norms}.

%To some extent, these methods are all designed to find a normalization constant within a sample (effective library size) that are robust to a small subset of differential taxa and could eliminate the bias induced by the compositionality.
These normalization techniques can be combined with different statistical procedures in differential abundance analysis. For example, we can divide the counts by the normalizing factor from the normalization techniques in Additional file \ref{supp1}: Table S1 and then apply standard statistical tools based on the normalized data. 
%Normalization is also useful in estimating the mean parameter of a negative binomial distribution, which can model the count distribution.
The normalizing factor could also be included as an offset in regression models such as edgeR \cite{Robinson:2010-1}, DESeq2 \cite{Love:2014}, MicrobiomeDDA \cite{Chenj:2018}, and metagenomeSeq \cite{Paulson:2013}, where the TMM, RLE, GMPR, and CSS normalization are the accompanying normalization methods.  The recently developed MaAsLin2 \cite{Mallick:2020} uses log linear models on the normalized abundance data. Different normalization approaches including TSS, TMM, CSS, and CLR are options in MaAsLin2. 
 A variant to the robust normalization approach is to find a reference taxon or a set of reference taxa, which are assumed to be non-differential with respect to the covariate of interest. The data are then normalized by the count of the reference taxon (or the sum of the counts of the reference taxa).  This strategy was used in RAIDA \cite{Sohn:2015} and DACOMP \cite{Brill:2020}.   

Another line of methods to tackle the compositional effect uses (log) ratio approach since only ratios are well defined for compositional data \cite{Aitchison:1986}. The ALDEx2 method by \cite{Fernandes:2014} uses the centered log-ratio (CLR) transformation, where the counts of a sample are divided by their geometric mean before taking logarithms. Differential abundance analysis is then performed using Wilcoxon rank sum test or t-test based on the CLR transformed data.  In the CLR approach, the geometric mean can also be regarded as a robust normalizing factor. The ANCOM proposed by \cite{Mandal:2015} computes the pairwise ratios of the relative abundances and identifies the taxa with the most differential ratios. This is based on the observation that the abundance ratios for those differential taxa to other taxa are all differential assuming distinct effect sizes while the ratios for those non-differential taxa are mostly non-differential.  Therefore, by analyzing the pattern of the pairwise ratios, one could distinguish the differential taxa from a background of non-differential taxa with high accuracy. Recently,  a bias-corrected version of ANCOM (called ANCOM-BC) has been proposed \cite{Lin:2020}, which uses a linear regression framework based on log-transformed taxa count and estimates the unknown bias term due to the compositional effect through an EM algorithm.

In the work of \cite{Weiss:2017, Hawinkel:2019}, the authors evaluated several popular methods in differential abundance analysis (ANCOM-BC/MaAsLin2 not included)  and showed that the inflation of the false discovery rate (FDR) is still a ubiquitous problem, and no method is satisfactory in all aspects.   A method that is computationally efficient, relatively robust and powerful, and flexible enough to allow covariate adjustment and application to correlated microbiome data is still lacking in the field.  In this paper, we propose a linear regression framework for differential abundance analysis (LinDA) to fill the methodological gap. LinDA involves three simple steps that can be carried out efficiently. First, it runs linear regressions using the CLR transformed abundance data as the response. Then it identifies a bias term due to the compositional effect and corrects for the bias using the mode of the regression coefficients across different taxa. Finally, it computes the p-values based on the bias-corrected regression coefficients and applies the Benjamini-Hochberg (BH) procedure to control the FDR. 
We rigorously prove the asymptotic FDR control of the proposed method, making it the first procedure that enjoys a theoretical FDR control guarantee. Our approach is related to ANCOM-BC but differs in several aspects. (i) Our derivation provides a clear interpretation of the bias term and suggests a simple way to correct it. (ii) Our procedure does not involve the EM-algorithm and can be 100--1000 times faster than ANCOM-BC in our numerical studies. (iii) Our method can be directly extended to the mixed-effect models. Longitudinal and repeated measurement-based microbiome studies have been increasingly common \cite{Faust:2015, Lewis:2015} but statistical tools for correlated microbiome data analysis remain scarce. LinDA can analyze the correlated microbiome data using the classic linear mixed effects models. Through extensive simulation studies and real data analyses, we show that the new method outperforms the state-of-the-art approaches in terms of FDR control and power.

\section{Results}
\subsection{Numerical studies}

{\bf Setups}
We conducted comprehensive simulations to evaluate the performance of the proposed method under different setups. We set $m=500$ as the baseline for the number of taxa, which is similar to the number of tests at the species level for a typical microbiome study. We investigated the sample size $n=50,\;200$ representing small and large sample sizes, respectively. More combinations of $m$ and $n$ were  studied in additional settings. We generated the absolute abundances from the log normal distribution and considered three cases for the covariate of interest and confounders: the covariate of interest follows the Bernoulli distribution and no confounder (denoted as C0),  the covariate of interest follows the standard normal distribution and no confounder (C1), and the covariate of interest follows the Bernoulli distribution and two confounders (C2). In addition to the basic setting (log normal abundance distribution, denoted as S0), we investigated other settings to study the robustness of the proposed method including zero-inflated absolute abundances (S1),  correlated absolute abundances (S2),  gamma abundance distribution (S3), smaller $m$ (S4), smaller $n$ (S5), 10-fold difference in library size (S6),  negative binomial abundance distribution (S7) and correlated microbiome data generated by mixed-effect model (S8). See section \ref{sec-detailed-setups} for more details.

\vskip0.5cm
{\noindent\bf Competing methods}
We compared our method with ANCOM-BC, ALDEx2, DESeq2, edgeR, metagenomeSeq and MaAsLin2. For DESeq2 and edgeR, we replaced  their native normalization methods  with GMPR normalization, which was shown to improve the power and false positive control in differential abundance analysis \cite{Chen:2018}. For metagenomeSeq, there are two implementations, \texttt{fitZig} and \texttt{fitFeatureModel},  in the R Bioconductor package \texttt{metagenomeSeq}.  Currently,  \texttt{fitFeatureModel} is only applicable to binary covariate case (C0).  We use metagenomeSeq2 and metagenomeSeq to denote  the \texttt{fitFeatureModel} and \texttt{fitZig} procedures, respectively. We also compared with the standard non-parametric methods: Wilcoxon rank sum test for case C0 (a binary covariate) and Spearman correlation test for case C1 (a continuous covariate), both with the GMPR normalized data.
% We applied the Wilcoxon rank sum test using only the variable $u_s$ for case C2 for comparison. % 
%and Wilcoxon to represent both the Wilcoxon rank sum test (C0 and C2) and Spearman correlation test (C1).

For the proposed method, we considered two zero-handling approaches. The first approach adds a pseudo count of 0.5 to all the counts, which is widely used in microbiome data analysis on the log scale.  However, it has been shown to be problematic under certain situations \cite{Brill:2020}. We thus designed a new imputation-based approach, where the zeros were imputed by $N_s/(\max_{k:Y_{ik}=0}N_k)$ for $i$th taxon in $s$th sample, where $N_s$ denotes the library size (sequencing depth) of $s$th sample and $Y_{ik}$ denotes  the read count of  $i$th taxon in  $k$th sample.  In other words, zeros were imputed differently according to the library size of the sample, and zeros in the sample with a larger library size were replaced with larger fractions. In this imputation approach, we treat zeros as left-censored missing data. Suppose we only know the library sizes, then a natural strategy is to impute zeros in proportion to the library size with the  sample of the largest library size receiving a fractional count close to 1 (in our approach, we simply set it as 1). The purpose of the imputation strategy is to reduce false positives when the library size is correlated with the covariate of interest. As shown in the simulation studies, the pseudo-count approach worked sufficiently well  in most settings except the setting S6, where the library size between the groups differed by 10 folds. In contrast, the imputation approach reduced the false positive rate extensively for the setting S6  (Additional file \ref{supp2}:  Fig. S1).  However, it was slightly less powerful than the pseudo-count approach when the library size was a not confounder (Additional file \ref{supp2}:  Fig. S2). Thus, in the implementation, we used an adaptive approach: we first tested the association between the covariate of interest and the library size. If the p-value was smaller than 0.1, we used the imputation approach conservatively; otherwise, we used the pseudo-count approach.  Additional file \ref{supp2}: Fig. S1 and S2 show that the adaptive method controls the false positives when the library sizes are very different among groups while retaining the power when the library sizes are similar.

The proposed LinDA method can be viewed as a three-step procedure: CLR-OLS-BC (OLS stands for ordinary least squares and BC stands for bias correction), which can be easily extended to the linear mixed-effects model using CLR-LMM-BC (LMM stands for linear mixed-effect model). In the setting S8 (correlated microbiome data), we compared  CLR-LMM-BC to CLR-OLS-BC, CLR-OLS, and CLR-LMM to demonstrate the utility of LinDA for correlated microbiome data analysis.

\vskip0.5cm
{\noindent \bf Results}
We use S0C0 (log normal abundance distribution, a binary covariate) to denote the setting S0 (log normal abundance distribution) with the covariate design C0 (a binary covariate) and likewise for other setups. For S0, we studied all the three covariate designs (C0--C2), and for S1--S8, we only performed C0 for demonstration.
We found that DESeq2, edgeR and metagenomeSeq had severe FDR inflation under most settings.  To increase the readability of the results (presented in figures), we did not include them in the main comparison and focused on the comparison between LinDA, ANCOM-BC, ALDEx2, metagenomeSeq2, MaAsLin2 and Wilcoxon (Fig. \ref{fig-S0C0}--\ref{fig-S3C0} and Additional file \ref{supp2}: Fig. S1--S15). Full results of all methods are presented in Additional file \ref{supp2}: Fig.  S20--S30.

We first point out that MaAsLin2  with CLR normalization is essentially the same as the CLR-OLS procedure we described earlier. Additional file \ref{supp2}:  Fig. S3 compares LinDA, CLR-OLS, MaAsLin2-TSS, MaAsLin2-TMM, MaAsLin2-CSS and MaAsLin2-CLR under the setting S0C0 (log normal abundance distribution, a binary covariate). We can see that MaAsLin2-CLR is close to CLR-OLS,  both of which suffer from FDR inflation. We included MaAsLin2 with its default configuration (i.e., TSS normalization) in our comparisons below.

Fig. \ref{fig-S0C0} and Additional file \ref{supp2}: Fig. S4 and S5 show the results of the competing methods under the log-normal abundance distribution with three covariate designs: a binary covariate (S0C0), a continuous covariate (S0C1), and a binary variable of  interest with confounders (S0C2), respectively.  Generally speaking, LinDA and ANCOM-BC have the best FDR and power trade-off. Under C0 and C2 (Fig. \ref{fig-S0C0} and Additional file \ref{supp2}: Fig. S5), both methods control the FDR around the target level, and ANCOM-BC is slightly more powerful than LinDA when the sample size is small. However, under C1 (Additional file \ref{supp2}: Fig. S4), LinDA controls FDR at the target level at both sample sizes while ANCOM-BC has slight FDR inflation when the sample size is small.  LinDA is also slightly more powerful than ANCOM-BC at a small sample size. The Wilcoxon rank sum test based on GMPR normalized data and MaAsLin2 perform well under C0 (a binary covariate, Fig. \ref{fig-S0C0}) with  slightly inflated FDR at larger effect sizes and reasonable power across settings.  In contrast, for a continuous covariate (C1, Additional file \ref{supp2}: Fig. S4),   the Spearman rank correlation test and MaAsLin2 have large FDR inflation when the signal is dense. When there are confounders (C2, Additional file \ref{supp2}: Fig. S5), Wilcoxon has severe FDR inflation when the sample size is large due to its inability to adjust for confounders, while MaAsLin2 provides acceptable results as under C0. ALDEx2 is a conservative method, which offers the strongest FDR control but is much less powerful. metagenomeSeq2 performs well when the signal is sparse but fails to control the FDR when the signal is dense. We also studied the effect of zero inflation and the correlations among taxa (S1C0 and S2C0, Additional file \ref{supp2}: Fig. S6 and S7),  where we observed similar patterns such that LinDA and ANCOM-BC had overall the best performance among the compared methods.

Since LinDA assumes a log normal distribution of the absolute abundance, it is interesting to evaluate its performance when the log normal assumption is violated. We thus simulated the absolute abundance data using a gamma distribution (S3C0), and the results are depicted in Fig. \ref{fig-S3C0}. 
It shows that LinDA controls the FDR close to the target level and has the highest power. When the signal is dense (20\%), ANCOM-BC has a noticeable FDR inflation while ALDEx2, metagenomeSeq2, MaAsLin2 and Wilcoxon have severe FDR inflation when the signal is dense.  
%metagenomeSeq is also less powerful in the sparse signal setting, compared to its performance in the log normal model??.

With a smaller number of taxa  ($m = 50$, S4C0, Additional file \ref{supp2}: Fig. S8), ANCOM-BC controls the FDR and the power is also high. LinDA is the most powerful but it has slight FDR inflation. metagenomeSeq2, MaAsLin2 and Wilcoxon control the FDR but are less powerful in the case of sparse signal. However, when the signal is dense, they could not control the FDR properly. When the sample size is very small ($n=20$ or 30, S5C0), LinDA controls the FDR around the target level and maintains high power (Additional file \ref{supp2}: Fig. S9).  ANCOM-BC and metagenomeSeq2 have large FDR inflation and the inflation seems to increase as the sample size gets smaller. 
MaAsLin2 and Wilcoxon are much less powerful and ALDEx2 has virtually no power. Under the setting S6C0, where the sequencing depth differs by 10 folds, only ALDEx2 and our proposed method with adaptive zero-handling approach are able to control the FDR (Additional file \ref{supp2}: Fig. S10). LinDA achieves a better performance in both the FDR control and power than ALDEx2. It is interesting that ALDEx2 performs better under S6C0 than under other settings. We point out here that when we implemented ANCOM-BC, we disabled its zero treatment. To further investigate whether its zero treatment option improves its performance, we also run the procedure enabling its zero treatment (zero\_cut = 0.9, lib\_cut = 1000, struc\_zero = TRUE), and found the results were very similar (S6C0, Additional file \ref{supp2}: Fig. S11).

Under the previous simulation settings, we found that DESeq2 and edgeR had the worst false positive control (Additional file \ref{supp2}: Fig.  S20--S28). As the two methods assume negative binomial distribution for the counts, it is interesting to see their performance when the data are generated by their assumed model (S7C0). Additional file \ref{supp2}: Fig. S29 shows that DESeq2 and edgeR (and metagenomeSeq) remain to have serious FDR inflation, indicating that the normalization approach to address the compositional effect is not sufficient. In contrast, LinDA and ANCOM-BC perform the best among competitors as in other settings, and ANCOM-BC achieves higher power than LinDA (Additional file \ref{supp2}: Fig. S12).

Finally, we applied LinDA to correlated microbiome data (S8C0), where the other competing methods except MaAsLin2 are not applicable to correlated samples. Additional file \ref{supp2}: Fig. S13 and S14 compare the methods CLR-LMM-BC (LinDA-LMM), CLR-OLS-BC (LinDA-OLS), CLR-LMM, CLR-OLS and MaAsLin2 for correlated data. In the scenario of comparing the pre-treatment and post-treatment samples (S8.1, Additional file \ref{supp2}: Fig. S13), we could clearly see that ignoring the bias  tremendously increases the FDR level especially under dense signals (LinDA-LMM vs CLR-LMM). In addition, LinDA-LMM is more powerful than LinDA-OLS due to its ability to exploit the correlation between pre- and post-treatment samples. 
Under the replicate sampling setting (S8.2, Additional file \ref{supp2}: Fig. S14), we see that the LinDA-OLS has significant FDR inflation by treating the replicate samples as independent ones. In contrast, LinDA-LMM controls the FDR at the target level. MaAsLin2 control the FDR under both settings but is less powerful than LinDA-LMM.

Based on the presented simulation settings, we summarize that LinDA and ANCOM-BC have overall the most robust performance among the methods evaluated. However, ANCOM-BC is computationally  intensive. As shown in Table \ref{tab-speed}, LinDA could be 100--1000 times faster than ANCOM-BC, making LinDA a highly scalable method in practice. In addition, the extension of LinDA to the mixed-effect models is easily carried out and performs well.

\subsection{Real data applications}\label{sec-real}
{\bf Datasets}
We applied LinDA and the competing methods to three real datasets with independent samples from the studies of C. difficile infection \cite[CDI,][]{Schubert:2014}, inflammatory bowel disease \cite[IBD,][]{Morgan:2012}, and rheumatoid arthritis \cite[RA,][]{Scher:2013}. 
To demonstrate the use of LinDA on correlated microbiome samples, we applied LinDA to a dataset from the study of the smoking effect on the human upper respiratory tract \cite[SMOKE,][]{Charlson:2010}. We used the microbiome samples from the throat for illustration, where each subject has two samples from the left and right sides of the throat. The CDI and RA datasets were provided by the authors while the IBD and the SMOKE datasets were downloaded from the Qiita database \cite{Gonzalez:2018} with the study ID 1460 and 524. All the datasets have binary phenotypes. Antibiotics use is the confounder for the IBD dataset  ($p=0.03$ and OR = 0) while sex is the confounder for the SMOKE dataset ($p = 0.02$ and OR= 2.26). They will be adjusted in methods that are capable of covariate adjustment. We excluded samples with less than 1000 read counts and taxa which appear in less than 10\% of the samples. The basic characteristics for the four filtered datasets are summarized in Table \ref{tab-real}. 
%We did the following analyses: (i) Behavior under the complete null. We first shuffle the variable $u$ and then do the differential abundance analysis. (ii) Behavior under the complete null with different sequencing depth. After shuffling the variable $u$, we order the sequencing depths within each group, and use the first half for one group and the second half for the other group. For continuous $u$, we take similar steps by regarding the smaller half $u$ as one group and larger half as the other group. 
We compared the detection power as well as their overlap patterns for LinDA, ANCOM-BC, ALDEx2, metagenomeSeq2, MaAsLin2 and Wilcoxon. Specifically, we compared the number of discoveries at different FDR levels (0.01--0.25) and  used UpSet plot \cite{Lex:2014} to show the overlap at the target FDR of 0.1. We used winsorization at quantile 0.97 to reduce the impact of potential outliers as recommended in \cite{Chenj:2018}. 
%In addition, we also studied the sensitivity of the p-values to the number of taxa included by comparing the results of data that exclude the taxa with prevalence less than 10\% and 20\%. The numbers of taxa included are summarized in Table \ref{tab-real}. We expect p-values to be robust to the taxa included.

\vskip0.5cm
{\noindent\bf Results}
%Tables \ref{tab-real-shuffle1} and \ref{tab-real-shuffle2} in the supplementary material summarize the results of analyses (i) Behavior under the complete null and (ii) Behavior under the complete null with different sequencing depth. From Table \ref{tab-real-shuffle1}, we see that ALDEx2 is extremely conservative and edgeR and metagenomeSeq-2 suffer severe FDR inflation, which are consistent with the results from the previous numerical studies. 
For the CDI dataset, LinDA or MaAsLin2 made the most discoveries at different FDR levels (Fig. \ref{fig-real-curve}). At 10\% FDR, LinDA discovered eight and MaAsLin2 discovered six taxa associated with CDI. In contrast, ANCOM-BC, ALDEx2, and Wilcoxon discovered three while metagenomeSeq2 discovered two. 
As discussed in \cite{Schubert:2014}, subjects with CDI were more likely to  have the bacterial family Lachnospiraceae and Erysipelotrichaceae. LinDA found one more taxon belonging to Lachnospiraceae than other methods (blue bars in Fig. \ref{fig-real-venn}). Besides, LinDA, MaAsLin2 and Wilcoxon found one differential taxon belonging to Erysipelotrichaceae while the other three methods did not identify any (orange bar in Fig. \ref{fig-real-venn}). For the IBD dataset, LinDA detected a similar number of taxa as ANCOM-BC and MaAsLin2. Wilcoxon rank sum test detected a large number of taxa associated with the disease status, but this could be due to the confounding effects of antibiotics use since it could not adjust for covariates.  From Fig. \ref{fig-real-venn}, we observe that most discoveries by LinDA are shared by ANCOM-BC, MaAsLin2 or Wilcoxon. For the RA dataset, LinDA detected a similar number of taxa as ANCOM-BC and more taxa than other methods. The differential taxa detected by LinDA and ANCOM-BC largely overlapped. Overall, the results are consistent with the simulation studies.

Finally, we applied LinDA-LMM and MaAsLin2 to the SMOKE dataset, where each subject has two replicate samples from the throat.  The aim is to identify smoking-associated taxa adjusting for the sex.  To account for the correlation between the replicate samples, we included a subject-level random intercept in LinDA-LMM.  As a comparison, we also applied  LinDA-OLS and MaAsLin2 to the right and left throat samples separately. LinDA-OLS based on the left or right throat samples alone discovered 12 and 15 differential taxa at 10\% FDR. When both left and right samples were used in LinDA-LMM, 21 differential taxa were identified, covering the majority of the taxa identified based on the left or right throat samples alone (Fig. \ref{fig-real-venn}). In addition, LinDA-LMM detected five taxa that were missed by analyzing the left or right samples separately.  Compared to MaAsLin2, LinDA-LMM discovered seven more differential taxa.
Therefore, LinDA-LMM provides a convenient way to analyze correlated microbiome datasets and enjoys the power improvement by analyzing all samples together. 
%They show some FDR inflation on the datasets COMBO and IBD, but the number of false rejections are small, and control the FDR on the dataset RA as well as the metagenomeSeq. DESeq2 has larger FDR inflation but also has small number of false rejections. Wilcoxon controls the FDR on all the three datasets. When there is significant correlation between the sequencing depth and the variable of interest (Table \ref{tab-real-shuffle2}), all the methods have larger FDR, and ANCOM-BC has slightly more fasle rejections compared to LinDA, DESeq-2 and Wilcoxon.

To visualize the results, \texttt{LinDA} provides a function to generate the effect size plot  and volcano plot for differential taxa. Additional file \ref{supp2}: Fig. S16--S19 display the effect size plots of differential taxa at FDR level of 0.1 and volcano plots for the four datasets, respectively.  
In the effect size plots (Additional file \ref{supp2}: Fig. S16A--S19A), the taxa in black are detected by LinDA and taxa in red are detected solely by LinDA.
%, and taxa in blue are missed by LinDA but detected by ANCOM-BC and one more method. 
%In Supplementary Figures \ref{fig-real-IBD}A, the taxa in blue are detected only by ANCOM-BC and Wilcoxon. In Supplementary Figures \ref{fig-real-RA}A, the taxa in blue are detected only by ANCOM-BC and metagenomeSeq.
In Additional file \ref{supp2}: Fig. S16A (CDI), the taxa in blue are missed by LinDA but detected by one or more other methods. 
In Additional file \ref{supp2}: Fig. S17A (IBD) and S18A (RA), the taxa in blue are missed by LinDA but detected by two or more  other methods.
No taxa are detected by MaAsLin2 but missed by LinDA-LMM for the SMOKE dataset. Based on the effect size plots for the CDI, IBD and RA datasets, we can see  that, for the taxa solely detected by LinDA, the effect sizes tend to be underestimated without bias correction and bias correction improves the power in these cases.  On the contrary, for the taxa missed by LinDA, the effect sizes tend to be overestimated without bias correction. In addition, we observe that the differential taxa for the IBD, CDI and RA datasets are more unbalanced (i.e., more negative or positive associations)  while the differential taxa for the SMOKE dataset are relatively balanced (i.e., similar numbers of negative and positive associations). %compared to the diarrheal control group, the C. difficile infection group has three less abundant taxa and six more abundant taxa. On the other hand, comparing the Crohn's disease group with the healthy group, and NORA group with the healthy group, we find that most differential taxa are less abundant in the disease groups than the healthy groups. For SMOKE dataset, around half of differential taxa are less abundant and half are more abundant in the smoker group compared to the non-smoker group. Therefore, we expect IBD, CDI, and RA datasets to have stronger compositional effects than the SMOKE dataset since the changes are more unbalanced. 
Indeed, the effect size plots, where we plot both the debiased and un-debiased coefficients, revealed larger biases for the IBD, CDI, and RA datasets.

\section{Discussion}\label{sec-dis}

Differential abundance analysis is at the core of the statistical analysis of microbiome data. Microbiome data are compositional in nature and all we know are the relative abundances, making the identification of  differentially abundant taxa at the ecological site  particularly challenging \cite{Gloor:2017,Tsilimigras:2016}. Numerous differential abundance analysis methods have been proposed focusing on addressing the compositional effects \cite{Robinson:2010-1,Love:2014,Chenj:2018,Paulson:2013,Sohn:2015,Brill:2020,Fernandes:2014,Mandal:2015,Lin:2020}.  Among all the competing methods, ANCOM-BC is the state-of-the-art method, it has been demonstrated to be more robust and powerful than the competing methods.  However, there are two drawbacks of ANCOM-BC. First, it is computationally intensive for large-scale microbiome datasets such as the AmericanGut dataset.  Due to the huge inter-subject variation, large-scale microbiome studies have been increasingly popular, resulting in larger sample sizes.  On the other hand, metagenomic sequencing has become increasingly deeper to have a high-resolution view of the microbiome, leading to an unprecedented number of new microbial  features.  To meet the analysis need for such large-scale datasets,  a computationally efficient tool is much needed. Secondly, ANCOM-BC is not applicable to correlated/clustered microbiome datasets such as those from family/longitudinal microbiome studies or studies with paired and repeated measurements \cite{Faust:2015,Lewis:2015}. Longitudinal microbiome analysis, which enables the study of the trajectory of the microbiome as well as controls for potential confounders, has been increasingly employed in human microbiome studies.  Unfortunately, statistical tools for longitudinal microbiome studies are scarce.  In contrast, LinDA is computationally efficient since it only involves fitting regular linear regression models and could be easily scaled to hundreds of thousands of taxa.  Moreover, the extension of LinDA to linear mixed effects models (LMM) is straightforward and we have highly efficient tools such as the R \texttt{lme4} package \cite{Bates:2015} for fitting LMM.  Therefore, differential abundance analysis of correlated/clustered microbiome datasets could be easily performed using LinDA. Our framework also gives more insights into the CLR-based approach, which has been widely used in compositional data analysis \cite{Aitchison:1986}.  However, the bias of CLR regression models has not been formally recognized to our best knowledge.  Our framework justifies the use of CLR regression and provides a solution to correct the bias associated with CLR regression.

In the simulation, we found that Wilcoxon rank sum test and MaAsLin2 showed similar FDR/power curves and performed fairly well in most settings. As we mentioned earlier, MaAsLin2 is based on log linear models on the normalized count data, and it is essentially a two-sample t-test when no confounders are included, which explains why MaAsLin2 is close to Wilcoxon rank sum test.  However, %It also did not perform well when the abundance data followed a gamma distribution (FDR inflation) or the sample size was small (very low power).  
when we simulated an even stronger compositional effect by drawing the differential taxa from the top 25\% most abundant taxa, we found that Wilcoxon rank sum test and MaAsLin2 began to break down (Additional file \ref{supp2}: Fig. S15). ANCOM-BC was overall robust and powerful, but it had inflated type I error at small sample sizes.  metagenomeSeq2 did not perform well when the signal was dense and was generally less powerful than ANCOM-BC and LinDA.  ALDEx2 was the most conservative method: its strong FDR control  was at the expense of statistical power.  LinDA was as competitive as ANCOM-BC in most settings. It had better FDR control than ANCOM-BC when the sample size was small or the covariate of interest was continuous. However, LinDA had some FDR inflation when the number of taxa was small. Under a very strong compositional effect (Additional file \ref{supp2}: Fig. S15), LinDA  also showed some FDR inflation but overall it had the best  FDR and power trade-off.
%ANCOM overall good. slightly more powerful n = 50. n>50 for ANCOM. But continuous covariate FDR inflation and less powerful   LinDA slightly more powerful. Wilcoxon good for log normal but bad on gamma dense signal. Wilcoxon good sparse signal, slight inflation dense signal, bad in gamma.  Very less powerful at small sample size. sequence depth variation only our method works.   In confounding library size, rarefaction is still needed but with a power loss. 

When the library size was associated with the covariate of interest, all existing methods had severe type I error inflation.  Fortunately, such  association is detectable and if we see a significant association,  rarefaction should be used for those methods. Although rarefaction controls the effect of uneven library sizes, it discards a significant portion of the reads and thus loses much information in the data.   When there are many samples with  small library sizes, the users have to decide whether to retain more reads or more samples. In LinDA, we implemented a heuristic imputation method, where the imputed values are  proportional to the library sizes. This procedure makes the imputed data after CLR transformation independent of the library size and substantially reduces the inflated type I error due to library size confounding.

Although the presented simulation settings could give basic insights into the performance of various methods, such model-based simulations might not be able to capture the full characteristics of the real microbiome data. It is very likely that the performance of the compared methods will change using a different simulation framework.  Moreover, our simulation strategy purposely creates strong compositional effects, where  all differential taxa show the same direction of change.  Such setting is used to test the limit of the various methods in addressing the compositional effects. However,  in real data, the compositional effects may not be always strong, and the FDR inflation of many methods could be very moderate.  Therefore,  a future benchmarking study, which uses real data-based simulation strategy and investigates all biologically plausible differential settings,  is much needed to have a comprehensive and objective evaluation of existing differential abundance analysis methods.

%LinDA is based on the log linear model, where the coefficients can be interpreted as the log fold change in response to the one unit change of the covariate.  In the analysis of biological data, interpretation is one key factor in selecting relevant tools.  

As for all model-based approaches, LinDA has several assumptions and limitations.  First, LinDA relies on the assumption that there is a mode at 0 for the regression coefficients (Condition (vi) in Theorem \ref{the-fdr}).  This assumption is easy to be met if the signal is sparse. In the simulation, we show that when the signal density is around 20\%, LinDA is still very robust. However, when the signal is extremely dense, LinDA could fail. Second, LinDA assumes a log linear model on the absolute abundance.  Although this is a reasonable assumption, which has been widely adopted in the analysis of abundance data, the interaction between the host and the microbiome could be more complex than the simple log linear relationship.  Analysis of the residuals from the CLR regression could provide evidence about whether the assumption is reasonable.  If the model assumption is violated, a permutation test or transformation of the variables may be performed. Finally, although LinDA provides asymptotic FDR control, its finite-sample FDR control is not guaranteed.   Based on numerical simulations, we found that LinDA  performed well under small sample sizes. However, we did observe some FDR inflation under a small feature size due to inefficiency in mode estimation with few features.   Therefore, we do not recommend applying LinDA to datasets with small feature sizes (e.g., m $<$ 50) such as phylum-level abundance data. 

LinDA uses the relative abundance data and does not model the sampling variability of the read counts. This could reduce the statistical power  for those less abundant taxa, whose sampling variability is larger than those abundant taxa.  To remedy the power loss, another multinomial sampling layer could be imposed on top of LinDA.  However, the computational complexity will be increased significantly and breaking the simplicity of LinDA. Another approach is to perform posterior inference of the underlying proportions based on a Bayes approach. Once we obtain the posterior samples, LinDA can be applied to the posterior samples and results are then aggregated, similar in the spirit to the multiple imputation method \cite{Carpenter:2012}.

Besides microbiome data, LinDA could be applied to other sequencing data such as RNA-Seq data since all sequencing data are compositional in nature \cite{Quinn:2018}. Thus, LinDA could be an alternative for differential expression analysis if there are strong compositional effects, for example, when the highly abundant genes are differential with the same direction of change.

Finally, we comment that  addressing compostionality is more relevant when analyzing individual microbial features such as differential abundance analysis, since the major interest to biomedical investigators is to find those truly differential features (``driver") instead of those driven by the compositional effect (``passenger").  However, for community-wide analysis such as distance-based analysis \cite{Chenj:2012,Chenj:2021}, addressing the compositionality may not be necessary in order to control the type I error. This is because that compositional effect is only relevant under the alternative hypotheses. Considering compositionality in the community-wide analysis has also been found to have small effects on the statistical power \cite{Thorsen:2016, Weiss:2017}.
%since it does not need to distinguish whether the community-level shift is due to those ``driver" or ``passenger" features. 
Additionally, in microbiome-based predictive models \cite{Zhou:2019}, the relative abundances and/or their ratios could already be informative features for prediction  and addressing compositionality may not necessarily increase the prediction accuracy significantly.  Therefore, whether to address compositionality depends on the specific problems.

\section{Conclusions}
In summary, we  proposed LinDA for differential abundance analysis of microbiome compositional data.  LinDA identified a bias associated with traditional linear regression models based on CLR transformed abundance data and proposed a strategy to estimate and correct the bias. LinDA can  be extended to linear mixed effects model for analysis of correlated microbiome data. As a general methodology, LinDA can be applied to differential abundance analysis of other high-dimensional compositional data.
%  a log linear model for absolute abundance and applied the OLS estimation for the CLR transformed read counts. A bias term due to the CLR transformation and compositional effect is identified which can be corrected using the mode of the OLS estimates. The testing procedure based on the bias-corrected regression coefficients has been shown to asympotically control the FDR  when the error follows the normal distribution. To extend LinDA, one can consider generalized linear (mixed-effect) models for absolute abundance. Another promising extension is to use permutation-based procedure to estimate the number of false discoveries and hence to control the FDR. They are left for future research.

\section{Methods}\label{sec-met}
\subsection{Setup}
We use $C$, $C_1$, and $C_2$ to denote positive constants, which can be different from line to line. As summarized in the background, there are two ways to tackle the compositional effects in differential abundance analysis, namely normalization and log-ratio transformation. In this paper, we adopt the CLR transformation and develop a bias-correction procedure to address the compositional effects. Denote the absolute abundance and the observed read count of the $i$th taxon in the $s$th sample  by $X_{is}$ and $Y_{is}$, respectively. For the $s$th sample, the total read count of all taxa, $N_s=\sum^{m}_{i=1}Y_{is}$, is determined by the sequencing depth and DNA materials. Given $N_s$, it is natural to model the stratified count data over $m$ taxa through a multinomial distribution as
\begin{align}\label{m1}
	P(Y_{1s}=y_{1s},\dots,Y_{ms}=y_{ms})=\frac{N_s!}{\prod^{m}_{i=1}y_{is}!}
	\prod^{m}_{j=1}\left(\frac{X_{js}}{\sum_{i=1}^m X_{is}}\right)^{y_{js}}
\end{align}
Under (\ref{m1}), we have 
\begin{align}\label{m2}
	\log\left(\frac{Y_{is}}{\sum^{m}_{j=1}Y_{js}}\right)=\log\left(\frac{X_{is}}{\sum^{m}_{j=1}X_{js}}\right)+e_{is}, 
\end{align}
where $e_{is}$ denotes the estimation error,  which is expected to diminish as $N_s$ gets large.

\subsection{OLS estimation}
We consider the log linear model on the absolute abundance
\begin{align}\label{m3}
	\log\left(X_{is}\right)=u_s\alpha_i+(1,\mathbf{c}_s^\top)\boldsymbol{\beta}_i+\epsilon_{is},
\end{align}
where $\mathbf{c}_s=(c_{s1},...,c_{sd})^\top$ is the $d$-dimensional covariates to be adjusted, $u_s$ is the covariate of interest, and $\epsilon_{is}$ is the error term. Our goal is to discover taxa that are differentially abundant with respect to $u_s$. Statistically, we want to simultaneously test the following $m$ hypotheses $$H_{0,i}: \alpha_i=0 \text{ versus } H_{a,i}: \alpha_i\neq 0.$$
Set $\varepsilon_{is}=\epsilon_{is}+e_{is}$. Under (\ref{m2}) and (\ref{m3}), the CLR-transformed data satisfies the following linear model
\begin{align}
	W_{is}:&=\log\left\{\frac{Y_{is}}{(\prod^{m}_{j=1} Y_{js})^{1/m}}\right\}=\log\left(\frac{Y_{is}}{\sum^{m}_{k=1}Y_{ks}}\right)-\frac{1}{m}\sum_{j=1}^m\log\left(\frac{Y_{js}}{\sum^{m}_{k=1}Y_{ks}}\right)\nonumber\\
	&=\log(X_{is})-\frac{1}{m}\sum_{j=1}^m\log(X_{js})+e_{is}-\frac{1}{m}\sum_{j=1}^me_{js}\nonumber\\
	&=u_s\left(\alpha_i-\bar{\alpha}\right)+(1,\mathbf{c}^\top_s)\left(\boldsymbol{\beta}_i-\bar{\boldsymbol{\beta}}\right)+\varepsilon_{is}-\bar{\varepsilon}_s,\label{CLR}
\end{align}
where $\bar{\alpha}=m^{-1}\sum_{i=1}^m\alpha_i$, $\bar{\boldsymbol{\beta}}=m^{-1}\sum_{i=1}^m\boldsymbol{\beta}_i$, and $\bar{\varepsilon}_{s}=m^{-1}\sum_{i=1}^m\varepsilon_{is}$. From (\ref{CLR}), we can see that the OLS estimator for $\alpha$ based on the CLR transformed data is biased with the bias term being $\bar{\alpha}$.
Let $\bar{\alpha}_i=\alpha_i-\bar{\alpha}$, $\bar{\boldsymbol{\beta}}_i=\boldsymbol{\beta}_i-\bar{\boldsymbol{\beta}}$, $\bar{\varepsilon}_{is}=\varepsilon_{is}-\bar{\varepsilon}_s$, and $\bar{\sigma}_i^2=\text{var}(\bar{\varepsilon}_{is})$. Denote by $\tilde{\alpha}_i$, $\tilde{\boldsymbol{\beta}}_i$, and $\hat{\sigma}_i^2$ the OLS estimators of $\bar{\alpha}_i$, $\bar{\boldsymbol{\beta}}_i$, and $\bar{\sigma}_i^2$, respectively. We then have
\begin{align}\label{s1}
	&(\tilde{\alpha}_i,\tilde{\boldsymbol{\beta}}_i^\top)^\top=\left(\sum_{s=1}^n\mathbf{z}_s\mathbf{z}_s^\top\right)^{-1}\left(\sum_{s=1}^n\mathbf{z}_s W_{is}\right),\quad	\hat{\sigma}_i^2=\frac{1}{n-d-2}\sum_{s=1}^{n}\left\{W_{is}-\left(\tilde{\alpha}_i,\tilde{\boldsymbol{\beta}}_i^\top\right)\mathbf{z}_s\right\}^2,
\end{align}
where $\mathbf{z}_s=(u_s,1,\mathbf{c}_s^\top)^\top$. 
We respectively let $\text{var}_{\mathbf{z}}(\cdot)$ and  $\text{cov}_{\mathbf{z}}(\cdot,\cdot)$ denote the variance and covariance computed conditional on $\mathbf{z}_1,...,\mathbf{z}_n$. It can be shown that
\begin{align*}
	&\text{var}_{\mathbf{z}}(\tilde{\alpha}_i)=\hat{\rho} n^{-1}\bar{\sigma}_i^2=\hat{\rho} n^{-1}m^{-1}\left\{(m-2)\sigma_i^2+m^{-1}\sum_{i=1}^m\sigma_i^2\right\},\\
	&\text{cov}_{\mathbf{z}}(\tilde{\alpha}_i,\tilde{\alpha}_j)=\hat{\rho} n^{-1}m^{-1}\left\{-(\sigma_i^2+\sigma_j^2)+m^{-1}\sum_{i=1}^m\sigma_i^2\right\},\quad\text{for }i\neq j,
\end{align*}
where $\hat{\rho}$ is the $(1,1)$th element of $(n^{-1}\sum_{s=1}^n\mathbf{z}_s\mathbf{z}_s^\top)^{-1}$.

\subsection{Bias correction}
In many applications, it is reasonable to assume that there is only a small portion of differential taxa, i.e., most $\alpha_i$'s are equal to 0.
Under this assumption, as $\tilde{\alpha}_i$ is an unbiased estimator for $\bar{\alpha}_i=\alpha_i-\bar{\alpha}$, the mode of $\tilde{\alpha}_i$ is expected to be close to $-\bar{\alpha}$. This observation motivates us to estimate $-\bar{\alpha}$ by
\begin{align}\label{s2}
	\hspace*{-5pt}-\tilde{\alpha}=\frac{\widehat{\text{mode}}(\{\sqrt{n}\tilde{\alpha}_i\}^{m}_{i=1})}{\sqrt{n}},\quad  \text{where}\quad 
	\widehat{\text{mode}}(\{\sqrt{n}\tilde{\alpha}_i\}^{m}_{i=1})=	\argmax_{x\in\mathbb{R}}\frac{1}{mh}\sum_{i=1}^mK\left(\frac{x-\sqrt{n}\tilde{\alpha}_i}{h}\right).
\end{align}
In \eqref{s2}, $K$ is a non-negative even function with $\int_{-\infty}^{\infty}K(y)dy=1$, and $h$ is the bandwidth parameter. Under some regular conditions, we have $$\sqrt{n}(\tilde{\alpha}-\bar{\alpha})=o_{\mathbb{P}}(1)$$ 
as $m,n\to\infty$ (see the supplementary material for the proof). Therefore, one can estimate $\alpha_i$ by the bias-corrected estimator $\hat{\alpha}_i=\tilde{\alpha}_i+\tilde{\alpha}$. 

\subsection{Testing procedure}
To construct a statistic for testing $H_{0,i}$, we need to find a proper estimator for the variance of $\hat{\alpha}_i$. To this end, we note that
\begin{align*}
	&\text{var}_{\mathbf{z}}(\hat{\alpha}_i)=\text{var}_{\mathbf{z}}(\tilde{\alpha}_i)+\text{var}_{\mathbf{z}}(\tilde{\alpha})+2\text{cov}_{\mathbf{z}}(\tilde{\alpha}_i,\tilde{\alpha}).
\end{align*}
Since $\text{var}_{\mathbf{z}}(\tilde{\alpha}_i)$ is $\hat{\rho} \bar{\sigma}_i^2/n$, it dominates $\text{var}_{\mathbf{z}}(\tilde{\alpha})$ and $\text{cov}_{\mathbf{z}}(\tilde{\alpha}_i,\tilde{\alpha})$ as $n,m\to\infty$ under mild conditions. Thus, we estimate the variance of $\hat{\alpha}_i$ by $\hat{\rho}\hat{\sigma}_i^2/n$. As shown in the next section, the studentized statistic $T_i:=\sqrt{n}\hat{\alpha}_i/\sqrt{\hat{\rho}\hat{\sigma}_i^2}$ is asymptotically normal. However, for small sample, we found that t distribution provides a better approximation to the sampling distribution of $T_i$. We define the p-value for testing $H_{0,i}$ as 
\begin{align}\label{s3}
	p_i=2F_{n-d-2}\left(-|T_i|\right),
\end{align}
where $F_{n-d-2}(\cdot)$ denotes the cumulative distribution function of t distribution with $n-d-2$ degrees of freedom. Based on the p-values in \eqref{s3}, we can use the BH procedure to control the FDR. The above discussion leads to the following Algorithm \ref{algorithm}. 

\begin{algorithm}[ht]
	\caption{Linear models for differential abundance analysis (LinDA)}\label{algorithm}
	\vspace{1mm}
	\begin{enumerate}
		\item Step 1: Run OLS based on the CLR transformed observations and calculate $\tilde{\alpha}_i$ and $\hat{\sigma}_i^2$ as in (\ref{s1}).
		\item Step 2: Compute the bias-corrected estimates $\hat{\alpha}_i=\tilde{\alpha}_i+\tilde{\alpha}$ with $\tilde{\alpha}$ defined in (\ref{s2}).
		\item Step 3: Calculate the p-values as in (\ref{s3}) and run the BH procedure.
	\end{enumerate}
\end{algorithm}

\begin{remark}
	{\rm
		Built upon the linear regression framework, our method could be easily extended to the mixed-effect model:
		%\begin{align*}
		%\log\left(X_{is}\right)=u_s\alpha_i+(1,\mathbf{c}_s^\top)\boldsymbol{\beta}_i+\epsilon_{is} + r_{iv_s},
		%\end{align*}
		%	\begin{align*}
		%	\log\left(X_{ijk}\right)=u_{jk}\alpha_i+(1,\mathbf{c}_{jk}^\top)\boldsymbol{\beta}_i+\epsilon_{ijk} + r_{ik}
		%	\end{align*}
		\begin{align*}
			\log(X_{is})=u_s\alpha_i+(1,\mathbf{c}_s^\top)\boldsymbol{\beta}_i+\mathbf{r}_s^\top\boldsymbol{\gamma}_i+\varepsilon_{is},
		\end{align*}
		where $\boldsymbol{\gamma}_i$ is the random effect and $\mathbf{r}_s$ is the corresponding design. Mixed effects can be used to analyze correlated microbiome data from studies involving replicates or spatial sampling as well as family-based and longitudinal microbiome studies. We suggest using the R function \texttt{lmer} to estimate the parameters for the CLR-transformed data. Denote by $\tilde{\alpha}_{i,\text{lmer}}$, $\hat{\sigma}_{i,\text{lmer}}^2$, and $\text{df}_{i,\text{lmer}}$ the estimations for $\bar{\alpha}_i$, the variance of $\tilde{\alpha}_{i,\text{lmer}}$, and the degrees of freedom of $\tilde{\alpha}_{i,\text{lmer}}$ from the \texttt{lmer} function. We compute the bias-corrected estimates $\hat{\alpha}_{i,\text{lmer}}=\tilde{\alpha}_{i,\text{lmer}}+\tilde{\alpha}_{\text{lmer}}$, where $\tilde{\alpha}_{\text{lmer}}$ is obtained as the same procedure used in \eqref{s2}. Similarly, we let  $T_{i,\text{lmer}}=\hat{\alpha}_{i,\text{lmer}}/\hat{\sigma}_{i,\text{lmer}}$ and $p_{i,\text{lmer}}=2F_{\text{df}_{i,\text{lmer}}}(-|T_{i,\text{lmer}}|)$. The BH procedure on $p_{i,\text{lmer}}$ is finally used to control the FDR.
	}
\end{remark}

\begin{remark}
	{\rm
		Compared to the existing methods based on either normalization or CLR transformation, our method is computationally much more efficient and can be easily scaled to problems with tens of thousands of taxa. Table \ref{tab-speed} compares the computation time of LinDA and ANCOM-BC based on simulated datasets. We observe that our method is 100--1000 times faster than ANCOM-BC. We also tested on a massive dataset of the similar scale of the AmericanGut project \cite{McDonald:2018} ($m=5000$ and $n = 10000$). ANCOM-BC completed the analysis in 85 minutes compared to 28 seconds for our method (see the column of S0C0 in Table \ref{tab-speed}).  Large-scale microbiome studies have been increasingly common to overcome the large inter-subject variability, making our method practically useful for the analysis of big microbiome datasets. 
	}
\end{remark}

\subsection{Asymptotic FDR control}\label{sec-fdr}
Suppose the target FDR controlling level is $q$. The BH procedure is equivalent to finding the smallest $t^*$ such that $\widehat{\text{FDP}}(t^*)\le q$, where
\begin{align*}
	\widehat{\text{FDP}}(t)=\frac{2mF_{n-d-2}(-t)}{\sum_{i=1}^m\mathbb{I}\left(\sqrt{n}|\hat{\alpha}_i|/\sqrt{\hat{\rho}\hat{\sigma}_i^2}>t\right)}.
\end{align*}
Here $\mathbb{I}$ denotes the indicator function.
To show the asymptotic FDR control as $m,n\to\infty$, we take a Bayesian perspective by assuming that the parameters $\alpha_i$'s are independently generated from a common distribution. The key result is summarized in the following theorem and technical details can be found in the supplementary note.
\begin{theorem}\label{the-fdr}
	Let $\rho$ be the $(1,1)$th element of $\{\mathbb{E}(\mathbf{z}_s\mathbf{z}_s^\top)\}^{-1}$. Suppose the following conditions are satisfied:\\
	(i) $\mathbf{z}_s$'s are i.i.d.; $u_s$ and $c_{sa},a=1,...,d$, are sub-Gaussian; $\sigma_{\text{min}}\{\mathbb{E}(\mathbf{z}_s\mathbf{z}_s^\top)\}>C$, where $\sigma_{\text{min}}(\mathbf{A})$ represents the minimum eigenvalue of a matrix $\mathbf{A}$.\\
	(ii) $\sigma_i$'s are i.i.d. and $\mathbb{P}(C_1<\sigma_i<C_2)=1$.\\
	(iii) $\varepsilon_{is}/\sigma_i\sim^{\text{i.i.d.}}\mathcal{E}=^{d} N(0,1)$ for $i=1,...,m$ and $s=1,...,n$.\\
	(iv) $\alpha_i$'s are i.i.d.\\
	(v) $\mathbf{z}_s,\sigma_i,\varepsilon_{is}/\sigma_i$, and $\alpha_i$ for $i=1,...,m$ and $s=1,...,n$ are mutually independent.\\
	(vi) Denote by $f_n(\cdot;a)$ the density function of $\sqrt{n}\alpha_i+\sqrt{a}\varepsilon_{is}$ for any $a>0$.  For large enough $n$, the density $f_n(\cdot;\rho)$ has a unique mode at 0, i.e., $\arg\max_{x\in\mathbb{R}} f_n(x;\rho)=0$; for any $\epsilon>0$, there exists a $\delta>0$ such that $\min_n\inf_{|x|>\epsilon}|f_n(x;\rho)-f_n(0;\rho)|>\delta$.\\
	(vii) The Fourier transform $k(u)=\int_{-\infty}^{\infty}e^{-\imath uy}K(y)dy$ is absolutely integrable, where $\imath=\sqrt{-1}$ is the imaginary unit. \\
	(viii) $h=o(1)$ and $1/(mh^2)=o(1)$. \\
	(ix) $m=o(e^{Cn})$. \\
	(x) Let $S_{\infty,n}(t) =\mathbb{P}(|\mathcal{E}+\sqrt{n}\alpha_i/\sqrt{\rho\sigma_i^2}|>t)$. There exists $t_0$ such that for large enough $n$,
	$2F_{n-d-2}(-t_0)/S_{\infty,n}(t_0)\le q$.\\
	Let 
	\begin{align*}
		\text{FDR}_{m,n}(t)=\mathbb{E}\left\{\frac{\sum_{i:\alpha_i=0}\mathbb{I}\left(|\sqrt{n}\hat{\alpha}_i|/\sqrt{\hat{\rho}\hat{\sigma}_i^2}>t\right)}{1\vee\sum_{i=1}^m\mathbb{I}\left(|\sqrt{n}\hat{\alpha}_i|/\sqrt{\hat{\rho}\hat{\sigma}_i^2}>t\right)}\right\}.
	\end{align*}
	Under the above conditions, we have
	\begin{align*}
		\limsup_{m\to\infty,n\to\infty}\text{FDR}_{m,n}(t^*)\le q.
	\end{align*}
\end{theorem}
Conditions (i)--(v) help prove the consistency of the variance estimators and the mode of the regression coefficients. By assuming that the errors follow the normal distributions (Condition (iii)), we can integrate all the relevant covariate information in a single parameter $\hat\rho$, which facilitates the establishment of the consistency of the kernel density estimation and hence the estimator of mode. In the simulation studies, we also investigated the scenario of non-normal distribution. 
%To illustarte Condition (vi), we assume
We use an example to illustrate Condition (vi). In particular, we assume that $\sqrt{n}\alpha_i$ follows a discrete distribution with $\mathbb{P}(\sqrt{n}\alpha_i=a_{n,l})=\pi_l$ for $l=0,1$, where $a_{n,0}=0$, $a_{n,1}\neq 0$, $\pi_l>0$, and $\pi_0+\pi_1=1$. To reflect the sparsity, $\pi_0$ is set to be $0.8$. We choose $a_{n,1}=2$ and $5$ representing weak and strong signals, respectively. We consider two cases for the error variance: (i) $\sigma_i=1$; (ii) $\sigma_i\sim\text{IG}(a,b)$, i.e., $\sigma_i$ follows the inverse-gamma distribution with the shape parameter $a$ and scale parameter $b$. As seen from Fig. \ref{fig-ills_cond}, when the signal strength is weak, the mode of $\sqrt{n}\alpha_i+\sqrt{\rho}\varepsilon_{is}$ slightly deviates from 0 as the blue curve in the left panel indicates. For strong signals, the mode is exactly equal to zero. As shown in \cite{Parzen:1962}, Condition (vii) is fulfilled by many commonly used kernels such as the Gaussian kernel and the uniform kernel on $[-1, 1]$. Condition (ix) allows the number of taxa to be exponentially larger than the sample size. Condition (x) ensures the existence of a cut-off value to control the FDR at level $q$. A similar assumption was imposed in Theorem 4 of \cite{Storey:2004}. 

\subsection{Detailed setups for numerical studies}\label{sec-detailed-setups}
%We simulated two levels of signal density (i.e., percentage of differential taxa) $\gamma$ = 5\%, 20\%, roughly corresponding to sparse and dense signals. 
The differential taxa were randomly drawn from the entire set. 
In particular, let $H_i=0$ if the $i$th taxon is differentially abundant and $H_i=1$ otherwise. The underlying truth $H_i$ was generated from
$$H_i\sim^{\text{i.i.d.}}\text{Bernoulli}(\gamma).$$ 
We simulated two levels of signal density (i.e., percentage of differential taxa) $\gamma$ = 5\%, 20\%, roughly corresponding to sparse and dense signals.
We assumed that the baseline absolute abundance $X_{is}^{(0)}$ follows
\begin{align*}
\log\left(X_{is}^{(0)}\right)\sim^{\text{i.i.d.}}N\left(\beta_i^{(0)},\sigma_i^2\right),
\end{align*}
%and then apply the fold change
%\begin{align*}
%X_{is}=X_{is}^{(0)}\exp\left(u_s\alpha_i+\mathbf{c}_s^\top\boldsymbol{\beta}_i^{-(0)}\right),
%\end{align*}
and correspondingly the absolute abundance $X_{is}$ were draw based on
\begin{align*}
\log(X_{is})\sim^{\text{i.i.d.}}  N\left(\beta_i^{(0)}+u_s\alpha_i+\mathbf{c}_s^\top\boldsymbol{\beta}_i^{-(0)},\sigma_i^2\right),
\end{align*}
where $\boldsymbol{\beta}_i^{-(0)}$ represents the coefficients of the confounders, $i=1,\dots, m$. Let
\begin{align*}
\pi_{is}^{(0)} =\frac{X_{is}^{(0)}}{\sum_{j=1}^mX_{js}^{(0)}} \quad \mbox{and} \quad \pi_{is}=\frac{X_{is}}{\sum_{j=1}^mX_{js}}.
\end{align*}
The observed OTUs data were simulated by
\begin{align*}
(Y_{1s},\dots,Y_{ms})\sim^{\text{i.i.d.}}  \text{Multinomial}(N_s,\pi_{1s},\dots,\pi_{ms}).
\end{align*}
To create a power curve, we included six effect sizes labeled as  $\{1,2,...,6\}$ in the figures.  We made the effect sizes have the same signs for differential taxa (i.e., the differential taxa have the same direction of change), creating a relatively strong compositional effect. 
Since low-abundance taxa have much less statistical power, we up-weighted their effects so that the power will not be dominated by those abundant ones. Specifically, for a randomly drawn differential taxon $i$, we set
\begin{align*}
\alpha_i=
\begin{cases}
&\log(2\mu)\mathbb{I}\left(\bar\pi^{(0)}_{i}>0.005\right)+\log\left\{2\mu\left(0.005/\bar\pi_i^{(0)}\right)^{{1}/{3}}\right\}\mathbb{I}\left(\bar\pi^{(0)}_{i}\le 0.005\right)\text{ for }n=50,\\
&\log(\mu)\mathbb{I}\left(\bar\pi^{(0)}_{i}>0.005\right)+\log\left\{\mu\left(0.005/\bar\pi_i^{(0)}\right)^{{1}/{3}}\right\}\mathbb{I}\left(\bar\pi^{(0)}_{i}\le 0.005\right)\text{ for }n=200,
\end{cases}
\end{align*}
where $\mu$ is equally spaced on $[1.05, 2]$ and $\bar\pi_i^{(0)}=\sum_{s=1}^n\pi_{is}^{(0)} / n$.
We considered three cases for the covariate and confounders:
\begin{itemize}
	\item[C0.] \textit{A binary covariate.} $u_s\sim^{\text{i.i.d.}}\text{Bernoulli}(1/2)$ and no confounder.
	\item[C1.] \textit{A continuous covariate.} $u_s\sim^{\text{i.i.d.}}N(0,1)$ and no confounder.
	\item[C2.] \textit{A binary covariate of interest and two confounders.} $u_s\sim\text{Bernoulli}(\{1+\exp(-0.5c_{s1}-0.5c_{s2})\}^{-1})$ independently, where $c_{s1}$ and $c_{s2}$ are confounders (i.e., $\mathbf{c}_s=(c_{s1},c_{s2})^\top$). %$c_{s1}$ independently and identically follows the Rademacher distribution 
	In the above, $c_{s1}$ is specified to independently follow the Rademacher distribution and $c_{s2}\sim^{\text{i.i.d.}}N(0,1)$. 
	%	Let $\boldsymbol{\beta}^{(1)}\sim N(\mathbf{1}, \mathbf{I}_m)$ and $\boldsymbol{\beta}^{(2)}\sim N\left(\mathbf{2}, \mathbf{I}_m\right)$, 
	%	where $\mathbf{1}$ and $\mathbf{2}$ are vectors whose entries are all $1$ and $2$, respectively. 
	%  $\boldsymbol{\beta}_{i}^{-(0)}=(\beta^{(1)}_i, \beta^{(2)}_i)^\top$, where $\beta^{(1)}_i$ and $\beta^{(2)}_i$ are the $i$th elements of $\boldsymbol{\beta}^{(1)}$ and $\boldsymbol{\beta}^{(2)}$ respectively.
	The corresponding coefficients of the confounders $\boldsymbol{\beta}_{i}^{-(0)}=(\beta^{(1)}_i, \beta^{(2)}_i)^\top$, $i=1, \dots, m$, were independently generated from a 2-dimensional normal distribution with mean $(1,2)^\top$ and variance $\mathbf{I}_2$, where $\mathbf{I}_2$ denotes the 2 by 2 identity matrix.
\end{itemize}
%The parameters $\beta_i^{(0)}$, $\sigma_i^2$ and $N_s$ were estimated
The parameters $\beta_i^{(0)}$, $\sigma_i^2$, and $N_s$ were generated based on the estimation for a real dataset (COMBO) from the study of the gut microbiota in a general population \cite{Wu:2011}, which consists of 98 samples and 6674 taxa. 
We only used its 500 most abundant taxa. Since $\beta_i^{(0)}$ and $\sigma_i^2$ were not directly estimable using the relative abundance data, we estimated $\beta_i^{(0)}-\beta_j^{(0)}$ and $\sigma_i^{2}+\sigma_j^{2}$ based on the pairwise log ratios, forced some $\beta_i^{(0)}$'s to be zeros to obtain the estimators of $\beta_1^{(0)},\dots, \beta_m^{(0)}$, and derived $\sigma_i^2$ from the values of $\{\sigma_i^{2}+\sigma_j^{2}\}_{i,j}$. We assume that the library size for each sample follows the negative binomial distribution 
\begin{equation*}
N_s\sim^{\text{i.i.d.}}\text{NB}(7645, 5.3), 
\end{equation*}
%where $\text{NB}(a,b)$ represents the distribution that characterizes the number of failures which occur in a sequence of Bernoulli trials with success probability $b$ in each trial before a target number of successes $a$ is reached. 
where the mean and dispersion parameters were estimated based on the combo data. The resulting sparsity (percent of zeros) of the count matrix is around 65\%--75\%.

In addition to the basic setting (S0, log normal abundance distribution), we designed seven other settings to study the
robustness of the proposed method. Specifically, on top of S0 and C0 (a binary covariate), we studied
\begin{itemize}
	\itemsep0em 
	\item[S1.] \textit{zero-inflated absolute abundances.}  The microbiome data contains excessive zeros and many zeros in the microbiome data can be explained by insufficient sampling \cite{Silverman:2020} since majority of the taxa are of low-abundance. However, it is also possible that zeros are due to physical absence of the taxa \cite{Kaul:2017}. To study the effect of zero inflation on differential abundance analysis, we randomly forced 30\% of the absolute abundance data to be 0.  
	\item[S2.] \textit{Correlated absolute abundances.}  Existing differential abundance analysis methods assume independence among taxa. However, in practice, taxa are interconnected forming networks \cite{Kurtz:2015}.  It is interesting to see if  the methods compared are robust to the correlations among the taxa.  
	In this setting, we simulated block-correlation structure by dividing the 500 taxa into 25 equal-sized blocks. Within each block, we further divided the block into 2 by 2 sub-blocks and simulated equal positive correlations (0.5) within each sub-block and equal negative correlations ($-0.5$) between the two sub-blocks. This mimics the scenario that there are  mutualistic relationships within the group and competitive relationships between groups. 
	\item[S3.] \textit{Gamma abundance distribution.} Although the log normal distribution has been widely used for modeling species abundance data, other models such as gamma distribution are also possible \cite{Connolly:2014}. We thus did additional simulation studies using the gamma distribution. Let $X_{is}^{(0)}\sim^{\text{i.i.d.}}\text{Gamma}(\eta_i^{(0)},1)$ and $X_{is}\sim^{\text{i.i.d.}}\text{Gamma}(\eta_i^{(0)} \exp(u_s\alpha_i+\mathbf{c}_s^\top\boldsymbol{\beta}_i^{-(0)}),1)$. Similarly, we estimated $\eta_i^{(0)}$ from the COMBO data, where we first estimated the baseline proportion $\pi_i^{(0)}$ based on the Dirichlet-multinomial distribution using the R function \texttt{dirmult} and set the over-dispersion parameter $\theta^{(0)}$ to be 0.003, then let $\eta_i^{(0)}=\pi^{(0)}_i(1/\theta^{(0)} - 1)$.
	\item[S4.] \textit{Smaller $m$.} In microbiome data, each taxon can be assigned a taxonomic lineage and taxa abundances can be aggregated at different taxonomic ranks. Differential abundance analysis at higher ranks such as family and genus is also routinely performed. At the higher ranks, the number of taxa is much smaller. We thus studied a small number of taxa ($m=50$) to see if the proposed method is robust to a small $m$. In this setting, we randomly chose 50 elements from $\boldsymbol{\beta}^{(0)}=(\beta_1^{(0)},...,\beta_{500}^{(0)})^\top$ and $\boldsymbol{\sigma}^2=(\sigma_1^2,...,\sigma_{500}^2)^\top$ in each simulation run. We set $N_s\sim\text{NB}(1500, 5.3)$.
	\item[S5.] \textit{Smaller $n$.}  In pilot microbiome studies, the sample sizes are usually small.  It is interesting to study the performance of the methods at a much smaller sample size.  We studied $n=20,\;30$, and used the same effect size as $n=50$.
	\item[S6.] \textit{10-fold difference in library size.}  When the microbiome samples are not fully randomized in sequencing, it is likely that samples of the two groups end up in two separate sequencing runs leading to very different library sizes for the two groups. Since the presence/absence of a taxon strongly depends on the library size, the differential library size will confound the two-sample comparison, especially for those rare taxa \cite{Weiss:2017}. To create differential library sizes, we generated the library size from $N_s\sim\text{NB}(5000, 5.3)$ and $N_s\sim\text{NB}(50000, 5.3)$ for the two groups, respectively.
	\item[S7.] \textit{Negative bionomial abundance distribution.} DESeq2 and edgeR  assume
	negative binomial distribution for the counts, thus we included one more simulation setting, where we generated the counts from the negative binomial distribution. Let $X_{is}^{(0)}\sim^{\text{i.i.d.}}\text{NB}(\text{exp}(7645\kappa_i^{(0)} ), \theta_i^{(0)})$ and $Y_{is}\sim^{\text{i.i.d.}}\text{NB}(\text{exp}(N_s\kappa_i^{(0)} +u_s\alpha_i+\mathbf{c}_s^\top\boldsymbol{\beta}_i^{-(0)}), \theta_i^{(0)})$.  Similarly, we estimated the $\kappa_i^{(0)}$ (regression coefficient for the library size with respect to the log of the count of $i$th taxon) and $\theta_i^{(0)}$ (dispersion parameter) from the COMBO data using the R function \texttt{glm.nb}.
	\item[S8.] \textit{Mixed-effect model.}  We considered two scenarios: \textit{Pre-treatment and post-treatment comparison} (S8.1) and \textit{Replicate sampling} (S8.2). Under S8.1, for $n=50$ (or 200), we simulated $25$ (or 100) subjects and each has paired pre-treatment and post-treatment samples. The aim is to detect taxa affected by treatment.  Under S8.2, each subject has multiple measurements. For $n=50$ (or 200), we generated 25 (or 50) subjects with each having 2 (or 4) replicates. 
	Specifically, we let
	$$\log(X_{is})\sim \mathbf{r}_s^\top\boldsymbol{\gamma}_i+N(\beta_i^{(0)}+u_{s}\alpha_i+\mathbf{c}_{s}^\top\boldsymbol{\beta}_i^{-(0)},\sigma_i^2),$$ 
	where $\mathbf{r}_s$ has one element equal to 1 and all the others equal to 0 indicating the subject ID of sample $s$. %and the elements of
	Each element of $\boldsymbol{\gamma}_i$ follows $N(0,\tau_i^2)$ independently, where we let $\tau_i^2=a_i\sigma_i^2$ with $a_i\sim \text{Unif}([0,1])$.
\end{itemize}

\section*{Declarations}
\subsection*{Ethics approval and consent to participate}
Not applicable.
\subsection*{Consent for publication}
Not applicable.
\subsection*{Availability of data and materials}
LinDA is implemented as the \texttt{linda} function in the CRAN R package \texttt{MicrobiomeStat} (\url{https://CRAN.R-project.org/package=MicrobiomeStat}). The \texttt{LinDA} package is also available at GitHub (\url{https://github.com/zhouhj1994/LinDA}).
The entire codes and data for generating the presented results are available  at the repository Zenodo (\url{https://doi.org/10.5281/zenodo.6326019}) and at GitHub (\url{https://github.com/zhouhj1994/LinDA-manuscript-result}) under MIT License.  The CDI, IBD, RA, SMOKE and COMBO datasets are from \cite{Schubert:2014,Morgan:2012,Scher:2013,Charlson:2010,Wu:2011}, respectively.
\subsection*{Competing interests}
The authors declare no competing interests.
\subsection*{Funding}
This work was supported by National Institute of Health R01GM144351 (Chen \& Zhang),   National Science Foundation DMS-1830392,  DMS2113359,  DMS1811747 (Zhang \& Zhou) and  National Science Foundation DMS2113360 and Mayo Clinic Center for Individualized Medicine (Chen).
\subsection*{Authors' contributions}
J.C. and X.Z. conceived, designed and supervised the work together. X.Z. developed the method. X.Z., K.H and H.Z. performed the theoretical analysis. H.Z. and J.C. performed the evaluation and developed the software. J.C., H.Z.,  X.Z. and K.H wrote and revised the manuscript together.
\subsection*{Acknowledgements}
Not applicable.

\section*{Figures and Tables}
{\bf Fig. \ref{fig-S0C0}} Performance comparison (S0C0: log normal abundance distribution, a binary covariate). Empirical false discovery rate (A) and true positive rates (B) were averaged over 100 simulation runs. Error bars (A) represent the 95\% confidence intervals (CIs) of the method LinDA and the dashed horizontal line indicates the target FDR level of 0.05.

{\noindent\bf Fig. \ref{fig-S3C0}} Performance comparison (S3C0: gamma abundance distribution, a binary covariate). Empirical false discovery rate (A) and true positive rates (B) were averaged over 100 simulation runs. Error bars (A) represent the 95\% CIs of the method LinDA and the dashed horizontal line indicates the target FDR level of 0.05.

{\noindent\bf Fig. \ref{fig-real-curve}} Number of discoveries v.s. target FDR level (0.01--0.25) for the three real datasets.

{\noindent\bf Fig. \ref{fig-real-venn}} Overlaps of differential taxa with target FDR level of 0.1 for the four real datasets.
	
{\noindent\bf Fig. \ref{fig-ills_cond}} Density of $\sqrt{n}\alpha_i+\varepsilon_{is}$. The panels on the left and right correspond to $\sigma_i=1$ and $\sigma_i\sim\text{IG}(2, 1)$ respectively, where IG denotes the inverse-gamma distribution. The red curve is the density of the standard normal distribution. The blue and green curves are the densities of  $\sqrt{n}\alpha_i+\varepsilon_{is}$ with $\mathbb{P}(\sqrt{n}\alpha_i=0)=0.8$ and  $\mathbb{P}(\sqrt{n}\alpha_i=2)=0.2$, and $\mathbb{P}(\sqrt{n}\alpha_i=0)=0.8$ and  $\mathbb{P}(\sqrt{n}\alpha_i=5)=0.2$, respectively.

%{\noindent\bf{Table} \ref{tab-simu}} {\color{red}Summary of the performance comparisons. Table \ref{tab-simu}A compares the empirical FDR. Three $\star$  represents that the FDR is controlled; two, one  and zero $\star$ represent slight, large and severe FDR inflation, respectively.  Table \ref{tab-simu}B compares the true positive rate. More $\star$ represents higher power compared to the other methods.}
%The two numbers in each entry respectively represent the FDR score and power score averaged over the 24 sub-settings under each setting. FDR score: $<=0.05$, 5; $>0.05$\&$<=0.07$, 4; $>0.07$\&$<=0.1$, 3; $>0.1$\&$<=0.2$, 2; $>0.2$\&$<=0.4$, 1; $>0.4$, 0. Power score: power $/$ (maximum power over the five methods) $\times$ 5.}

{\noindent\bf Table \ref{tab-speed}} Runtime (in second) comparison under different settings (R version 4.0.3 (2020-10-10); Platform: x86\_64-pc-linux-gnu (64-bit); CPU: E5-2670 v2 @ 2.50GHz; Memory: 67.7 GB). The result is based on one simulation run.  The``elapsed" from the R command \texttt{system.time()} was used.

{\noindent\bf{Table} \ref{tab-real}} Characteristics of four real microbiome datasets. NORA represents new-onset untreated rheumatoid arthritis. The second and the third columns respectively list the number of taxa and sample size of each filtered dataset (prevalence $\geq$ 10\%, library size $\geq$ 1000).

\section*{Supplementary material}
{\noindent \bf Additional file 1\labeltext{1}{supp1}}: supplementary notes. {\bf Table S1} lists some robust normalization methods \cite{Robinson:2010,Anders:2010,Paulson:2013,Chen:2018}. {\bf Lemmas S1 -- S4} present intermediate results for proving Theorem \ref{the-fdr}. The proof depends on some useful results from \cite{Zhou:2021,Vershynin:2018,Wainwright:2019,Cao:2021}.

{\noindent \bf Additional file 2\labeltext{2}{supp2}}: supplementary figures. 	{\bf Fig. S1} and {\bf S2} compare the proposed method LinDA with different zero-handling approaches under settings S6C0 and S0C0. {\bf Fig. S3}
depicts the results of LinDA, CLR-OLS and MaAsLin2 with different normalization approaches under setting S0C0. 
%Figure \ref{fig-S6C0-ancombc} compares the ANCOM-BC disabling and enabling zero treatment for setting S6C0. 
{\bf Fig. S4}--{\bf S10} and {\bf S12}--{\bf S14} show the results of settings S0C1, S0C2, S1C0, S2C0, S4C0, S5C0, S6C0, S7C0, S8.1C0, and S8.2C0, respectively. The comparison between disabling and enabling zero treatment of the ANCOM-BC method is depicted in {\bf Fig. S11} under setting S6C0. 
{\bf Fig. S15} shows the results of setting S0C0 with stronger compositional effects. 
{\bf Fig. S16}--{\bf S19} show the effect size plots and volcano plots for the four datasets (CDI, IBD, RA, and SMOKE) respectively.
{\bf Fig. S20}--{\bf S30} present the full result of all methods under different simulation settings.

\newpage
\begin{figure}
	\begin{subfigure}[b]{1\textwidth}
		\centering
		\includegraphics[scale=0.5]{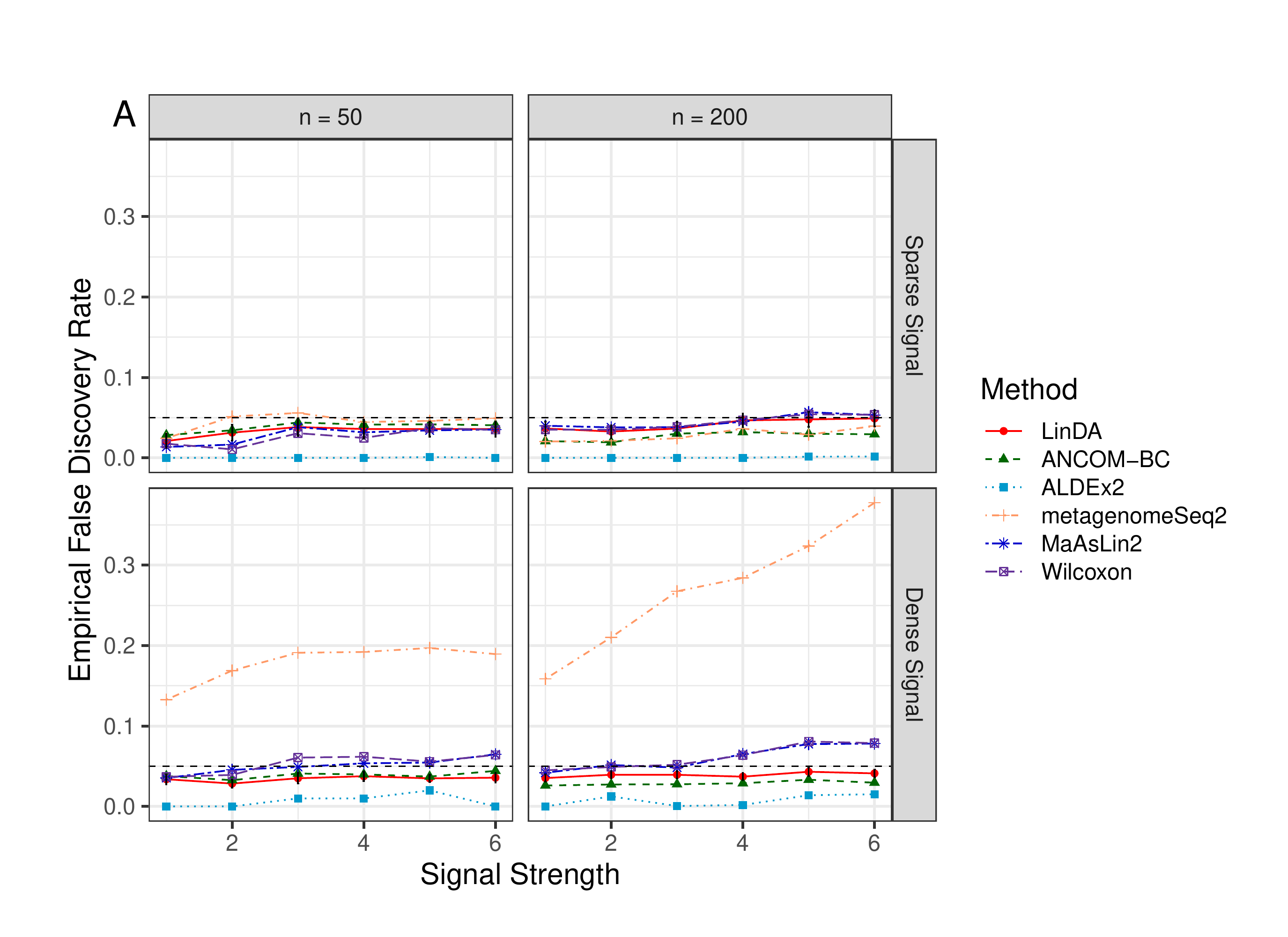}
	\end{subfigure}
	\begin{subfigure}[b]{1\textwidth}
		\centering
		\includegraphics[scale=0.5]{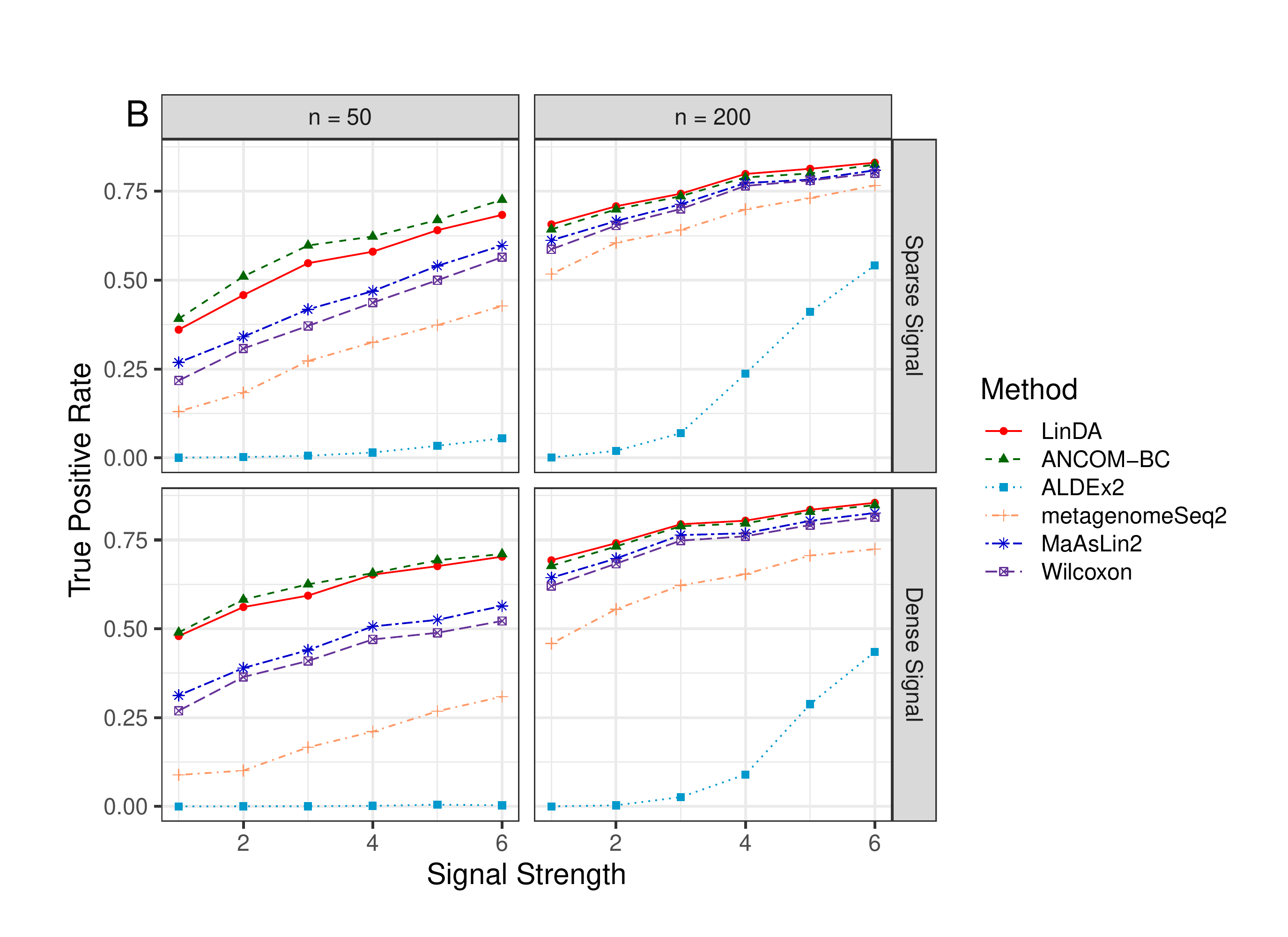}
	\end{subfigure}
	\caption{}
	\label{fig-S0C0}
\end{figure}
\begin{figure}
	\begin{subfigure}[b]{1\textwidth}
		\centering
		\includegraphics[scale=0.5]{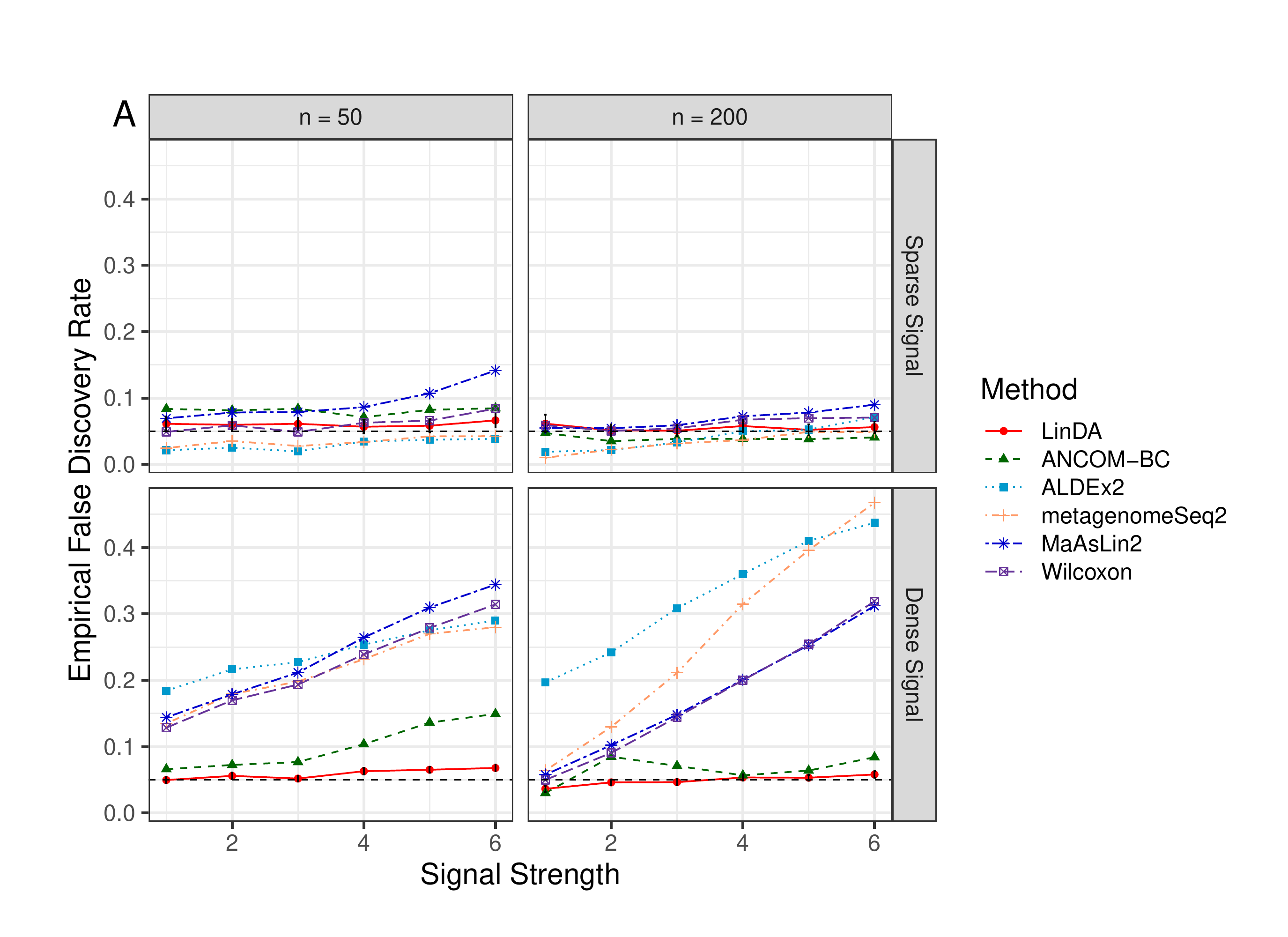}
	\end{subfigure}
	\begin{subfigure}[b]{1\textwidth}
		\centering
		\includegraphics[scale=0.5]{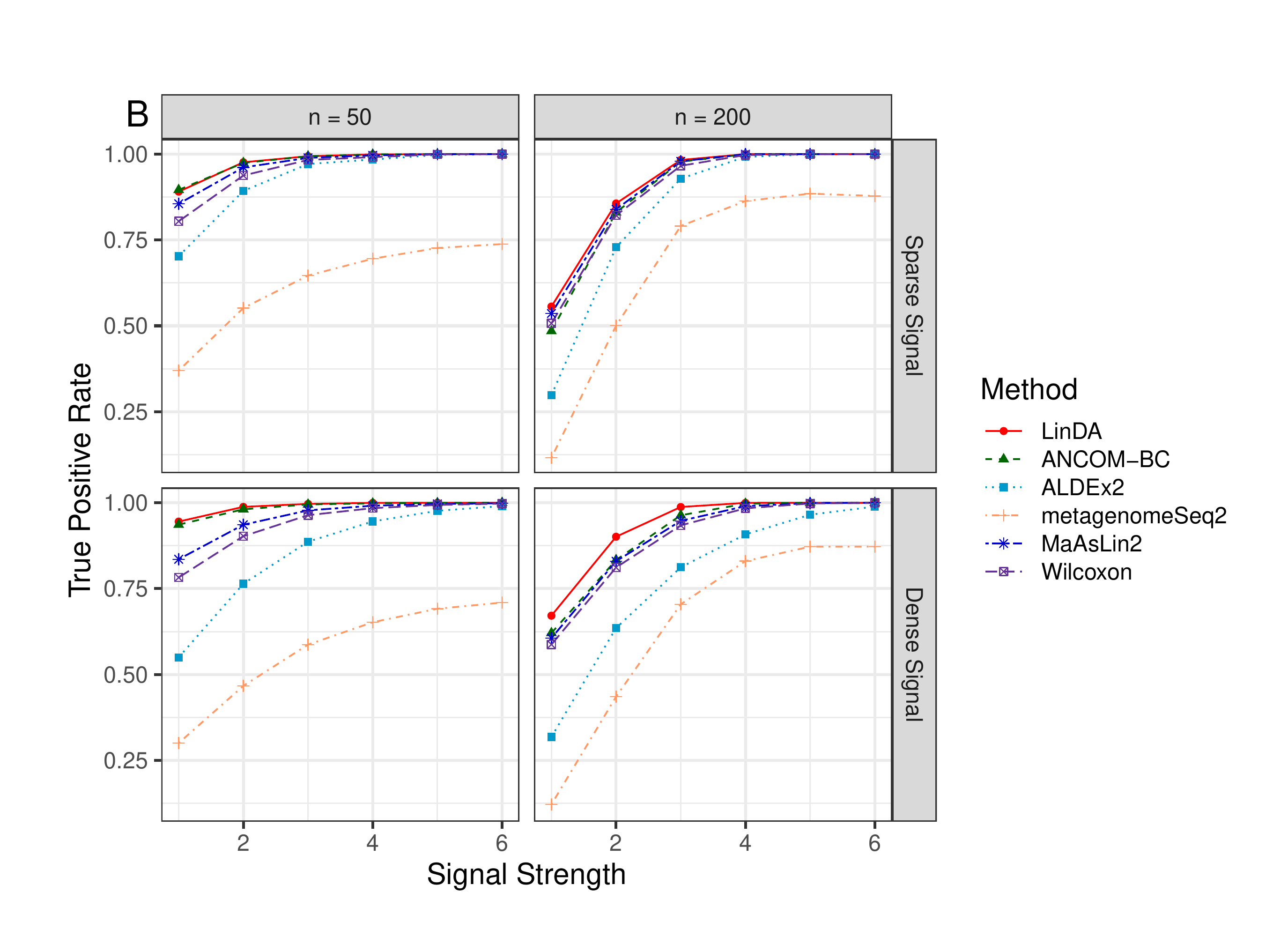}
	\end{subfigure}
	\caption{}
	\label{fig-S3C0}
\end{figure}

\begin{figure}
	\centering
	\includegraphics[scale=0.3]{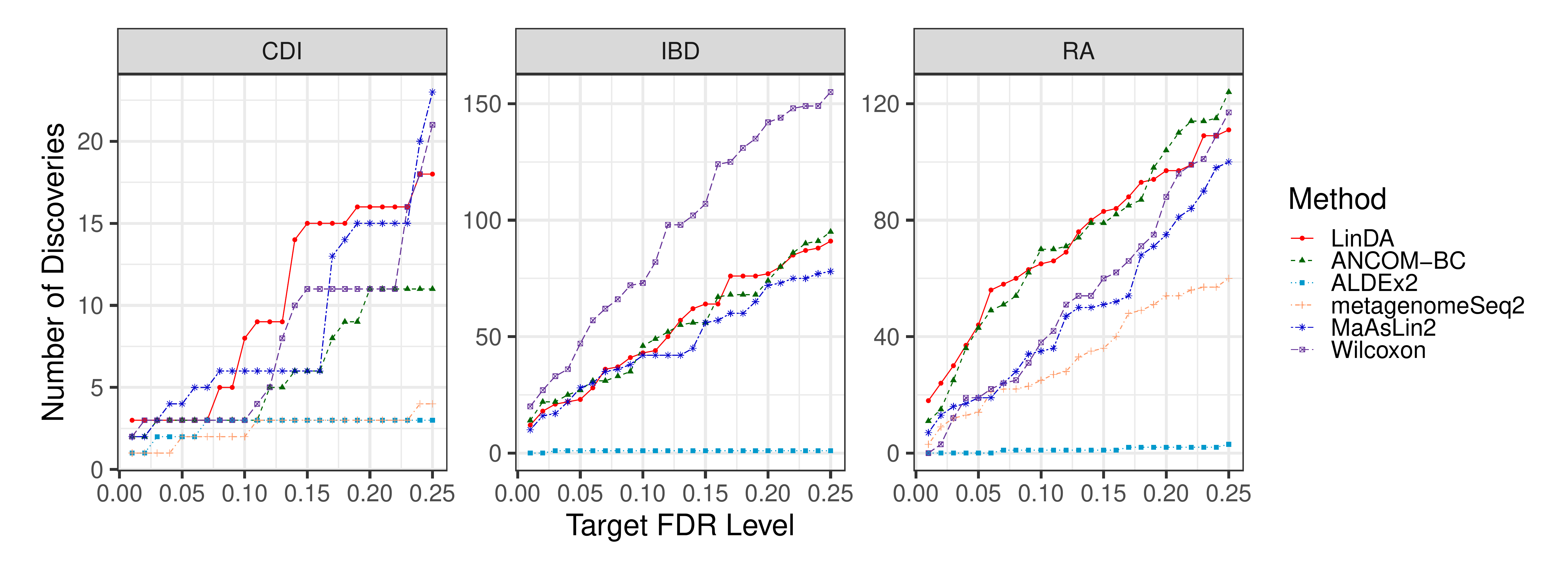}
	\caption{}
	\label{fig-real-curve}
\end{figure}

\begin{figure}
	\begin{subfigure}[b]{0.45\textwidth}
		\centering
		\includegraphics[scale=0.25]{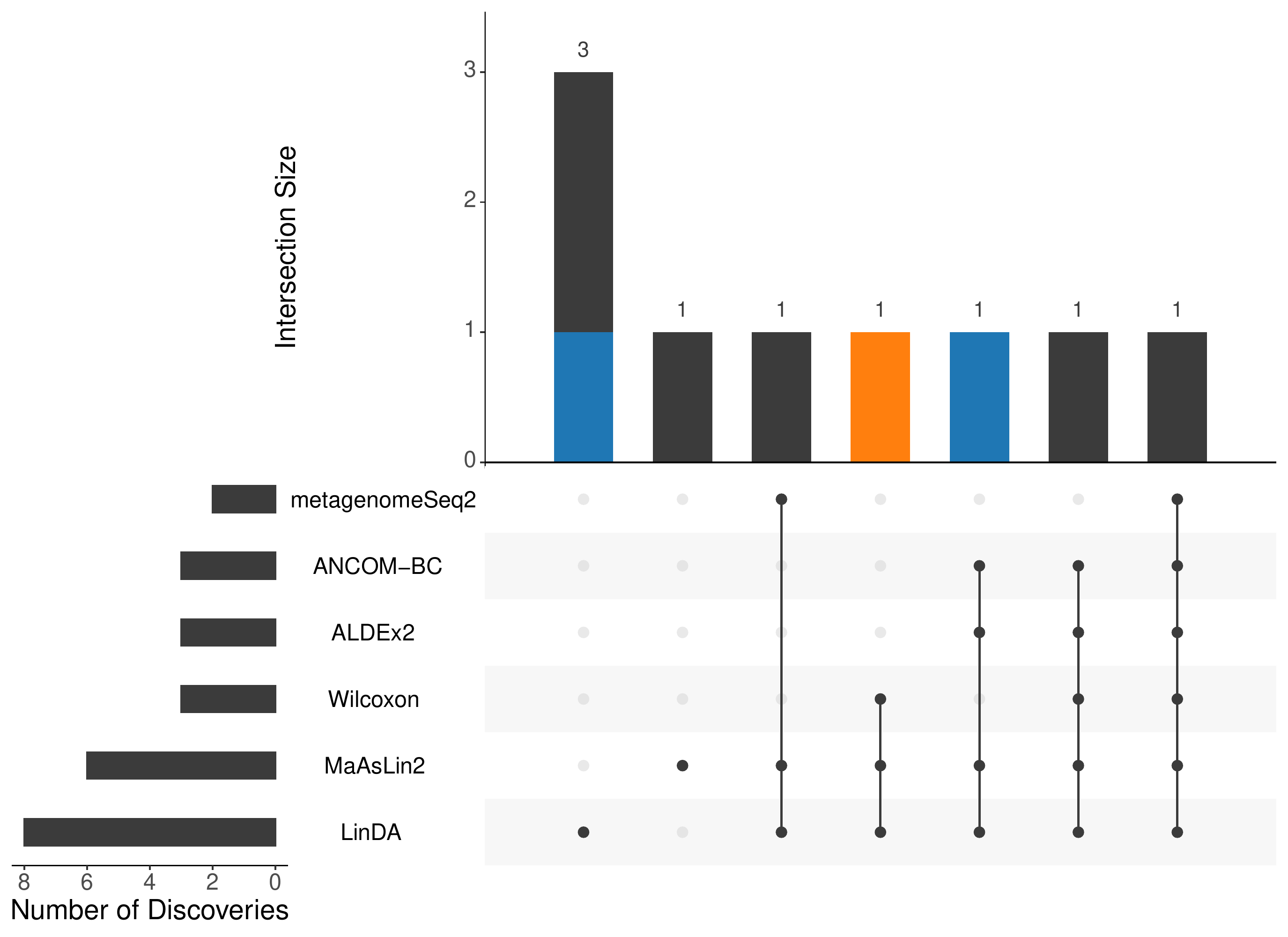}
		\subcaption*{CDI}
	\end{subfigure}
	\hfill
	\begin{subfigure}[b]{0.45\textwidth}
		\centering
		\includegraphics[scale=0.25]{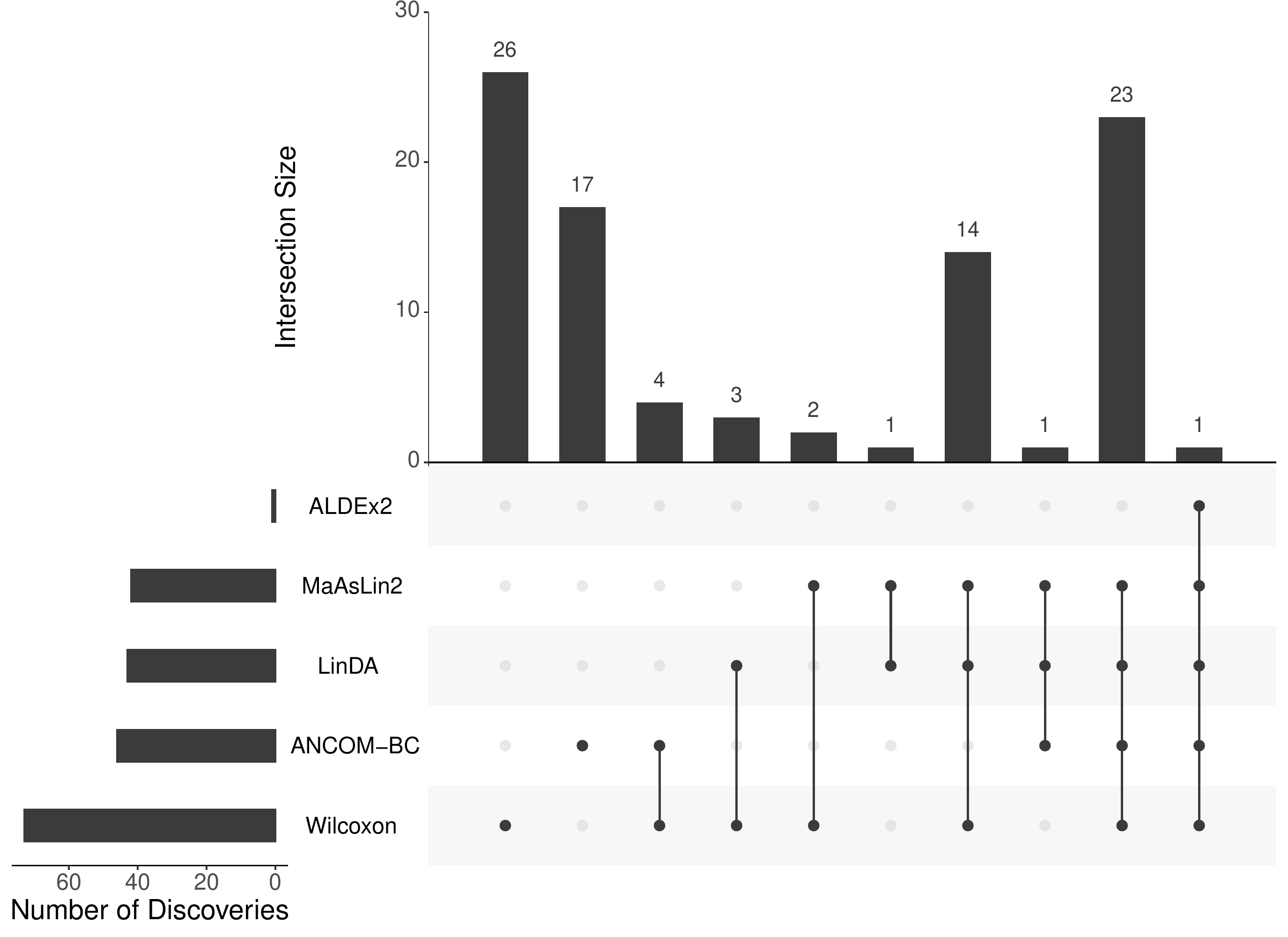}
		\subcaption*{IBD}
	\end{subfigure}
	\hfill
	\begin{subfigure}[b]{0.45\textwidth}
		\centering
		\includegraphics[scale=0.25]{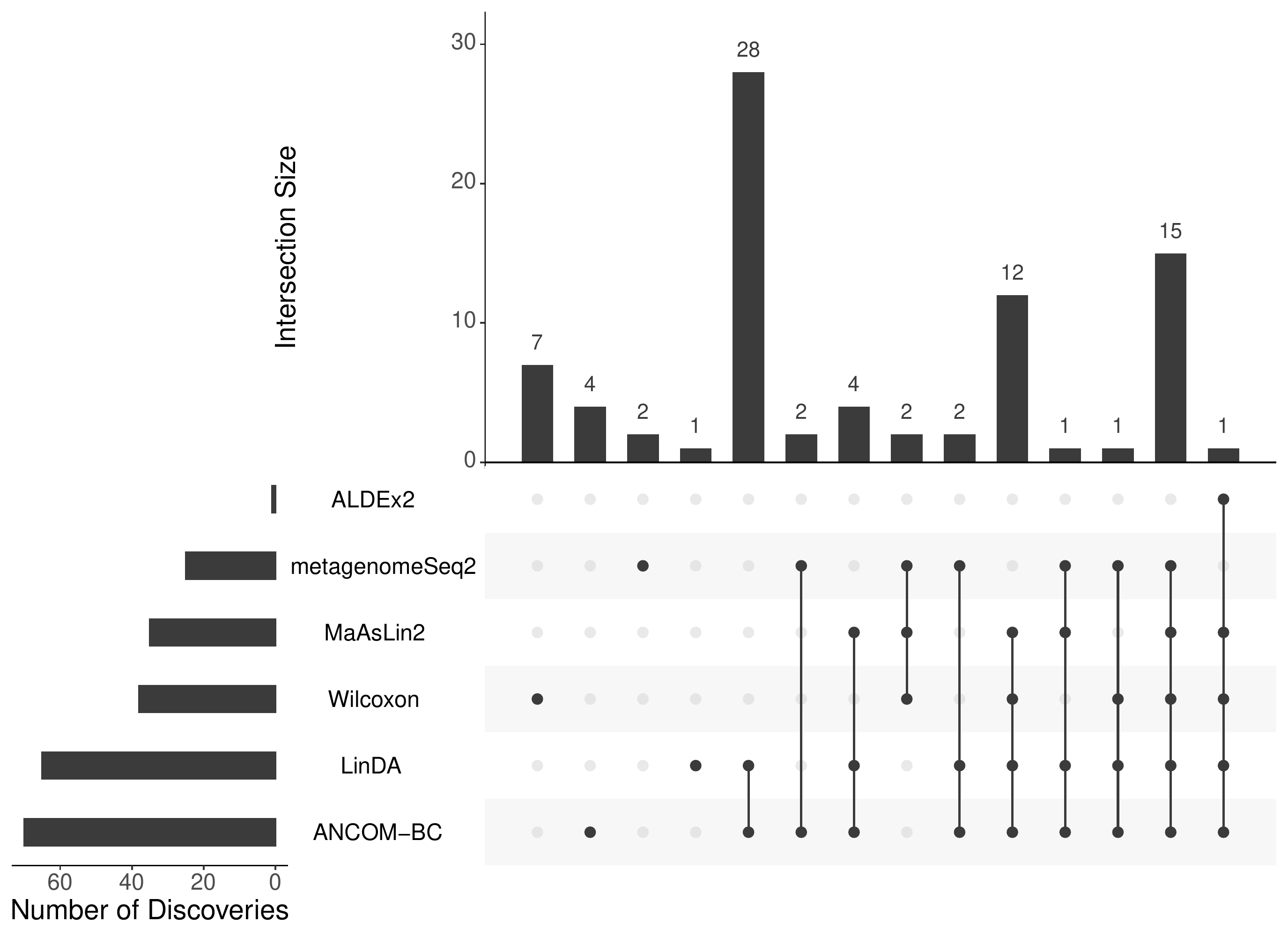}
		\subcaption*{RA}
	\end{subfigure}
	\hfill
	\begin{subfigure}[b]{0.45\textwidth}
		\centering
		\includegraphics[scale=0.25]{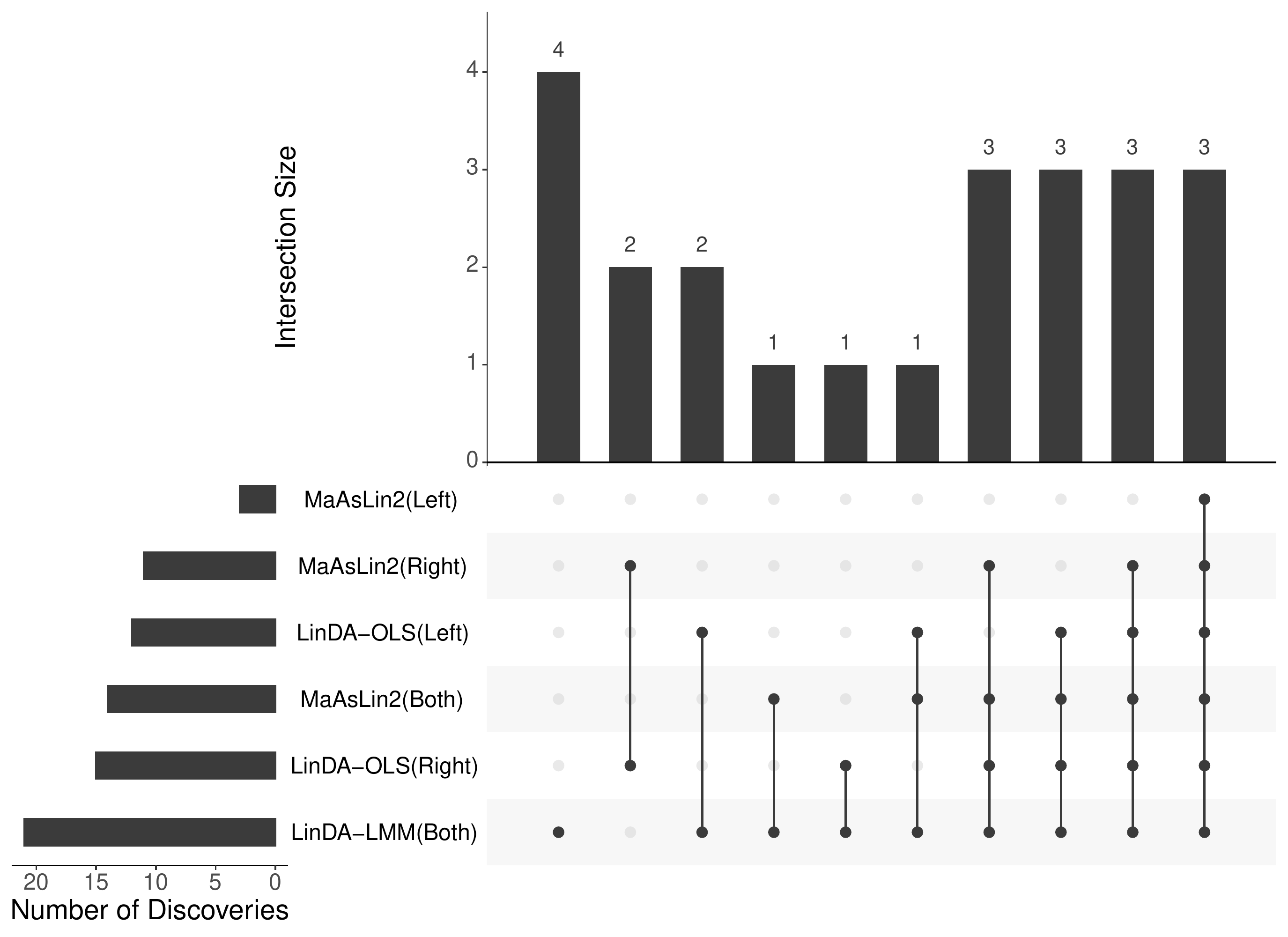}
		\subcaption*{SMOKE}
	\end{subfigure}
	\hfill
	\caption{}
	\label{fig-real-venn}
\end{figure}

\begin{figure}
	\centering
	\includegraphics[scale=0.55]{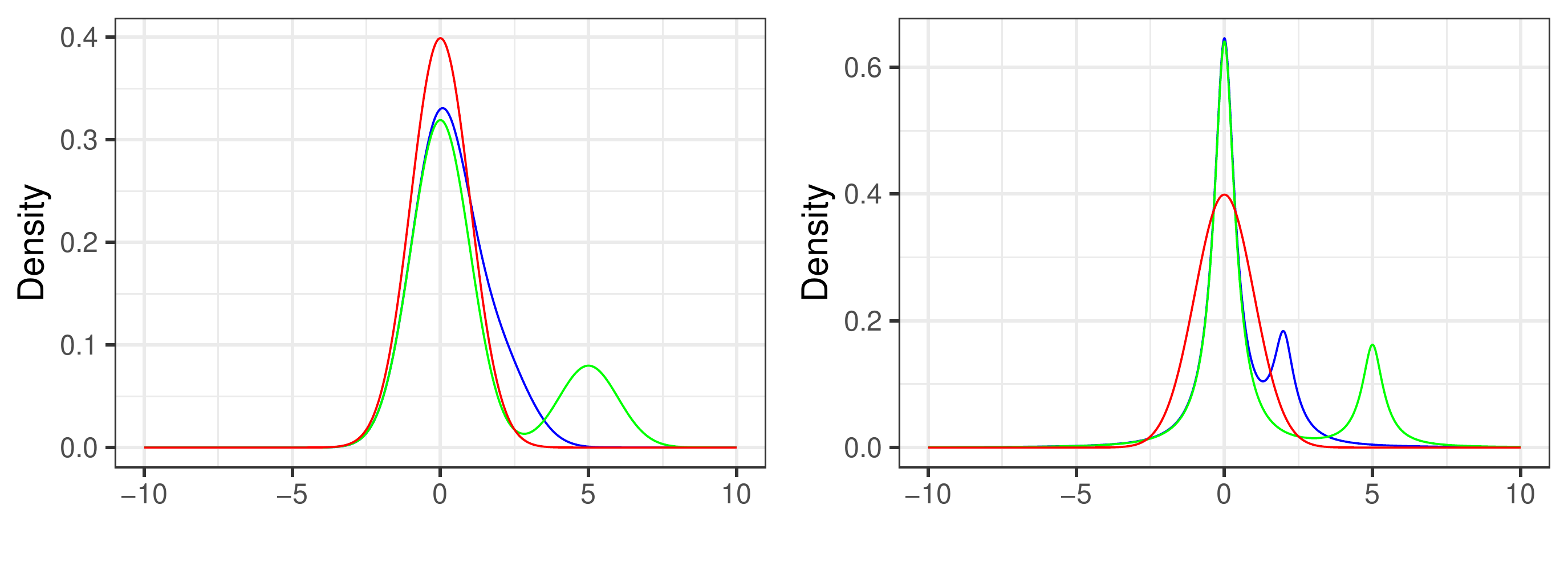}
	\caption{}
	\label{fig-ills_cond}
\end{figure}

\begin{table}\centering 
	\caption{} 
	\scriptsize
	\begin{tabular}{llllllll} 
		\hline 
		\hline 
		&&\multicolumn{2}{c}{S0C0}&\multicolumn{2}{c}{S0C1}&\multicolumn{2}{c}{S0C2}\\
		\hline
		&&LinDA &ANCOM-BC&LinDA &ANCOM-BC&LinDA &ANCOM-BC\\
		\hline
		$m=500$&$n=200$ & $\bf{0.454}$ & $21.835$ & $\bf{0.218}$ & $22.057$ & $\bf{0.206}$ & $64.519$ \\ 
		&$n=10000$ & $\bf{6.844}$ & $162.218$ & $\bf{4.043}$ & $163.552$ & $\bf{5.073}$ & $216.564$ \\ 
		$m=5000$&$n=200$ & $\bf{1.598}$ & $184.972$ & $\bf{1.607}$ & $162.611$ & $\bf{1.615}$ & $599.985$ \\ 
		&$n=10000$ & $\bf{28.253}$ & $5,135.393$ & $\bf{15.314}$ & $5,157.148$ & $\bf{15.494}$ & $5,506.353$ \\ 
		\hline 
	\end{tabular} 
	\label{tab-speed} 
\end{table}

\begin{table}\centering 
	\caption{} 
	\scriptsize
	\begin{tabular}{lcccc} 
		\hline 
		\hline 
		&$m $&$n$&$u$&$\mathbf{c}$\\
		\hline
		CDI &123&183 &CDI/Diarrhea control (94 v.s. 89) & \\
		IBD&579 &81 &Crohn's disease/Healthy (62 v.s. 19)&Antibiotic use (n/y, 48 + 19 v.s. 14 + 0)\\
		RA&438 &72 &NORA/Healthy (44 v.s. 28) &\\
		SMOKE&209&132&Smoke (n/y, 67 v.s. 65)& Female/Male (31 + 16 v.s.  36 + 49)\\
		\hline 
	\end{tabular} 
	\label{tab-real} 
\end{table}

\setcounter{section}{0}
\renewcommand{\thesection}{S\arabic{section}}
\setcounter{subsection}{0}
\renewcommand{\thesubsection}{S\arabic{subsection}}
\setcounter{equation}{0}
\renewcommand{\theequation}{S\arabic{equation}}
\setcounter{figure}{0}
\renewcommand{\thefigure}{S\arabic{figure}}
\setcounter{table}{0}
\renewcommand{\thetable}{S\arabic{table}}
\setcounter{lemma}{0}
\renewcommand{\thelemma}{S\arabic{lemma}}

	\newpage
	~
		\newpage
		~
		\newpage
	\spacingset{1.5} % DON'T change the spacing!
	{\huge Supplementary notes for ``LinDA: linear models for differential abundance analysis of microbiome compositional data"}
	\section{Normalization approaches}\label{sec-norm-tab}
	\begin{table}[H]\centering
		\caption{Some robust normalization methods \label{table-norms}}
		\begin{tabular}{p{0.3\textwidth}p{0.65\textwidth}}
			\hline
			Method&Description\\
			\hline
			Trimmed mean of M-values (TMM) [11, Robinson \& Oshlack, 2010]& TMM (in log scale) is the weighted mean of the log-ratio between the relative abundances and a referenced relative abundance after excluding the most abundant taxa and the taxa with the largest log-fold changes. \\
			\hline
			DESeq normalization (RLE) [12, Anders \& Huber, 2010]& In RLE, the normalizing factor is the median of the ratios between the counts and the geometric mean of the counts of all samples.\\
			\hline
			Cumulative-sum scaling (CSS) [13, Paulson et al., 2013]& In CSS, counts are divided by the cumulative sum of counts, up to a quantile determined by a data-driven approach. \\
			\hline
			Geometric mean of pairwise ratios (GMPR) [14, Chen et al., 2018]& GMPR is the geometric mean of the medians of the ratios between the pairs of counts of two samples, which reverses the order of the two steps in the RLE.\\
			\hline
		\end{tabular}
	\end{table}
	Note that the first number in the square brackets represents the reference number in the manuscript.
	\section{Technical details}\label{sec-proof}
	%When an arithmetic for a scalar is applied to a vector, %or a matrix, 
	%the operation is understood as element-wise. For example, $\mathbf{v}^2$ denotes taking square on each element of the vector $\mathbf{v}$. 
	In the following, we use $F_{X}(\cdot)$ to denote the cumulative distribution function of a random variable $X$. Denote by $o_{\mathbb{P}_m}$ ($O_{\mathbb{P}_m}$), $o_{\mathbb{P}_n}$ ($O_{\mathbb{P}_n}$), and $o_{\mathbb{P}}$ ($O_{\mathbb{P}}$) the corresponding rates of convergence as $m\to\infty$, $n\to\infty$, and $m,n\to\infty$ simultaneously, respectively. We first introduce some useful lemmas before proving Theorem 1. 
	%We point out that the conditions in Theorem \ref{the-fdr} are needed for the theorem, but when it comes to each of the following lemmas, they might can be repalced by weaker conditions. We mention it sometimes in the proof and usually it could just be implied by the proof.

	\begin{lemma}\label{lem-rho}
		Under Condition (i) in Theorem 1, we have 
		\begin{align*}
			|\hat{\rho}-\rho|=o_{\mathbb{P}_n}(1).
		\end{align*}		
	\end{lemma}
	\begin{lemma}\label{lem-sigma}
		Under Conditions (i), (ii), (iii), (v), and (ix) in Theorem 1, we have 
		\begin{align*}
			\max_i|\hat{\sigma}_i^2-\sigma_i^2|=o_{\mathbb{P}}(1).
		\end{align*}
	\end{lemma}
	\begin{lemma}\label{lem-mode}
		Under Conditions (i)--(viii) in Theorem 1, we have 
		\begin{align*}
			\sqrt{n}(\tilde{\alpha}-\bar{\alpha})=o_{\mathbb{P}}(1).
		\end{align*}
	\end{lemma}
	\begin{lemma}\label{lem-empirical}
		Suppose Conditions (i)--(ix) in Theorem 1 are satisfied.
		Let $m_0$ be the number of true null hypotheses and 
		\begin{align*}
			&V_{m,n}(t)=\sum_{i:\alpha_i=0}\mathbb{I}\left(|\sqrt{n}\hat{\alpha}_i|/\sqrt{\hat{\rho}\hat{\sigma}_i^2}>t\right),\\
			&S_{m,n}(t)=\sum_{i=1}^m\mathbb{I}\left(|\sqrt{n}\hat{\alpha}_i|/\sqrt{\hat{\rho}\hat{\sigma}_i^2}>t\right),\\
			&S_{\infty,n}(t)=\mathbb{P}\left(\Big|\mathcal{E}+\sqrt{n}\alpha_i/\sqrt{\rho\sigma_i^2}\Big|>t\right).
		\end{align*}
		Then for any $0<t_0<\infty$, 
		\begin{equation*}
			\sup_{0<t<t_0}\left|m^{-1}S_{m,n}(t)-S_{\infty,n}(t)\right|=o_{\mathbb{P}}(1) \, \mbox{ and } \, \sup_{0<t<t_0}\left|m_0^{-1}V_{m,n}(t)-2F_{n-d-2}(-t)\right|=o_{\mathbb{P}}(1).
		\end{equation*}
		%	and
		%	\begin{equation*}
			%	\sup_{0<t<t_0}\left|m_0^{-1}V_{m,n}(t)-2F_{n-d-2}(-t)\right|=o_{\mathbb{P}}(1).
			%	\end{equation*}	
		
	\end{lemma}

	\begin{proof}[Proof of Lemma \ref{lem-rho}]
		From Condition (i), we know that each element of $\mathbb{E}(\mathbf{z}_s\mathbf{z}_s^\top)$ is finite and $\text{det}\{\mathbb{E}(\mathbf{z}_s\mathbf{z}_s^\top)\}> C$. We have $\hat{\rho}=\text{det}(\hat{\mathbf{B}})/\text{det}(\hat{\mathbf{A}})$ and $\rho=\text{det}(\mathbf{B})/\text{det}(\mathbf{A})$, where $\mathbf{A}=\mathbb{E}(\mathbf{z}_s\mathbf{z}_s^\top)$, $\hat{\mathbf{A}}=n^{-1}\sum_{s=1}^n\mathbf{z}_s\mathbf{z}_s^\top$, and $\mathbf{B}$ and $\hat{\mathbf{B}}$ are the principal submatrices obtained by deleting the first row and first column of $\mathbf{A}$ and $\hat{\mathbf{A}}$ respectively. Thus we have that $|\text{det}(\hat{\mathbf{B}})-\text{det}(\mathbf{B})|=o_{\mathbb{P}}(1)$ and $|\text{det}(\hat{\mathbf{A}})-\text{det}(\mathbf{A})|=o_{\mathbb{P}}(1)$ using the law of large numbers. The Slutsky's theorem thus implies that $|\hat{\rho}-\rho|=o_{\mathbb{P}_n}(1)$.
	\end{proof}

	\begin{proof}[Proof of Lemma \ref{lem-sigma}]
		Throughout the proof, we shall assume that $\varepsilon_{is}/\sigma_i$ is $C$-sub-Gaussian, which is indeed slightly weaker than Condition (iii). For any $\lambda>0$, we have
		\begin{align*}
			&\mathbb{E}[e^{\lambda\varepsilon_{is}}\mid\sigma_i]=\mathbb{E}\left[e^{\lambda\sigma_i(\varepsilon_{is}/\sigma_i)}\mid\sigma_i\right]\le e^{\lambda^2\sigma_i^2C^2/2},\\
			&\mathbb{E}[e^{\lambda\bar{\varepsilon}_{is}}\mid\{\sigma_i\}]
			=\mathbb{E}\left[e^{\lambda\left\{(m-1)m^{-1}\varepsilon_{is}-m^{-1}\sum_{j\neq i}\varepsilon_{js}\right\}}\mid\{\sigma_i\}\right]\le e^{\lambda^2(\max_i\sigma_i^2)C^2/2}.
		\end{align*}
		Thus $\bar{\varepsilon}_{is}$ conditional on $\{\sigma_i\}$ is sub-Gaussian by Condition (ii). Let $\bar{\boldsymbol{\theta}}_i=(\bar{\alpha}_i,\bar{\boldsymbol{\beta}}_i^\top)^\top$ and $\tilde{\boldsymbol{\theta}}_i=(\tilde{\alpha}_i,\tilde{\boldsymbol{\beta}}_i^\top)^\top$. 
		Note that
		\begin{align*}
			\tilde{\boldsymbol{\theta}}_i&=\bar{\boldsymbol{\theta}}_i+\left(\sum_{s=1}^n\mathbf{z}_s\mathbf{z}_s^\top\right)^{-1}\left(\sum_{s=1}^n\mathbf{z}_s\bar{\varepsilon}_{is}\right),\\
			\hat{\sigma}_{i}^2&=\frac{1}{n-d-2}\sum_{s=1}^{n}\left(W_{is}-\mathbf{z}_s^\top\tilde{\boldsymbol{\theta}}_i\right)^2=\frac{1}{n-d-2}\sum_{s=1}^{n}\left(W_{is}-\mathbf{z}_s^\top\bar{\boldsymbol{\theta}}_i+\mathbf{z}_s^\top\bar{\boldsymbol{\theta}}_i-\mathbf{z}_s^\top\tilde{\boldsymbol{\theta}}_i\right)^2\\
			&=\frac{1}{n-d-2}\sum_{s=1}^n\bar{\varepsilon}_{is}^2+\frac{2}{n-d-2}(\bar{\boldsymbol{\theta}}_i-\tilde{\boldsymbol{\theta}}_i)^\top\sum_{s=1}^n\mathbf{z}_s\bar{\varepsilon}_{is}\\
			&\qquad \qquad +\frac{1}{n-d-2}(\bar{\boldsymbol{\theta}}_i-\tilde{\boldsymbol{\theta}}_i)^\top\left(\sum_{s=1}^n\mathbf{z}_s\mathbf{z}_s^\top\right)(\bar{\boldsymbol{\theta}}_i-\tilde{\boldsymbol{\theta}}_i)\\
			&=\frac{1}{n-d-2}\sum_{s=1}^n\bar{\varepsilon}_{is}^2-\frac{1}{n-d-2}\left(\sum_{s=1}^n\mathbf{z}_s\bar{\varepsilon}_{is}\right)^\top\left(\sum_{s=1}^n\mathbf{z}_s\mathbf{z}_s^\top\right)^{-1}\left(\sum_{s=1}^n\mathbf{z}_s\bar{\varepsilon}_{is}\right),
		\end{align*}
		and for any $\delta>0$,
		\begin{align*}
			\mathbb{P}\left(\left|\hat{\sigma}_i^2-\bar{\sigma}_{i}^2\right|>\delta\right)
			&\le \mathbb{P}\left(\left|\frac{1}{n-d-2}\sum_{s=1}^n\bar{\varepsilon}_{is}^2-\bar{\sigma}_i^2\right|>\frac{\delta}{2}\right)\\
			&\qquad +\mathbb{P}\left\{\left(\sum_{s=1}^n\mathbf{z}_s\bar{\varepsilon}_{is}\right)^\top\left(\sum_{s=1}^n\mathbf{z}_s\mathbf{z}_s^\top\right)^{-1}\left(\sum_{s=1}^n\mathbf{z}_s\bar{\varepsilon}_{is}\right)>\frac{(n-d-2)\delta}{2}\right\}.
		\end{align*}
		For the first term, we have
		\begin{align*}
			\mathbb{P}\left(\left|\frac{1}{n-d-2}\sum_{s=1}^n\bar\varepsilon_{is}^2-\bar\sigma_i^2\right|>\frac{\delta}{2}\right) 
			&\le\mathbb{P}\left\{\left|\frac{1}{n}\sum_{s=1}^n\bar\varepsilon_{is}^2-\bar\sigma_i^2\right|>\frac{(n-d-2)\delta}{4n}\right\}\\
			& \qquad  +\mathbb{P}\left(\frac{d+2}{n-d-2}\bar\sigma_i^2>\frac{\delta}{4}\right).
		\end{align*}
		For the second term, it can be shown that
		\begin{align*}
			&\mathbb{P}\left\{\left(\sum_{s=1}^n\mathbf{z}_s\bar{\varepsilon}_{is}\right)^\top\left(\sum_{s=1}^n\mathbf{z}_s\mathbf{z}_s^\top\right)^{-1}\left(\sum_{s=1}^n\mathbf{z}_s\bar{\varepsilon}_{is}\right)>\frac{(n-d-2)\delta}{2}\right\}\\
			&\qquad \le \mathbb{P}\Bigg\{\left(\frac{1}{n}\sum_{s=1}^n\mathbf{z}_s\bar{\varepsilon}_{is}\right)^\top\left(\frac{1}{n}\sum_{s=1}^n\mathbf{z}_s\mathbf{z}_s^\top\right)^{-1}\left(\frac{1}{n}\sum_{s=1}^n\mathbf{z}_s\bar{\varepsilon}_{is}\right)>\frac{(n-d-2)\delta}{2n},\\
			&\qquad\qquad\qquad\qquad \left\|\frac{1}{n}\sum_{s=1}^n\mathbf{z}_s\mathbf{z}_s^\top-\mathbb{E}(\mathbf{z}_s\mathbf{z}_s^\top)\right\|\le \delta_1\Bigg\}+\mathbb{P}\left\{\left\|\frac{1}{n}\sum_{s=1}^n\mathbf{z}_s\mathbf{z}_s^\top-\mathbb{E}(\mathbf{z}_s\mathbf{z}_s^\top)\right\|>\delta_1\right\}\\
			&\qquad \le \mathbb{P}\Bigg\{\left\|\frac{1}{n}\sum_{s=1}^n\mathbf{z}_s\bar{\varepsilon}_{is}\right\|>\sqrt{\frac{C(n-d-2)\delta}{n}}\Bigg\}+\mathbb{P}\left\{\left\|\frac{1}{n}\sum_{s=1}^n\mathbf{z}_s\mathbf{z}_s^\top-\mathbb{E}(\mathbf{z}_s\mathbf{z}_s^\top)\right\|>\delta_1\right\},
		\end{align*}
		with $\delta_1>0$ being a small enough constant. 
		%Here we have used the condition $\sigma_{\text{min}}\{\mathbb{E}(\mathbf{z}_s\mathbf{z}_s^\top)\}>C$ and Lemma S8 of \citet{Zhou:2020} to get the last inequality. 
		In the above, the last inequality is due to the condition $\sigma_{\text{min}}\{\mathbb{E}(\mathbf{z}_s\mathbf{z}_s^\top)\}>C$ and Lemma S8 of [50, Zhou et al., 2021]. We conclude that $|\hat{\sigma}_i^2-\bar{\sigma}_{i}^2|$ has an exponential tail of the order $O(e^{-C_1n})$ by using the Chernoff bound and the fact that the product of two sub-Gaussian variables is sub-exponential [51, Vershynin, 2018]. Thus by the union bound and Condition (ix), we have $\max_i|\hat{\sigma}_i^2-\bar{\sigma}_i^2|=o_{\mathbb{P}}(1)$. Observing that
		\begin{align*}
			|\bar{\sigma}_i^2-\sigma_i^2|=&\left|\frac{1}{m}\left\{(m-2)\sigma_{i}^2+m^{-1}\sum_{i=1}^m\sigma_{i}^2\right\}-\sigma_i^2\right|=\left|\frac{-2}{m}\sigma_i^2-\frac{1}{m^2}\sum_{i=1}^m\sigma_i^2\right|=o_{\mathbb{P}_m}(1),
		\end{align*}
		we obtain the desired result that  $\max_i|\hat{\sigma}_i^2-\sigma_i^2|=o_{\mathbb{P}}(1)$.
	\end{proof}

	\begin{proof}[Proof of Lemma \ref{lem-mode}]
		We have
		\begin{align*}
			\sqrt{n}\tilde{\alpha}_i=\sqrt{n}\bar{\alpha}_i+\sqrt{n}\hat{\boldsymbol{\eta}}^\top n^{-1}\sum_{s=1}^n\mathbf{z}_s\bar{\varepsilon}_{is}=\sqrt{n}\alpha_i-\sqrt{n}\bar{\alpha}+U_i-U,
		\end{align*}
		where
		\begin{align*}
			U_i=\hat{\boldsymbol{\eta}}^\top\frac{1}{\sqrt{n}}\sum_{s=1}^n\mathbf{z}_s\varepsilon_{is},\quad U=\hat{\boldsymbol{\eta}}^\top\frac{1}{\sqrt{n}}\sum_{s=1}^n\mathbf{z}_s\left(\frac{1}{m}\sum_{i=1}^m\varepsilon_{is}\right),
		\end{align*}
		and $\hat{\boldsymbol{\eta}}$ is the first row of $(n^{-1}\sum_{s=1}^n\mathbf{z}_s\mathbf{z}_s^\top)^{-1}$. We first prove that $U=o_{\mathbb{P}}(1)$. Using similar arguments as in the proof of Lemma \ref{lem-rho}, we have $|\hat{\boldsymbol{\eta}}-\boldsymbol{\eta}|=o_{\mathbb{P}_n}(1)$, where $\boldsymbol{\eta}$ is the first row of $\{\mathbb{E}(\mathbf{z}_s\mathbf{z}_s^\top)\}^{-1}$. Under Conditions (i), (iii), and (v),  $\mathbf{z}_s(\sum_{i=1}^m\varepsilon_{is})/\sqrt{m}$ are conditionally i.i.d. given $\sigma_1,\dots,\sigma_m$. Thus,
		\begin{align*}
			&\mathbb{E}\left\{\mathbf{z}_s\left(\frac{1}{\sqrt{m}}\sum_{i=1}^m\varepsilon_{is}\right)\;\bigg|\;\sigma_1,...,\sigma_m\right\}=0,\\
			&\mathbb{E}\left\{(\mathbf{z}_s \odot \mathbf{z}_s) \left(\frac{1}{\sqrt{m}}\sum_{i=1}^m\varepsilon_{is}\right)^2\;\bigg|\;\sigma_1,\dots,\sigma_m\right\}=\frac{\mathbb{E}(\mathbf{z}_s \odot \mathbf{z}_s)}{m}\sum_{i=1}^m\sigma_i^2,
		\end{align*}
		where $\odot$ denotes the Hadamard product (element-wise product).
		The above implies that 
		\begin{align*}
			\frac{1}{\sqrt{n}}\sum_{s=1}^n\mathbf{z}_s\left(\frac{1}{\sqrt{m}}\sum_{i=1}^m\varepsilon_{is}\right)=O_{\mathbb{P}_n}(1)
		\end{align*}
		whenever $\sum_{i=1}^m\sigma_i^2/m<\infty$. Using Condition (ii), we have $\mathbb{P}(\sum_{i=1}^m\sigma_i^2/m<\infty)=1$. Thus $U=O_{\mathbb{P}}(m^{-1/2})$. Recall that 
		$$\widehat{\text{mode}}(\{X_i\}^{m}_{i=1})=	\argmax_{x\in\mathbb{R}}\frac{1}{mh}\sum_{i=1}^mK\left(\frac{x-X_i}{h}\right).$$
		It is not hard to see that
		$\widehat{\text{mode}}(\{X_i+a\}_{i=1}^m)=\widehat{\text{mode}}(\{X_i\}_{i=1}^m)+a,$
		for any $a$, which may be related to $m$ but is independent of $i$. We then have
		\begin{align*}
			&\widehat{\text{mode}}(\{\sqrt{n}\tilde{\alpha}_i\}_{i=1}^m)=\widehat{\text{mode}}(\{\sqrt{n}\alpha_i-\sqrt{n}\bar{\alpha}+U_i-U\}_{i=1}^m)=\widehat{\text{mode}}(\{\sqrt{n}\alpha_i+U_i\}_{i=1}^m)-\sqrt{n}\bar{\alpha}-U.
		\end{align*}
		Therefore, we only need to show that $\tilde{M}:=\widehat{\text{mode}}(\{\sqrt{n}\alpha_i+U_i\}_{i=1}^m)=o_{\mathbb{P}}(1)$. To this end, let 
		\begin{align*}
			f_{m,h}(x)=\frac{1}{mh}\sum_{i=1}^mK\left(\frac{x-(\sqrt{n}\alpha_i+U_i)}{h}\right).
		\end{align*}
		Given Condition (vi), we have that for large enough $n$,
		\begin{align*}
			|f_n(\tilde{M};\rho)-f_n(0;\rho)|&\le|f_n(\tilde{M};\rho)-f_{m,h}(\tilde{M})|+|f_{m,h}(\tilde{M})-f_n(0;\rho)|\\
			&=|f_n(\tilde{M};\rho)-f_{m,h}(\tilde{M})|+\left|\sup_{x\in\mathbb{R}}f_{m,h}(x)-\sup_{x\in\mathbb{R}}f_n(x;\rho)\right| \\
			&\le 2\sup_{x\in\mathbb{R}}|f_{m,h}(x)-f_n(x;\rho)|,
		\end{align*}
		and then it boils down to show that 
		\begin{align*}
			\sup_{x\in\mathbb{R}}|f_{m,h}(x)-f_n(x;\rho)|=o_{\mathbb{P}}(1).
		\end{align*}
		Note that
		\begin{align*}
			f_n(x;a)=\int\int\frac{1}{\sqrt{a}u}\phi\left(\frac{x-v}{\sqrt{a}u}\right)dF_{\sigma_i}(u)dF_{\sqrt{n}\alpha_i}(v)
		\end{align*}
		for any $a>0$, where $\phi(\cdot)$ denotes the density function of the standard normal distribution. It implies that $f_n(x;a)$ is uniformly continuous and bounded uniformly over $n$ and $a>C$. In other words, for any $\epsilon>0$, there exists a $\delta>0$ such that $\sup_{n,a>C,|x_1-x_2|<\delta}|f_n(x_1;a)-f_n(x_2;a)|<\epsilon$ and $\sup_{n,a>C,x\in\mathbb{R}}f_n(x;a)<\infty$. Besides, $\sup_{n,x\in\mathbb{R}}|f_n(x;\hat\rho)-f_n(x;\rho)|$ can be made arbitrarily small as long as $|\hat{\rho}-\rho|$ is small enough and $\rho>C>0$. Thus we have
		\begin{align*}
			&\mathbb{P}\left\{\sup_{x\in\mathbb{R}}\left|f_{m,h}(x)-f_n(x;\rho)\right|>\delta\right\}\\
			&\qquad \le \mathbb{P}\left\{\sup_{x\in\mathbb{R}}\left|f_{m,h}(x)-f_n(x;\rho)\right|>\delta,|\hat{\rho}-\rho|\le\delta_1\right\}+\mathbb{P}\left(|\hat{\rho}-\rho|>\delta_1\right)\\
			&\qquad \le \mathbb{P}\left\{\sup_{x\in\mathbb{R}}\left|f_{m,h}(x)-f_n(x;\hat{\rho})\right|>\delta/2,|\hat{\rho}-\rho|\le\delta_1\right\}+\mathbb{P}\left(|\hat{\rho}-\rho|>\delta_1\right)\\	
			&\qquad =  \int_{|u-\rho|\le\delta_1}\mathbb{P}\left\{\sup_{x\in\mathbb{R}}|f_{m,h}(x)-f_n(x;\hat{\rho})|>\delta/2\mid \hat{\rho}=u\right\}dF_{\hat{\rho}}(u)+\mathbb{P}\left(|\hat{\rho}-\rho|>\delta_1\right)
		\end{align*}
		for any $\delta>0$ and small enough $\delta_1>0$. Because $|\hat\rho-\rho|=o_{\mathbb{P}_n}(1)$ as shown in Lemma \ref{lem-rho}, our goal narrows down to proving that for any $\delta>0$ and $\epsilon>0$, there exists a $\xi>0$ such that when $m$ is large enough,
		\begin{align*} %\label{eq-kernel-consis}
			\sup_{n,\,|\hat{\rho}-\rho|\le\xi}\mathbb{P}\left\{\sup_{x\in\mathbb{R}}|f_{m,h}(x)-f_n(x;\hat{\rho})|>\delta\mid \hat{\rho}\right\}<\epsilon.
		\end{align*}
		To show the above displayed inequality holds for some $\xi>0$, it is sufficient to show 
		\begin{equation}\label{eqn:suffRhoSurro}
			\sup_{\substack{n,\,x\in\mathbb{R},\\ |\hat{\rho}-\rho|\le\xi}}\left|\mathbb{E}\{f_{m,h}(x)\mid\hat{\rho}\}-f_n(x;\hat{\rho})\right|<\epsilon
		\end{equation}
		and
		\begin{equation}\label{eqn:suffRhoCheby}
			\sup_{n, \, |\hat{\rho}-\rho|\le\xi}\mathbb{E}\left[\sup_{x\in\mathbb{R}}\left|f_{m,h}(x)-\mathbb{E}\{f_{m,h}(x)\mid\hat{\rho}\}\right|^2\mid\hat{\rho}\right]<\epsilon
		\end{equation}
		are fulfilled for some small enough $\xi >0$. %and $\nu>0$ be small enough constants. 
		
		For \eqref{eqn:suffRhoSurro}, using $\int_{-\infty}^{\infty}K(y)dy=1$ with $K(y) \geq 0$, we observe that 
		\allowdisplaybreaks
		\begin{align}\label{eqn:suffRhoSurroPar}
			&\sup_{\substack{n,\,x\in\mathbb{R},\\ |\hat{\rho}-\rho|\le\xi}}\left|\mathbb{E}\{f_{m,h}(x)\mid\hat{\rho}\}-f_n(x;\hat{\rho})\right| \nonumber\\
			&\qquad\qquad =\sup_{\substack{n,\,x\in\mathbb{R},\\ |\hat{\rho}-\rho|\le\xi}}\left|\int_{-\infty}^{\infty}\frac{1}{h}K\left(\frac{x-y}{h}\right)f_n(y;\hat{\rho})dy-f_n(x;\hat{\rho})\right| \nonumber\\
			&\qquad\qquad =\sup_{\substack{n,\,x\in\mathbb{R},\\ |\hat{\rho}-\rho|\le\xi}}\left|\int_{-\infty}^{\infty}\frac{1}{h}K\left(\frac{y}{h}\right)\left\{f_n(x-y;\hat{\rho})-f_n(x;\hat{\rho})\right\}dy\right| \nonumber\\
			&\qquad\qquad =\sup_{\substack{n,\,x\in\mathbb{R},\\ |\hat{\rho}-\rho|\le\xi}}\left|\int_{|y|\le\nu}\frac{1}{h}K\left(\frac{y}{h}\right)\left\{f_n(x-y;\hat{\rho})-f_n(x;\hat{\rho})\right\}dy\right| \nonumber\\
			&\qquad\qquad\qquad\qquad + \sup_{\substack{n,\,x\in\mathbb{R},\\ |\hat{\rho}-\rho|\le\xi}}\left|\int_{|y|>\nu}\frac{1}{h}K\left(\frac{y}{h}\right)\left\{f_n(x-y;\hat{\rho})-f_n(x;\hat{\rho})\right\}dy\right| \nonumber\\
			&\qquad\qquad \le\sup_{\substack{n,\,|\hat{\rho}-\rho|\le\xi,\\ x\in\mathbb{R},\,|y|\le\nu}}|f_n(x-y;\hat{\rho})-f_n(x;\hat{\rho})|\int_{|u|\le\nu/h}K(u)du \nonumber\\
			&\qquad\qquad\qquad\qquad +\sup_{\substack{n,\,x\in\mathbb{R},\\ |\hat{\rho}-\rho|\le\xi}}f_n(x;\hat{\rho})\int_{|u|>\nu/h}K(u)du.%<\epsilon,
		\end{align}
		Due to the condition that $f_n(x;a)$ is uniformly continuous and upper bounded uniformly over $n$, \eqref{eqn:suffRhoSurroPar} is less than $\epsilon$ for some small enough $\xi>0$ and $\nu>0$ (depending on $\epsilon$). It completes \eqref{eqn:suffRhoSurro} for some $\xi$.
		
		For \eqref{eqn:suffRhoCheby}, note that $U_i$'s have the same distribution as $\sqrt{\hat{\rho}}\varepsilon_{is}$ and are independent given $\hat{\rho}$. Let $X_i=\sqrt{n}\alpha_i+U_i$. Define
		\begin{align*}
			\varphi_m(u)=m^{-1}\sum_{i=1}^me^{\imath uX_i}.
		\end{align*}
		The inverse Fourier transformation provides
		\begin{align*}
			K(y)=(2\pi)^{-1}\int_{-\infty}^{\infty}k(u)e^{\imath uy}du.
		\end{align*}
		After plugging this expression into the definition of $f_{m,h}$, it shows that
		\begin{align*}
			f_{m,h}(x)&=\frac{1}{mh}\sum_{i=1}^mK\left(\frac{x-X_i}{h}\right)\\
			&=(2\pi m h)^{-1}\sum_{i=1}^m\int_{-\infty}^{\infty}k(u)e^{\imath u\frac{x-X_i}{h}}du\\
			&=(2\pi m)^{-1}\sum_{i=1}^m\int_{-\infty}^{\infty}k(hu)e^{\imath u(x-X_i)}du\\
			&=(2\pi)^{-1}\int_{-\infty}^{\infty}e^{-\imath ux}k(hu)\varphi_m(u)du,
		\end{align*}
		where the last equality is because $k(u)$ is even. This result further implies
		\begin{align*}
			\sup_{x\in\mathbb{R}}\left|f_{m,h}(x)-\mathbb{E}\{f_{m,h}(x)\mid\hat{\rho}\}\right|\le(2\pi)^{-1}\int_{-\infty}^{\infty}|k(hu)|\,|\varphi_m(u)-\mathbb{E}\{\varphi_m(u)\mid\hat{\rho}\}|\, du.
		\end{align*}
		Using the above inequality, the Cauchy-Schwartz inequality, and Euler's identity (i.e., $|e^{\imath x}| = 1$), it shows that the left hand side of \eqref{eqn:suffRhoCheby} satisfies
		\begin{align*}
			&\sup_{n,\,|\hat\rho-\rho|\le\xi}\mathbb{E}\left[\sup_{x\in\mathbb{R}}\left|f_{m,h}(x)-\mathbb{E}\{f_{m,h}(x)\mid\hat{\rho}\}\right|^2\mid\hat{\rho}\right]\\
			&\qquad\qquad \le\sup_{n,\,|\hat\rho-\rho|\le\xi}\mathbb{E}\left(\left[(2\pi)^{-1}\int_{-\infty}^{\infty}|k(hu)|\,|\varphi_m(u)-\mathbb{E}\{\varphi_m(u)\mid\hat{\rho}\}|\, du\right]^2 \mid\hat{\rho}\right)\\
			&\qquad\qquad \le\sup_{n,\,|\hat\rho-\rho|\le\xi}(2\pi)^{-2}\int_{-\infty}^{\infty}|k(hu)|du\int_{-\infty}^{\infty}|k(hu)|\, \mathbb{E}\left[|\varphi_m(u)-\mathbb{E}\{\varphi_m(u)\mid\hat{\rho}\}|^2\mid\hat{\rho}\right]du\\
			&\qquad\qquad =\sup_{n,\,|\hat\rho-\rho|\le\xi}(2\pi)^{-2}m^{-1}\int_{-\infty}^{\infty}|k(hu)|du\int_{-\infty}^{\infty}|k(hu)|\, \mathbb{E}\left[|e^{\imath uX_i}-\mathbb{E}\{e^{\imath uX_i}\mid\hat{\rho}\}|^2\mid\hat{\rho}\right]du\\
			&\qquad\qquad \le \pi^{-2}m^{-1}h^{-2}\left\{\int_{-\infty}^{\infty}|k(u)|du\right\}^2\to 0,
		\end{align*}
		where the result of converging to $0$ is due to Conditions (vii) and (viii). Therefore, \eqref{eqn:suffRhoCheby} is satisfied, which completes the proof.
	\end{proof}

	\begin{proof}[Proof of Lemma \ref{lem-empirical}]
		In the following, we focus on showing  $\sup_{0<t<t_0}\left|m^{-1}S_{m,n}(t)-S_{\infty,n}(t)\right|=o_{\mathbb{P}}(1) $. The proof of the second statement can be obtained by similar arguments, and thus is omitted.
		
		Let $$S_{m,n}^{-}(t)=\sum_{i=1}^m\mathbb{I}\left(\sqrt{n}\hat{\alpha}_i/\sqrt{\hat{\rho}\hat{\sigma}_i^2}<-t\right).$$ 
		The goal is to show
		\begin{align*}
			\sup_{0<t<t_0}\left|\frac{1}{m}S_{m,n}^{-}(t)-\mathbb{P}\left(\mathcal{E}+\sqrt{n}\alpha_i/\sqrt{\rho\sigma_i^2}<-t\right)\right|=o_{\mathbb{P}}(1).
		\end{align*}
		Recall in the proof of Lemma \ref{lem-mode}, we have
		\begin{align*}
			&\sqrt{n}\hat{\alpha}_i=\sqrt{n}(\tilde{\alpha}_i+\tilde{\alpha})=\sqrt{n}\alpha_i+\sqrt{n}(\tilde{\alpha}-\bar{\alpha})+U_i-U,
		\end{align*}
		where $U_i/\sqrt{\hat{\rho}\sigma_i^2}\sim^{\text{i.i.d.}}N(0,1)$, $U=o_{\mathbb{P}}(1)$, and $\sqrt{n}(\tilde{\alpha}-\bar{\alpha})=o_{\mathbb{P}}(1)$. These results imply that 
		\begin{align*}
			&\frac{1}{m}S_{m,n}^{-}(t)
			=\frac{1}{m}\sum_{i=1}^m\mathbb{I}\left\{\frac{U_i}{\sqrt{\hat{\rho}\sigma_i^2}}+\frac{\alpha_i}{\sqrt{\hat{\rho}\sigma_i^2/n}}<-t\frac{\hat{\sigma}_i}{\sigma_i}+\frac{U-\sqrt{n}(\tilde{\alpha}-\bar{\alpha})}{\sqrt{\hat{\rho}\sigma_i^2}}\right\},
		\end{align*}
		and
		\begin{align*}
			&\mathbb{P}\left\{\sup_{0<t<t_0}\left|\frac{1}{m}S_{m,n}^{-}(t)-\mathbb{P}\left(\mathcal{E}+\alpha_i/\sqrt{\rho\sigma_i^2/n}<-t\right)\right|>\delta\right\}\\
			& \le \mathbb{P}\Bigg[\sup_{0<t<t_0}\Bigg|\frac{1}{m}\sum_{i=1}^m\mathbb{I}\Bigg\{\frac{U_i}{\sqrt{\hat{\rho}\sigma_i^2}}+\frac{\alpha_i}{\sqrt{\hat{\rho}\sigma_i^2/n}}<-t\frac{\hat{\sigma}_i}{\sigma_i}+\frac{U-\sqrt{n}(\tilde{\alpha}-\bar{\alpha})}{\sqrt{\hat{\rho}\sigma_i^2}}\Bigg\}\\
			&\qquad\qquad-\mathbb{P}\left(\mathcal{E}+\frac{\alpha_i}{\sqrt{\rho\sigma_i^2/n}}<-t\right)\Bigg|>\delta,\, \sup_i\left|\frac{\hat{\sigma}_i}{\sigma_i}-1\right|\le\delta_1,\, \sup_i\Bigg|\frac{U-\sqrt{n}(\tilde{\alpha}-\bar{\alpha})}{\sqrt{\hat{\rho}\sigma_i^2}}\Bigg|\le\delta_2\Bigg]\\
			&\qquad+\mathbb{P}\left(\sup_i\left|\frac{\hat{\sigma}_i}{\sigma_i}-1\right|>\delta_1\right)+\mathbb{P}\Bigg\{\sup_i\Bigg|\frac{U-\sqrt{n}(\tilde{\alpha}-\bar{\alpha})}{\sqrt{\hat{\rho}\sigma_i^2}}\Bigg|>\delta_2\Bigg\}\\
			& \le \mathbb{P}\Bigg\{\sup_{0<t<t_0}\Bigg|\frac{1}{m}\sum_{i=1}^m\mathbb{I}\Bigg(\frac{U_i}{\sqrt{\hat{\rho}\sigma_i^2}}+\frac{\alpha_i}{\sqrt{\hat{\rho}\sigma_i^2/n}}<-t-t\delta_1-\delta_2\Bigg)-\mathbb{P}\Bigg(\mathcal{E}+\frac{\alpha_i}{\sqrt{\rho\sigma_i^2/n}}<-t\Bigg)\Bigg|>\delta\Bigg\}\\
			&\qquad+\mathbb{P}\Bigg\{\sup_{0<t<t_0}\Bigg|\frac{1}{m}\sum_{i=1}^m\mathbb{I}\Bigg(\frac{U_i}{\sqrt{\hat{\rho}\sigma_i^2}}+\frac{\alpha_i}{\sqrt{\hat{\rho}\sigma_i^2/n}}<-t+t\delta_1+\delta_2\Bigg)-\mathbb{P}\Bigg(\mathcal{E}+\frac{\alpha_i}{\sqrt{\rho\sigma_i^2/n}}<-t\Bigg)\Bigg|>\delta\Bigg\}\\
			&\qquad+o(1)
		\end{align*}
		for any positive constants $\delta$, $\delta_1$, and $\delta_2$, where the last step is due to $\rho>C$, $\sigma_i>C$, and the results from Lemmas \ref{lem-rho}--\ref{lem-mode}. Thus we only need to show that for any $\delta>0$ and $\epsilon>0$, there exist $\xi>0$, $\delta_1\neq 0$ and $\delta_2\neq 0$ such that for large enough $m$,
		\begin{align*}
			\begin{split}
				&\sup_{n,\,|\hat{\rho}-\rho|<\xi}\mathbb{P}\Bigg\{\sup_{0<t<t_0}\Bigg|\frac{1}{m}\sum_{i=1}^m\mathbb{I}\Bigg(\frac{U_i}{\sqrt{\hat{\rho}\sigma_i^2}}+\frac{\alpha_i}{\sqrt{\hat{\rho}\sigma_i^2/n}}<-t+t\delta_1+\delta_2\Bigg)\\
				&\qquad\qquad\qquad\qquad\;\;\;-\mathbb{P}\Bigg(\mathcal{E}+\frac{\alpha_i}{\sqrt{\rho\sigma_i^2/n}}<-t\Bigg)\Bigg|>\delta\;\bigg|\;\hat{\rho}\Bigg\}<\epsilon,
			\end{split}
		\end{align*}
		or sufficiently,
		\begin{align}
			&\sup_{n,\,|\hat{\rho}-\rho|<\xi}\mathbb{P}\Bigg\{\sup_{0<t<t_0}\Bigg|\frac{1}{m}\sum_{i=1}^m\mathbb{I}\Bigg(\frac{U_i}{\sqrt{\hat{\rho}\sigma_i^2}}+\frac{\alpha_i}{\sqrt{\hat{\rho}\sigma_i^2/n}}<-t+t\delta_1+\delta_2\Bigg) \nonumber\\
			&\qquad\qquad\qquad\qquad\;\;\;-\mathbb{P}\Bigg(\mathcal{E}+\frac{\alpha_i}{\sqrt{\hat{\rho}\sigma_i^2/n}}<-t+t\delta_1+\delta_2\;\bigg|\;\hat\rho\Bigg)\Bigg|>\delta\;\bigg|\;\hat{\rho}\Bigg\}<\epsilon, \label{eqn:SmnSuffPar1}\\
			&\sup_{\substack{n,\,|\hat{\rho}-\rho|<\xi,\\0<t<t_0}}\Bigg|\mathbb{P}\Bigg(\mathcal{E}+\frac{\alpha_i}{\sqrt{\hat{\rho}\sigma_i^2/n}}<-t+t\delta_1+\delta_2\;\bigg|\;\hat{\rho}\Bigg)-\mathbb{P}\Bigg(\mathcal{E}+\frac{\alpha_i}{\sqrt{\hat{\rho}\sigma_i^2/n}}<-t\;\bigg|\;\hat{\rho}\Bigg)\Bigg|<\epsilon, \label{eqn:SmnSuffPar2}
		\end{align}
		and
		\begin{align}\label{eqn:SmnSuffPar3}
			&\sup_{\substack{n,\,|\hat{\rho}-\rho|<\xi,\\0<t<t_0}}\Bigg|\mathbb{P}\Bigg(\mathcal{E}+\frac{\alpha_i}{\sqrt{\hat{\rho}\sigma_i^2/n}}<-t\;\bigg|\;\hat{\rho}\Bigg)-\mathbb{P}\Bigg(\mathcal{E}+\frac{\alpha_i}{\sqrt{\rho\sigma_i^2/n}}<-t\Bigg)\Bigg|<\epsilon.
		\end{align}
		First, \eqref{eqn:SmnSuffPar1} is a direct result of applying the Glivenko-Cantelli theorem [52, Wainwright, 2019]. %\citep{Zhangj:2009}. 
		For \eqref{eqn:SmnSuffPar2}, we note that the cumulative distribution function of $\mathcal{E}+\alpha_i/\sqrt{a\sigma_i^2/n}$ for any $a>0$, denoted by $G_n(\cdot;a)$, can be expressed as
		\begin{align*}
			G_n(x;a)=\int_{-\infty}^{\infty}\Phi\left(x-u\right)dF_{\alpha_i/\sqrt{a\sigma_i^2/n}}(u)=\int_{-\infty}^{\infty}\Phi\left(x-\sqrt{\frac{\rho}{a}}u\right)dF_{\alpha_i/\sqrt{\rho\sigma_i^2/n}}(u),
		\end{align*}
		where $\Phi(\cdot)$ represents the cumulative distribution function of the standard normal distribution. Thus $G_n(x;a)$ is equicontinuous uniformly  over $n$ and $a>0$. 
		In other words, for any $\epsilon>0$, there exists a $\delta>0$ such that $$\sup_{\substack{n,\,a>0,\\|x_1-x_2|<\delta}}|G_n(x_1;a)-G_n(x_2;a)|<\epsilon,$$ which verifies the \eqref{eqn:SmnSuffPar2}. 
		Further, $$\sup_{\substack{n,\,|\hat{\rho}-\rho|<\xi,\\|x|<t_0}}|G_n(x;\hat\rho)-G_n(x;\rho)|$$ can be arbitrarily small as long as $\xi$ is small enough, which confirms \eqref{eqn:SmnSuffPar3}.
	\end{proof}

	\begin{proof}[Proof of Theorem 1]
		Observe that 
		\begin{align*}
			\left|\widehat{\text{FDP}}(t)-\frac{2F_{n-d-2}(-t)}{S_{\infty,n}(t)}\right|=\left|2F_{n-d-2}(-t)\left\{\frac{1}{S_{m,n}(t)/m}-\frac{1}{S_{\infty,n}(t)}\right\}\right|,
		\end{align*}
		where $S_{m,n}(t)$ and $S_{\infty,n}(t)$ are defined in Lemma \ref{lem-empirical}. Together with Lemma \ref{lem-empirical} and Condition (x), we deduce that there exists some $t_0$ such that $t^*<t_0$ for large enough $n$, 
		\begin{align*}
			&\sup_{0<t<t_0}\left|\frac{V_{m,n}(t)}{1\vee S_{m,n}(t)}-\frac{m_0}{m}\frac{2F_{n-d-2}(-t)}{S_{\infty,n}(t)}\right|\\
			=&\sup_{0<t<t_0}\left|\frac{V_{m,n}(t)}{1\vee S_{m,n}(t)}-\frac{2F_{n-d-2}(-t)}{\{1\vee S_{m,n}(t)\}/m_0}+\frac{2F_{n-d-2}(-t)}{\{1\vee S_{m,n}(t)\}/m_0}-\frac{m_0}{m}\frac{2F_{n-d-2}(-t)}{S_{\infty,n}(t)}\right|\\
			\le&\sup_{0<t<t_0}\left|\frac{m_0^{-1}V_{m,n}(t)-2F_{n-d-2}(-t)}{\{1\vee S_{m,n}(t)\}/m_0}\right|+\sup_{0<t<t_0}\left|\frac{2m_0F_{n-d-2}(-t)}{m}\left[\frac{1}{\{1\vee S_{m,n}(t)\}/m}-\frac{1}{S_{\infty,n}(t)}\right]\right|\\
			=&o_{\mathbb{P}}(1),
		\end{align*}
		and
		\begin{align*}
			\sup_{0<t<t_0}	\left|\widehat{\text{FDP}}(t)-\frac{2F_{n-d-2}(-t)}{S_{\infty,n}(t)}\right|=\sup_{0<t<t_0}\left|2F_{n-d-2}(-t)\left\{\frac{1}{S_{m,n}(t)/m}-\frac{1}{S_{\infty,n}(t)}\right\}\right|=o_{\mathbb{P}}(1).
		\end{align*}
		Therefore, we have
		\begin{align*}
			\frac{V_{m,n}(t^*)}{1\vee S_{m,n}(t^*)}&\le \frac{V_{m,n}(t^*)}{1\vee S_{m,n}(t^*)}-\frac{m_0}{m}\frac{2F_{n-d-2}(-t^*)}{S_{\infty,n}(t^*)}+\frac{2F_{n-d-2}(-t^*)}{S_{\infty,n}(t^*)}-\widehat{\text{FDP}}(t^*)+\widehat{\text{FDP}}(t^*) \\
			&\le q+o_{\mathbb{P}}(1).
		\end{align*}
		The conclusion follows by using Lemma 8.3 of [53, Cao et al., 2021].
	\end{proof}

\setcounter{section}{0}
\renewcommand{\thesection}{S\arabic{section}}
\setcounter{subsection}{0}
\renewcommand{\thesubsection}{S\arabic{subsection}}
\setcounter{equation}{0}
\renewcommand{\theequation}{S\arabic{equation}}
\setcounter{figure}{0}
\renewcommand{\thefigure}{S\arabic{figure}}
\setcounter{table}{0}
\renewcommand{\thetable}{S\arabic{table}}
\setcounter{lemma}{0}
\renewcommand{\thelemma}{S\arabic{lemma}}

\addtolength{\textheight}{-0.7in}%
\renewcommand{\figurename}{Fig.}
	\newpage
\spacingset{1.5} % DON'T change the spacing!
{\huge Supplementary figures for ``LinDA: linear models for differential abundance analysis of microbiome compositional data"}

	\section{Additional main comparisons of numerical studies}\label{sec-supp-simu}
Fig. \ref{fig-S6C0-proposed} and Fig. \ref{fig-S0C0-proposed} compare the proposed method LinDA with different zero-handling approaches under settings S6C0 and S0C0. Fig. \ref{fig-S0C0-Maaslin2LinDA} 
depicts the results of LinDA, CLR-OLS and MaAsLin2 with different normalization approaches under setting S0C0. 
%Figure \ref{fig-S6C0-ancombc} compares the ANCOM-BC disabling and enabling zero treatment for setting S6C0. 
Fig. \ref{fig-S0C1}--\ref{fig-S6C0}, \ref{fig-S7C0} and \ref{fig-S81C0}--\ref{fig-S82C0} show the results of settings S0C1, S0C2, S1C0, S2C0, S4C0, S5C0, S6C0, S7C0, S8.1C0, and S8.2C0, respectively. The comparison between disabling and enabling zero treatment of the ANCOM-BC method is depicted in Fig. \ref{fig-S6C0-ancombc} under setting S6C0. 
Fig. \ref{fig-S0C0-strong} shows the results of setting S0C0 with stronger compositional effects.

\begin{figure}
	\begin{subfigure}[b]{1\textwidth}
		\centering
		\includegraphics[scale=0.5]{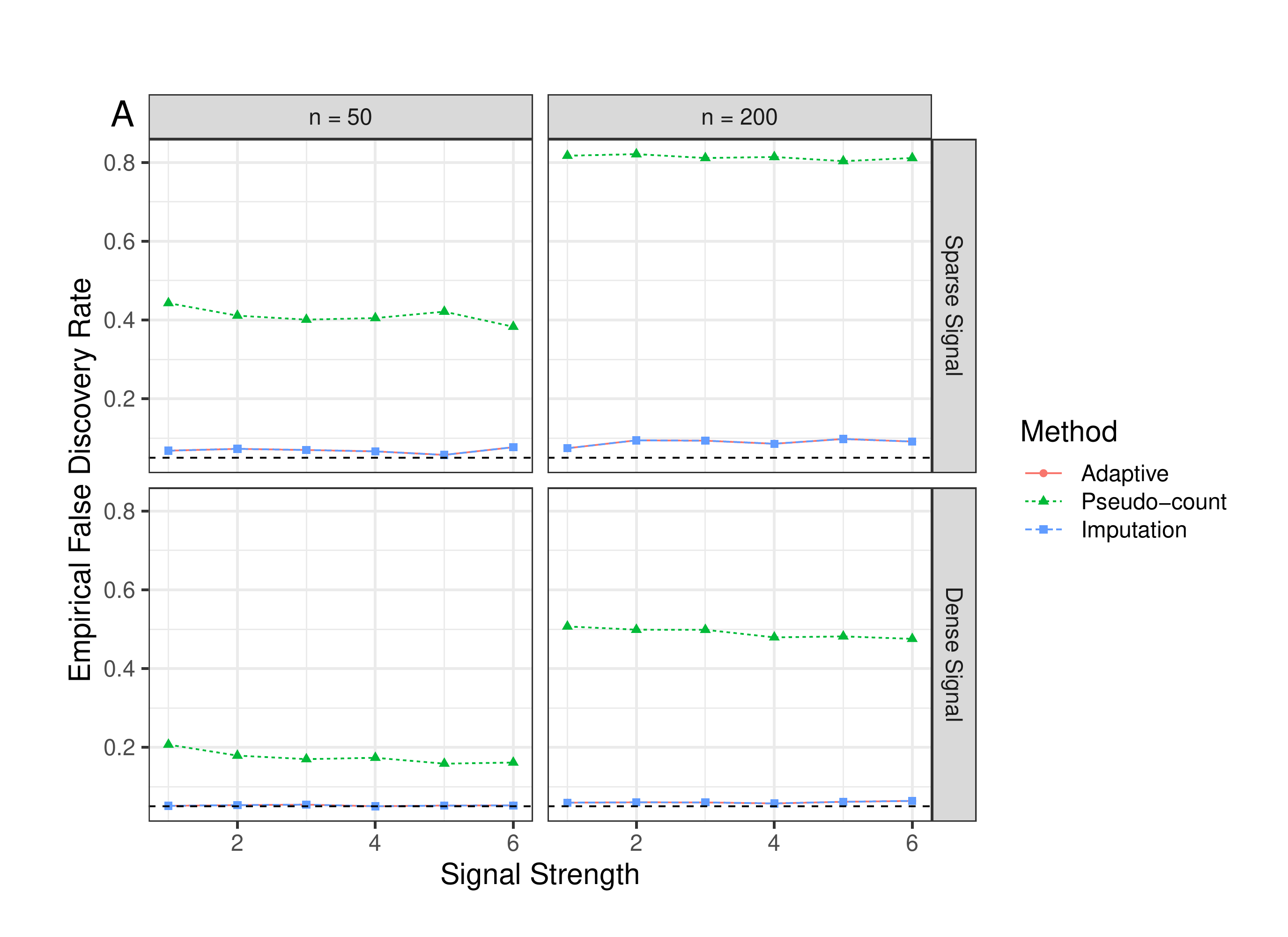}
	\end{subfigure}
	\begin{subfigure}[b]{1\textwidth}
		\centering
		\includegraphics[scale=0.5]{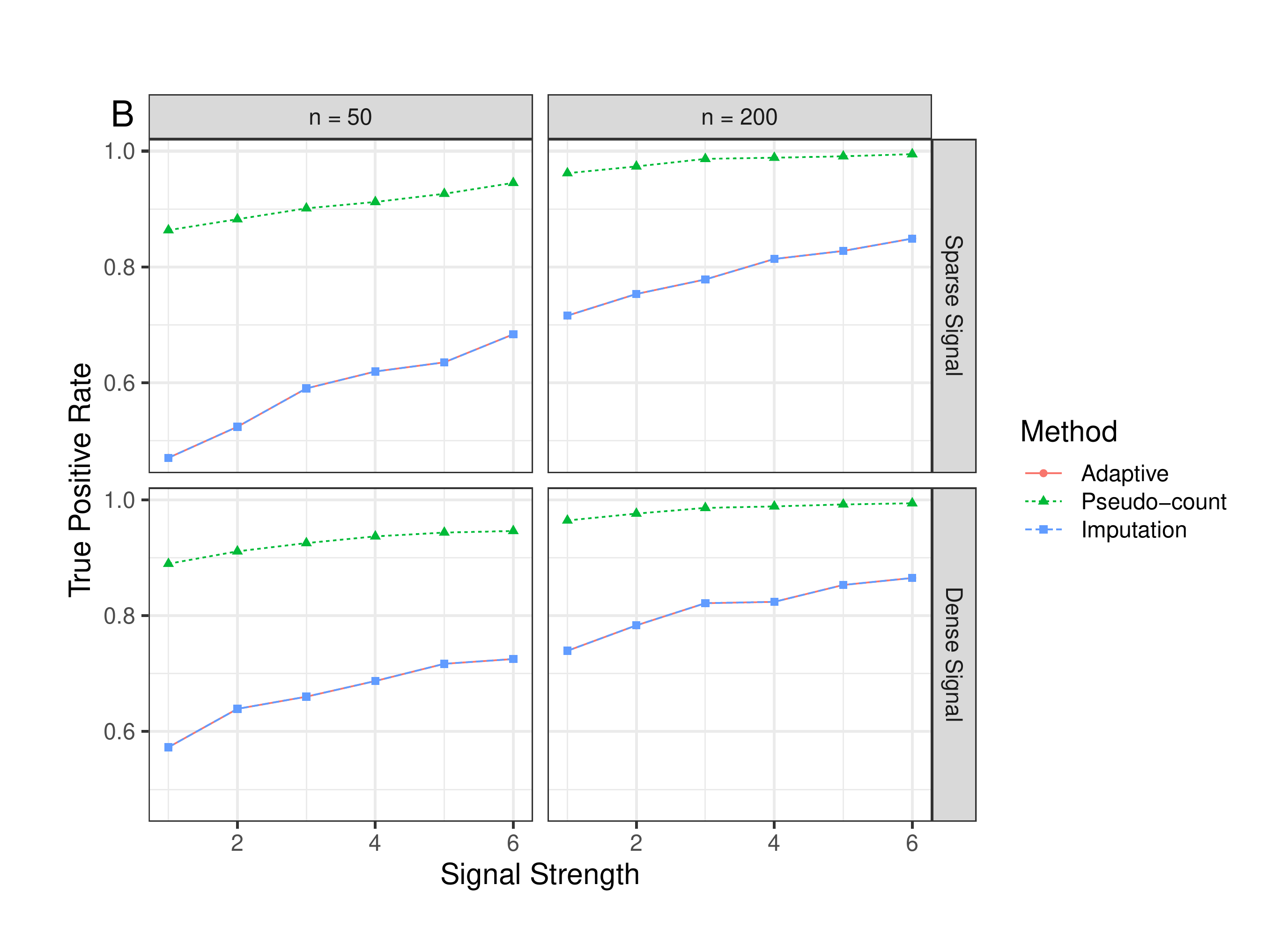}
	\end{subfigure}
	\caption{Performance of LinDA with different zero-handling approaches (S6C0: 10-fold difference in library size, a binary covariate). Empirical false discovery rate (A) and true positive rates (B) were averaged over 100 simulation runs. The dashed horizontal line (A) indicates the target FDR level of 0.05. Note that the red and blue lines are overlapped as the covariate and sequencing depth are significantly correlated.}
	\label{fig-S6C0-proposed}
\end{figure}

\begin{figure}
	\begin{subfigure}[b]{1\textwidth}
		\centering
		\includegraphics[scale=0.5]{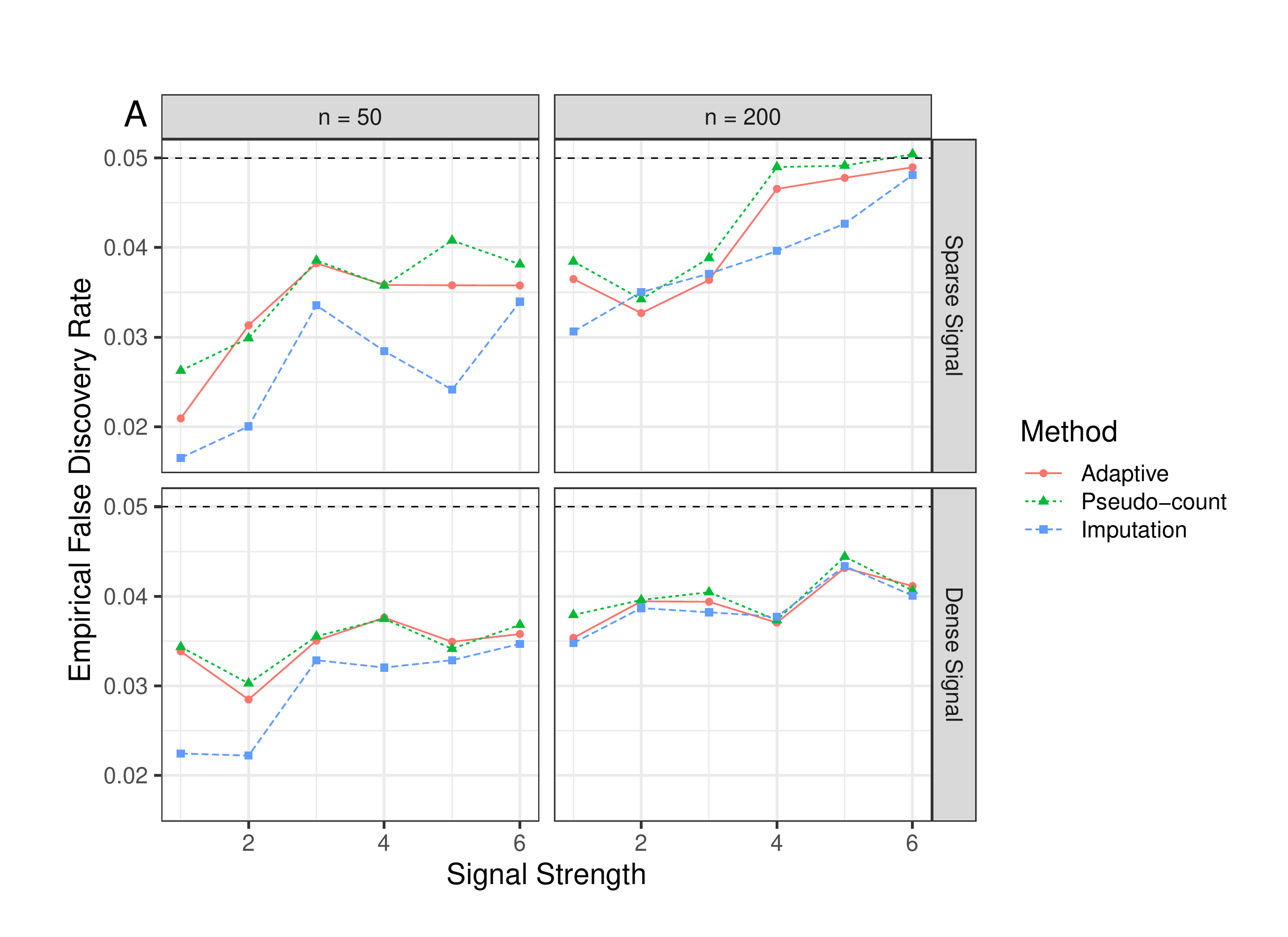}
	\end{subfigure}
	\begin{subfigure}[b]{1\textwidth}
		\centering
		\includegraphics[scale=0.5]{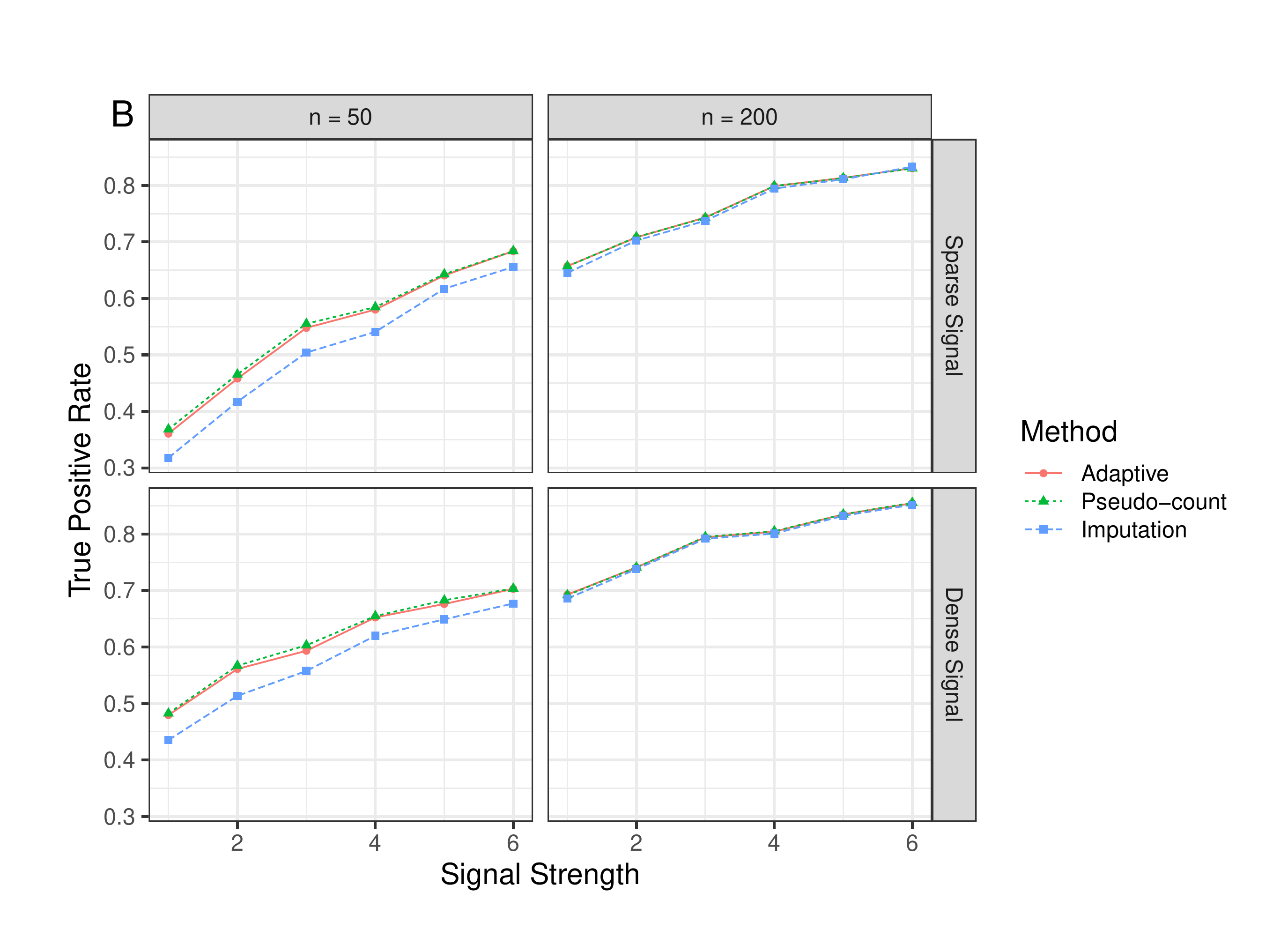}
	\end{subfigure}
	\caption{Performance of LinDA with different zero-handling approaches (S0C0: log normal abundance distribution, a binary covariate). Empirical false discovery rate (A) and true positive rates (B) were averaged over 100 simulation runs. The dashed horizontal line (A) indicates the target FDR level of 0.05.}
	\label{fig-S0C0-proposed}
\end{figure}

\begin{figure}
	\begin{subfigure}[b]{1\textwidth}
		\centering
		\includegraphics[scale=0.5]{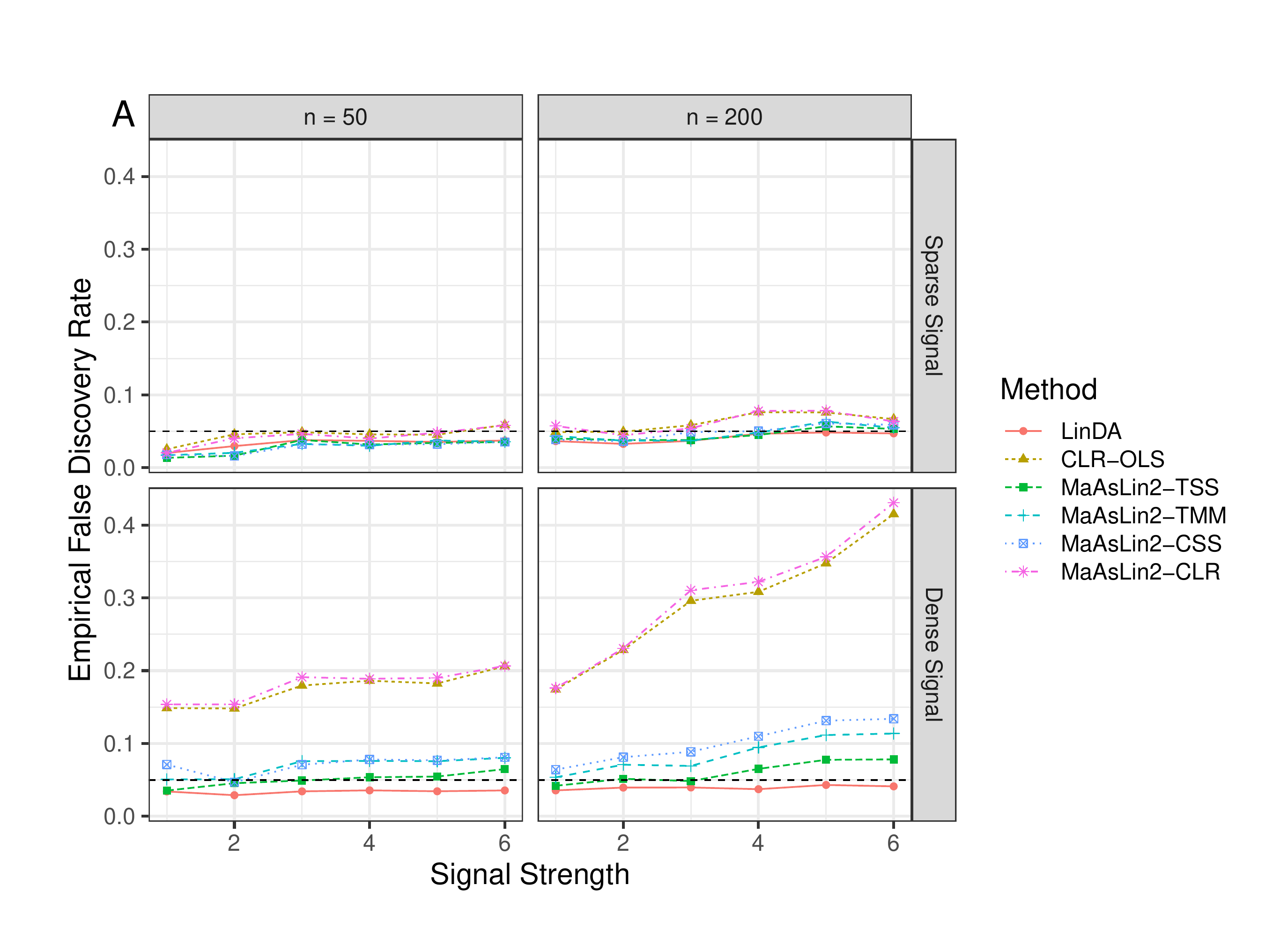}
	\end{subfigure}
	\begin{subfigure}[b]{1\textwidth}
		\centering
		\includegraphics[scale=0.5]{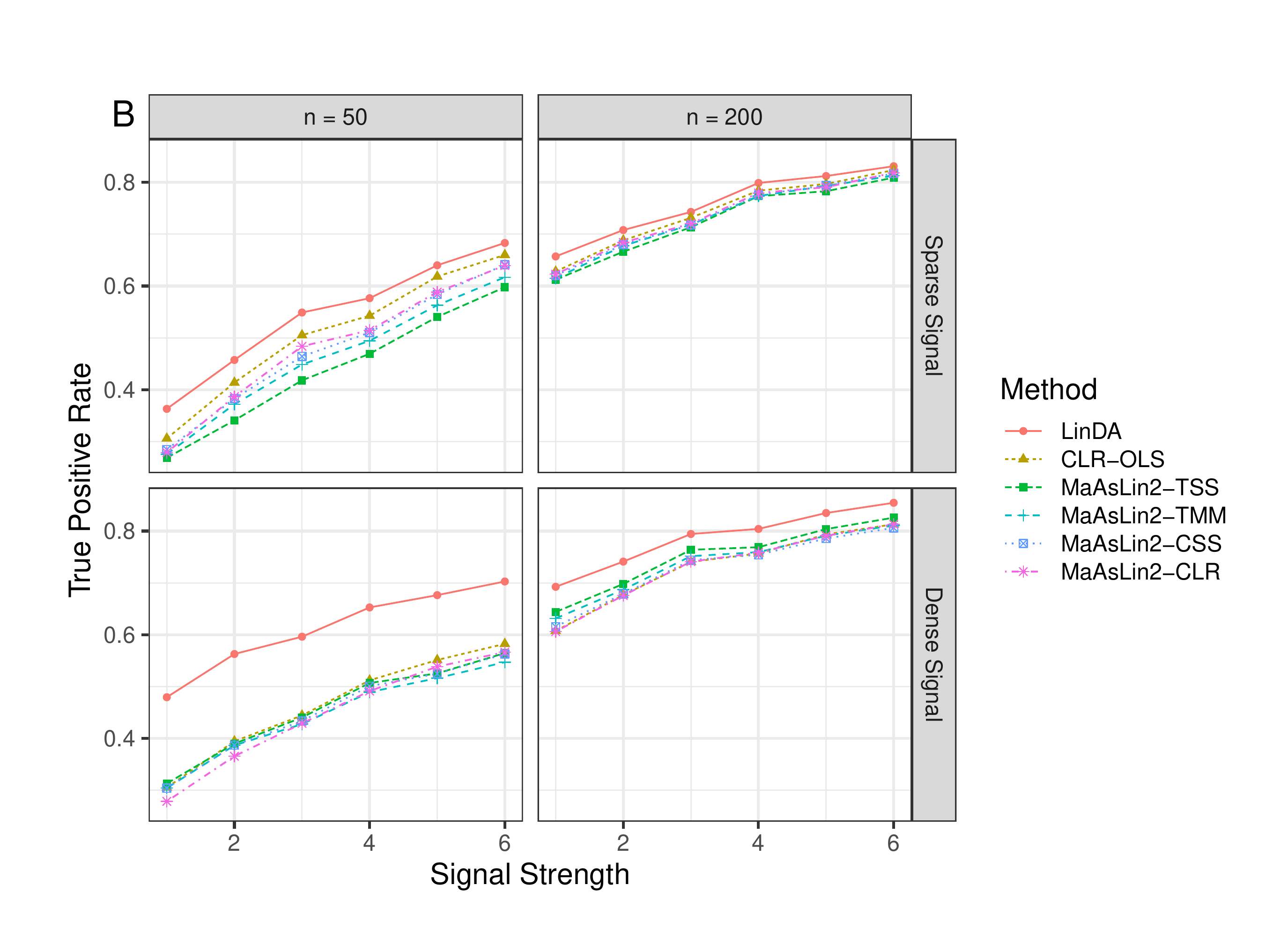}
	\end{subfigure}
	\caption{Performance comparison between LinDA and MaAsLin2 (S0C0: log normal abundance distribution, a binary covariate). Empirical false discovery rate (A) and true positive rates (B) were averaged over 100 simulation runs. The dashed horizontal line (A) indicates the target FDR level of 0.05.}
	\label{fig-S0C0-Maaslin2LinDA}
\end{figure}

\begin{figure}
	\begin{subfigure}[b]{1\textwidth}
		\centering
		\includegraphics[scale=0.5]{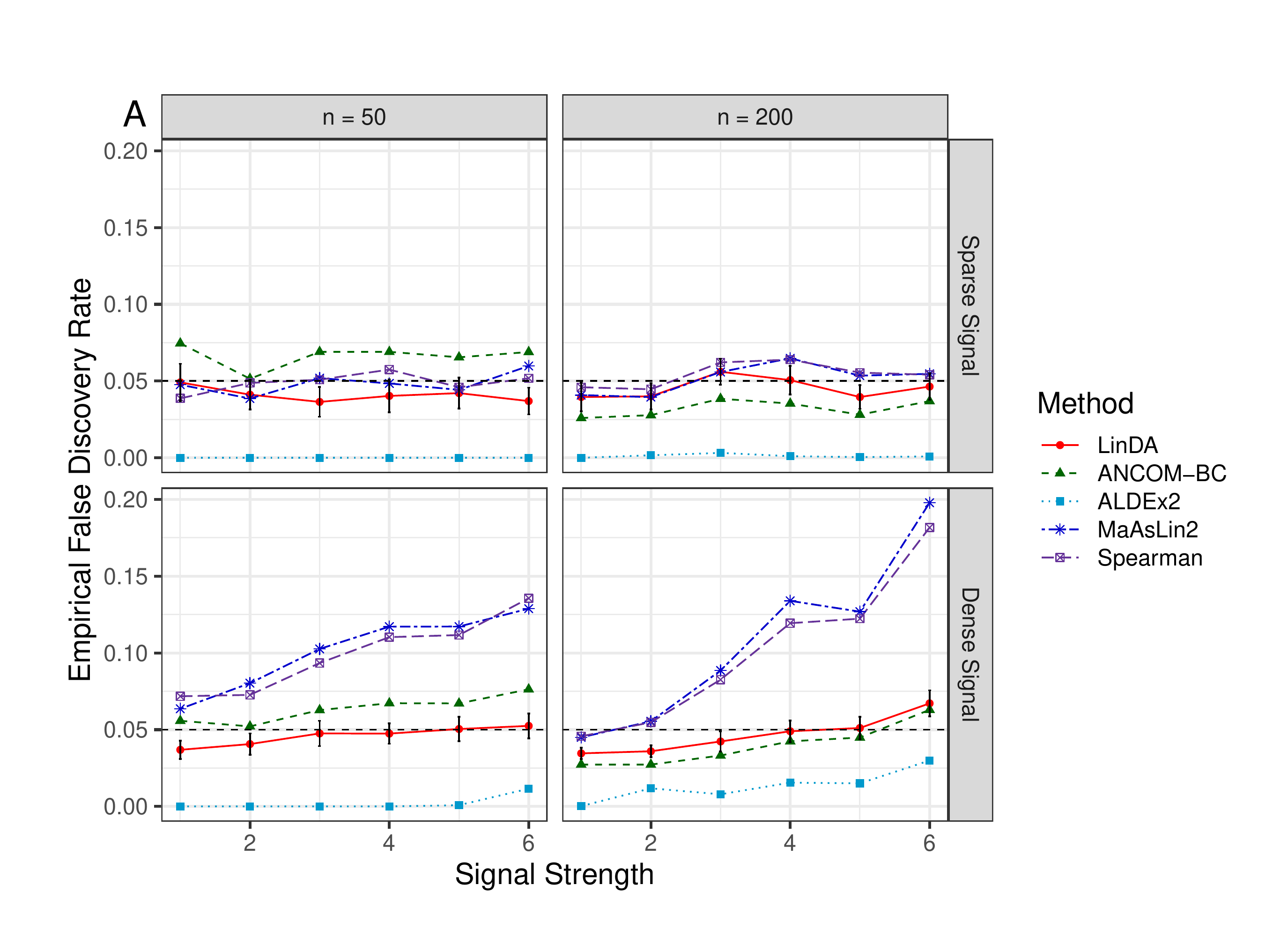}
	\end{subfigure}
	\begin{subfigure}[b]{1\textwidth}
		\centering
		\includegraphics[scale=0.5]{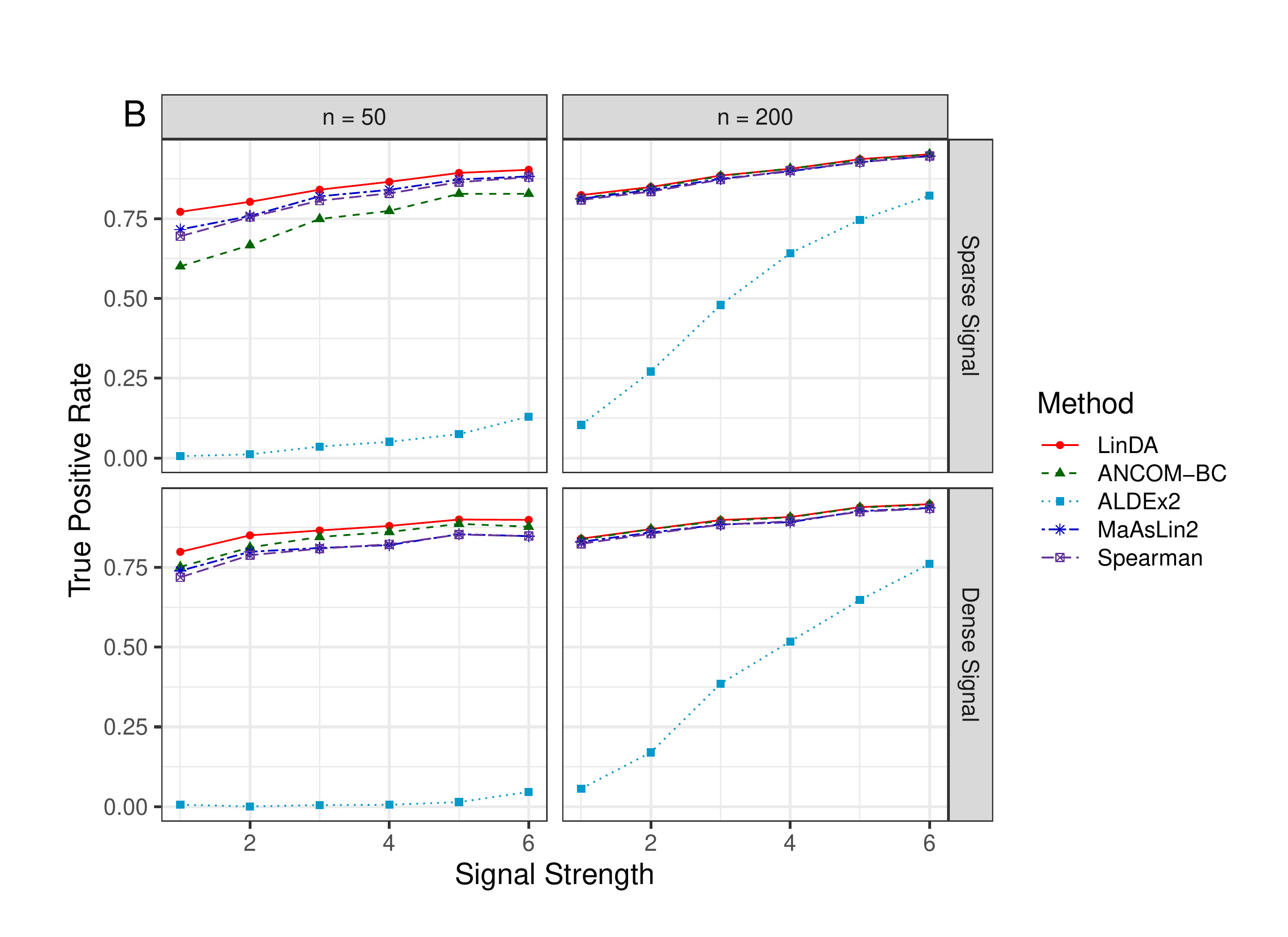}
	\end{subfigure}
	\caption{Performance comparison (S0C1: log normal abundance distribution, a continuous covariate). Empirical false discovery rate (A) and true positive rates (B) were averaged over 100 simulation runs. Error bars (A) represent the 95\% CIs of the method LinDA and the dashed horizontal line indicates the target FDR level of 0.05.}
	\label{fig-S0C1}
\end{figure}

\begin{figure}
	\begin{subfigure}[b]{1\textwidth}
		\centering
		\includegraphics[scale=0.5]{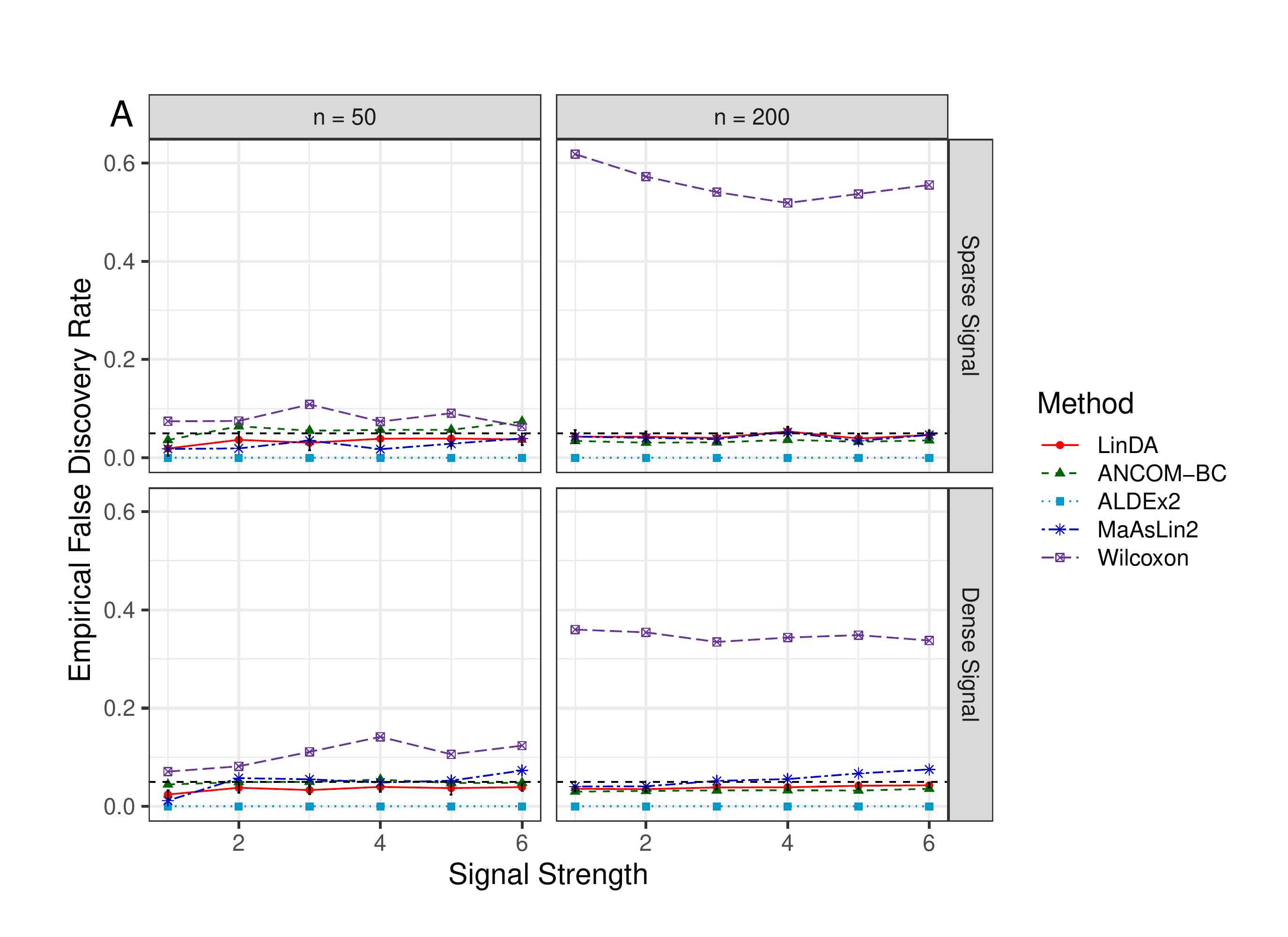}
	\end{subfigure}
	\begin{subfigure}[b]{1\textwidth}
		\centering
		\includegraphics[scale=0.5]{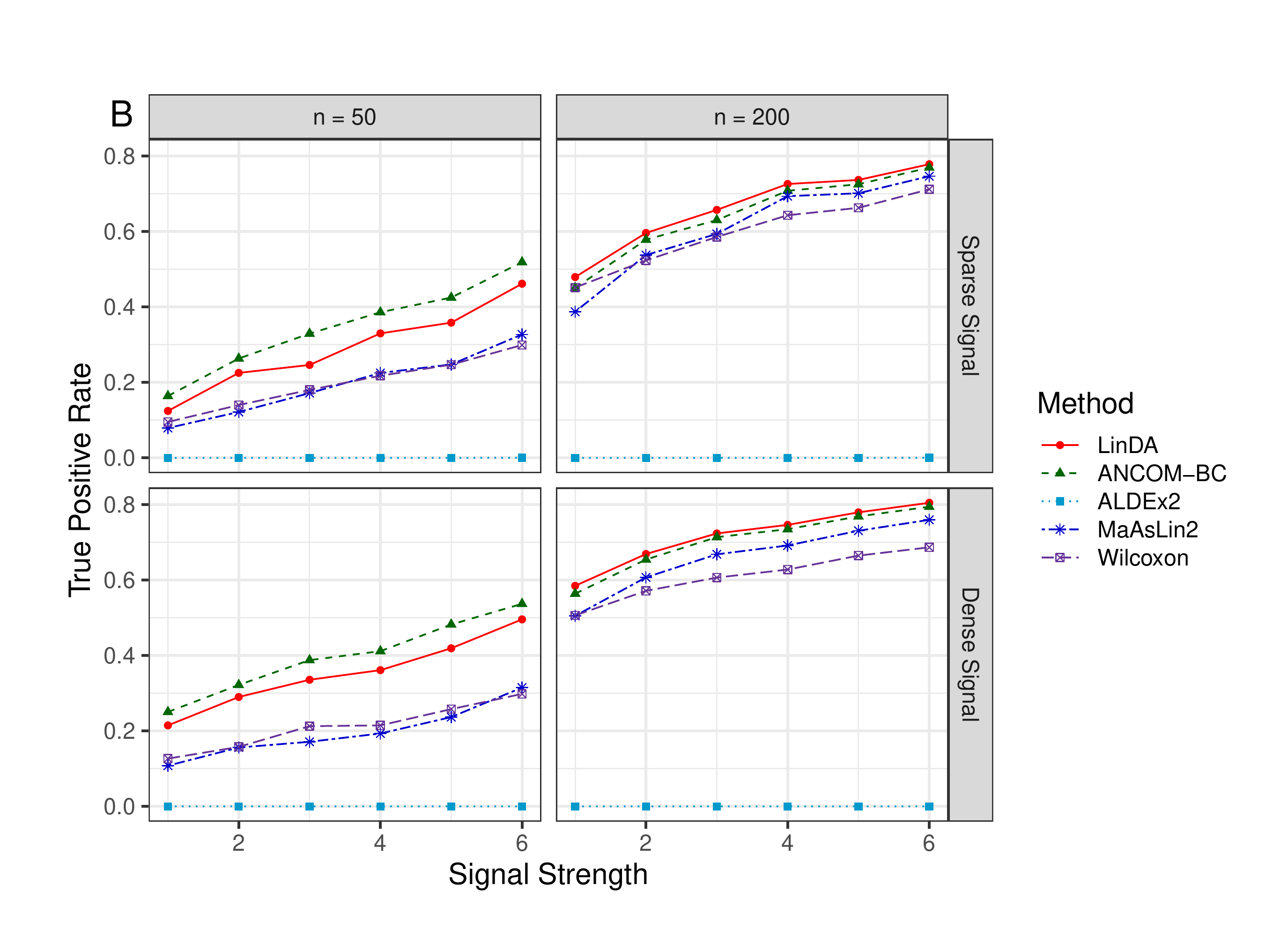}
	\end{subfigure}
	\caption{Performance comparison (S0C2: log normal abundance distribution, a binary variable of interest and two confounders). Empirical false discovery rate (A) and true positive rates (B) were averaged over 100 simulation runs. Error bars (A) represent the 95\% CIs of the method LinDA and the dashed horizontal line indicates the target FDR level of 0.05.}
	\label{fig-S0C2}
\end{figure}

\begin{figure}
	\begin{subfigure}[b]{1\textwidth}
		\centering
		\includegraphics[scale=0.5]{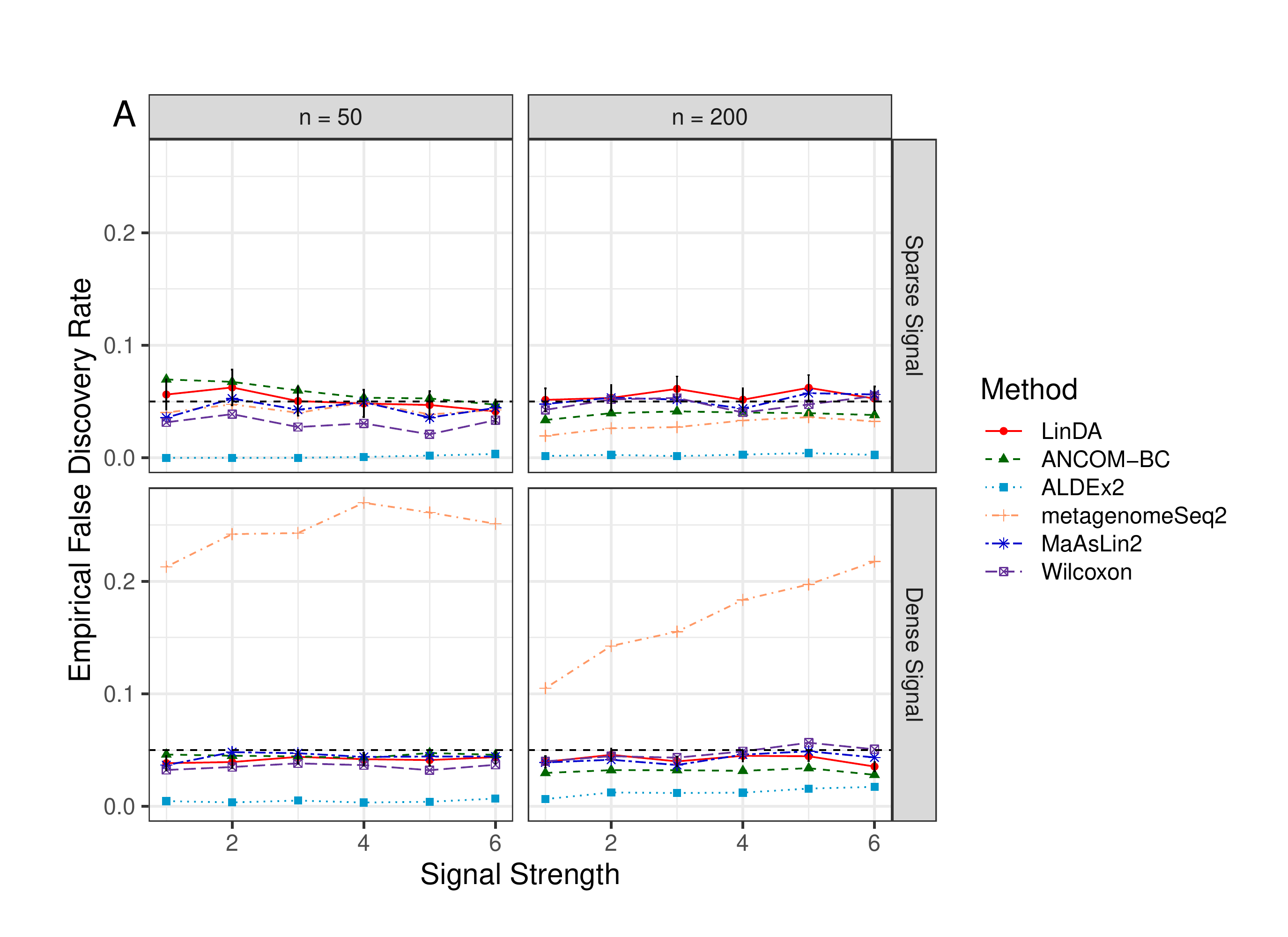}
	\end{subfigure}
	\begin{subfigure}[b]{1\textwidth}
		\centering
		\includegraphics[scale=0.5]{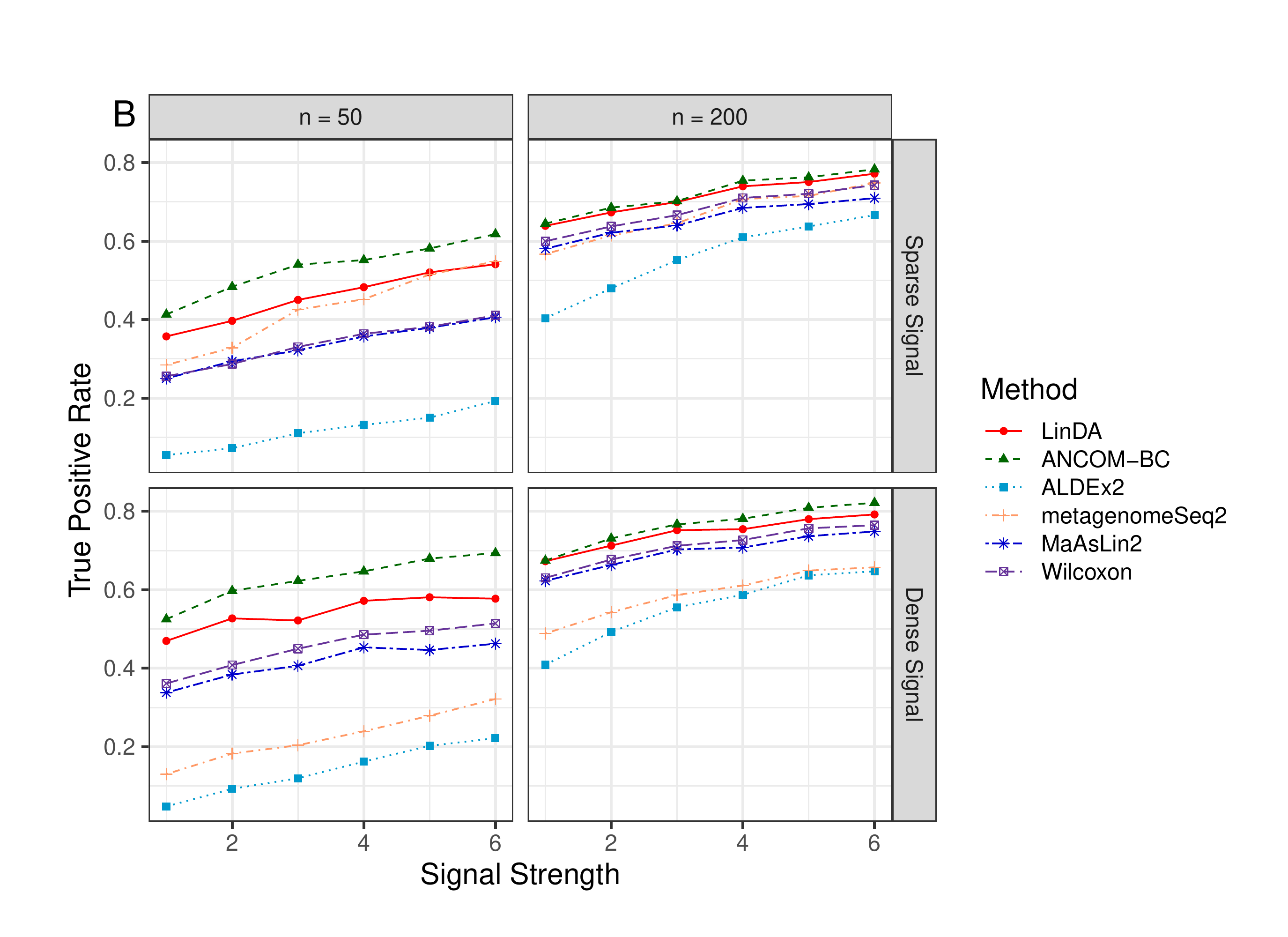}
	\end{subfigure}
	\caption{Performance comparison (S1C0: zero inflated absolute abundances, a binary covariate). Empirical false discovery rate (A) and true positive rates (B) were averaged over 100 simulation runs. Error bars (A) represent the 95\% CIs of the method LinDA and the dashed horizontal line indicates the target FDR level of 0.05.}
	\label{fig-S1C0}
\end{figure}

\begin{figure}
	\begin{subfigure}[b]{1\textwidth}
		\centering
		\includegraphics[scale=0.5]{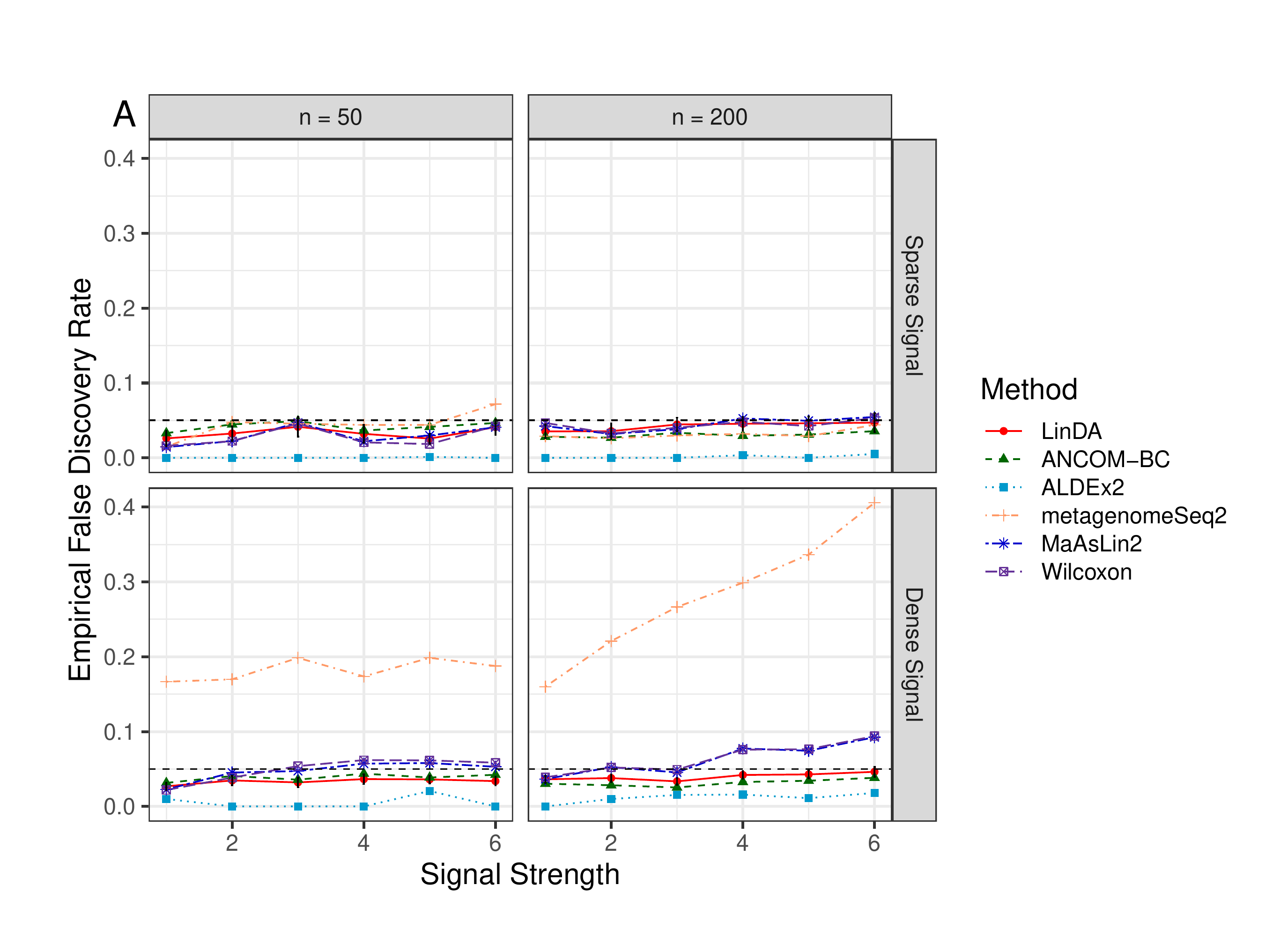}
	\end{subfigure}
	\begin{subfigure}[b]{1\textwidth}
		\centering
		\includegraphics[scale=0.5]{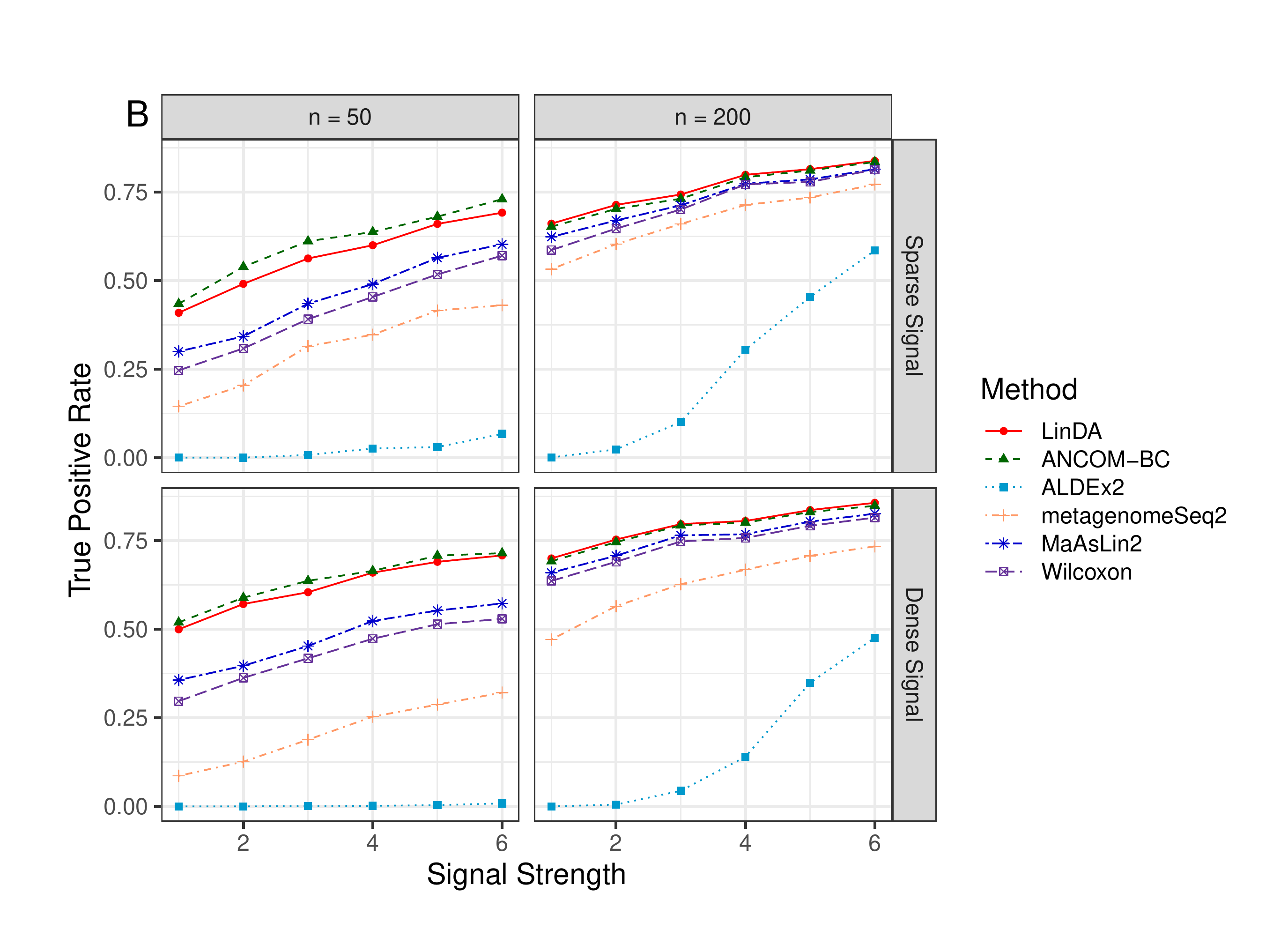}
	\end{subfigure}
	\caption{Performance comparison (S2C0: correlated absolute abundances, a binary covariate). Empirical false discovery rate (A) and true positive rates (B) were averaged over 100 simulation runs. Error bars (A) represent the 95\% CIs of the method LinDA and the dashed horizontal line indicates the target FDR level of 0.05.}
	\label{fig-S2C0}
\end{figure}

\begin{figure}
	\begin{subfigure}[b]{1\textwidth}
		\centering
		\includegraphics[scale=0.5]{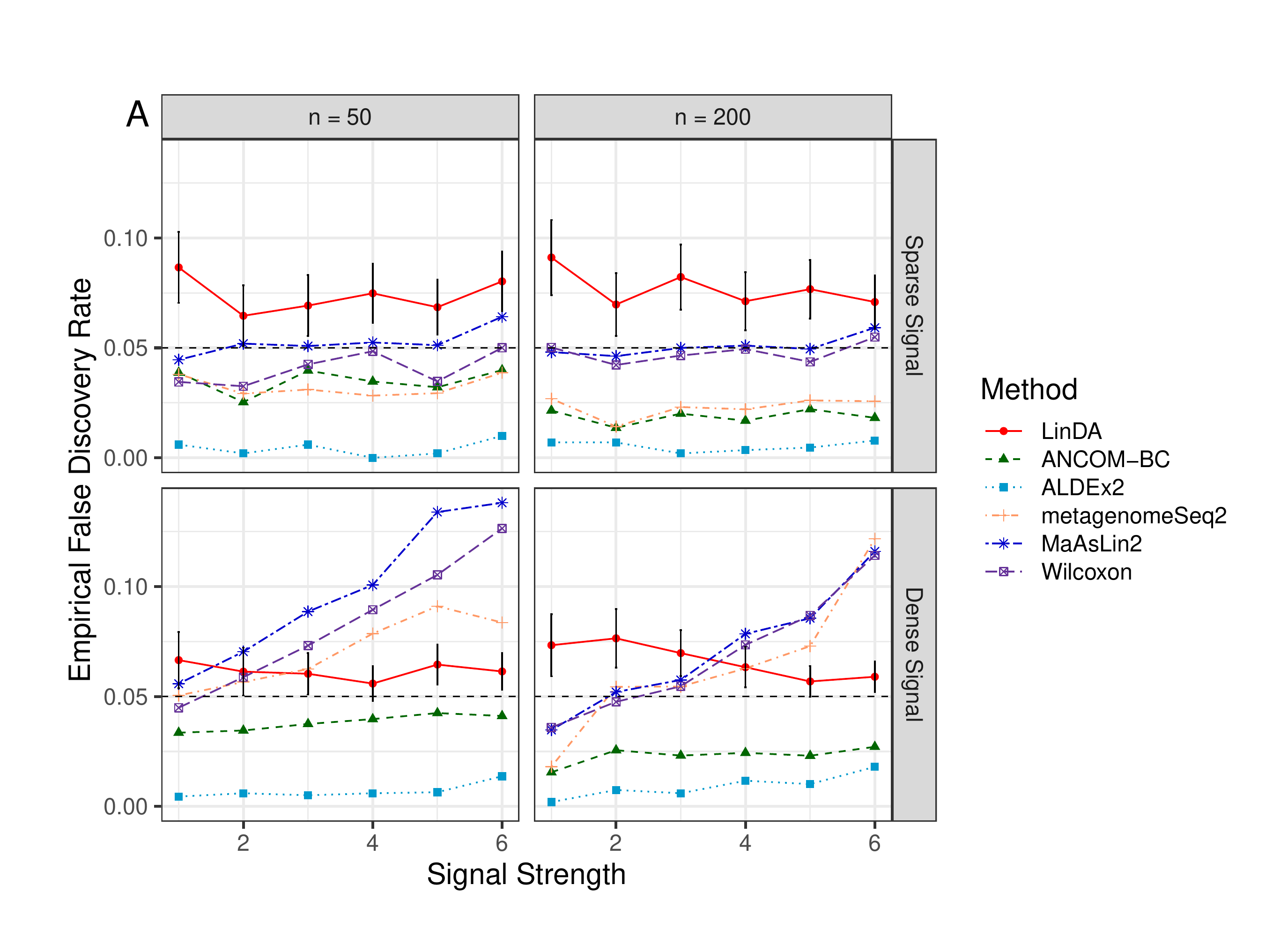}
	\end{subfigure}
	\begin{subfigure}[b]{1\textwidth}
		\centering
		\includegraphics[scale=0.5]{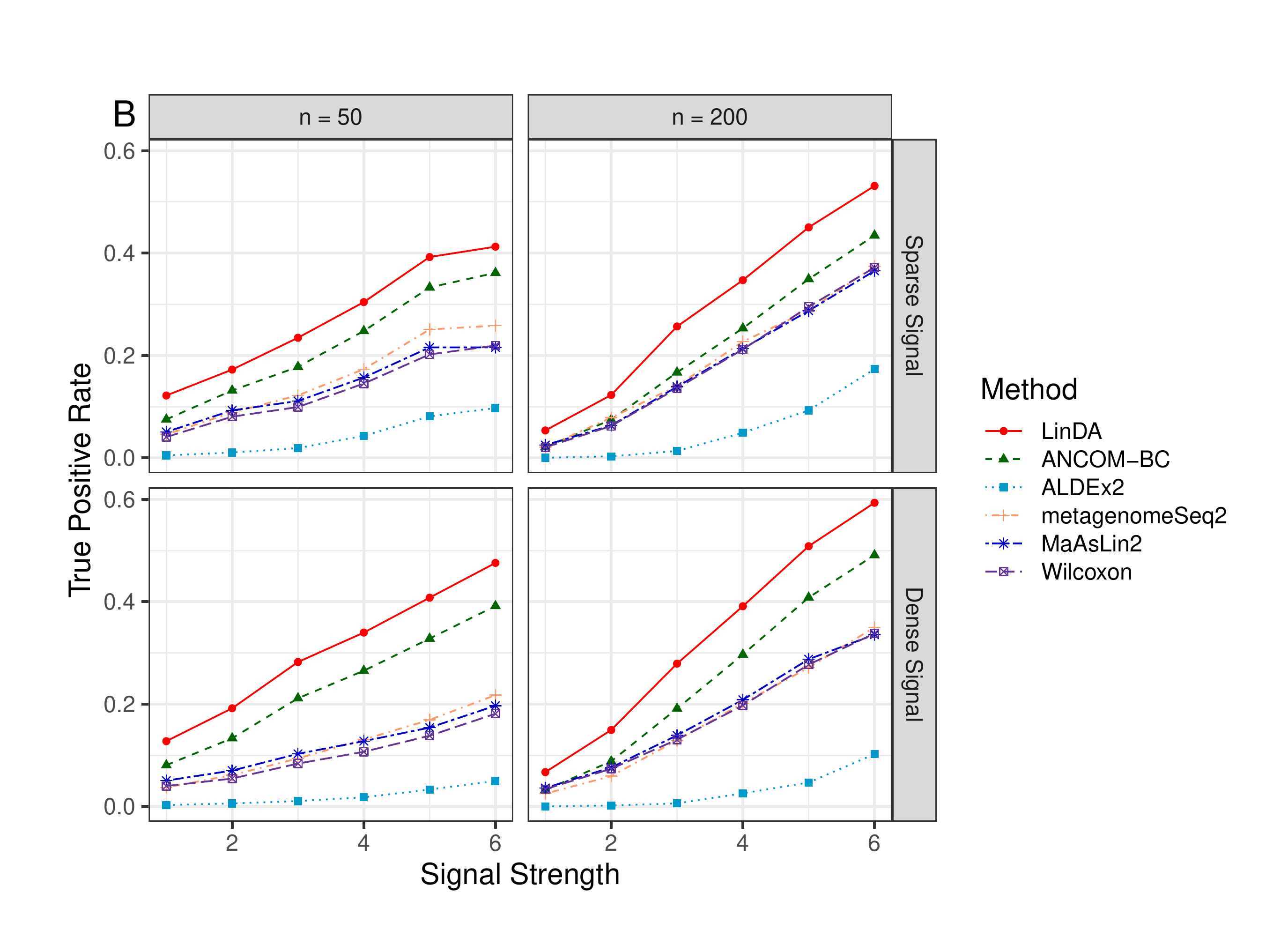}
	\end{subfigure}
	\caption{Performance comparison (S4C0: smaller $m$, a binary covariate). Empirical false discovery rate (A) and true positive rates (B) were averaged over 1000 simulation runs. Error bars (A) represent the 95\% CIs of the method LinDA and the dashed horizontal line indicates the target FDR level of 0.05.}
	\label{fig-S4C0}
\end{figure}

\begin{figure}
	\begin{subfigure}[b]{1\textwidth}
		\centering
		\includegraphics[scale=0.5]{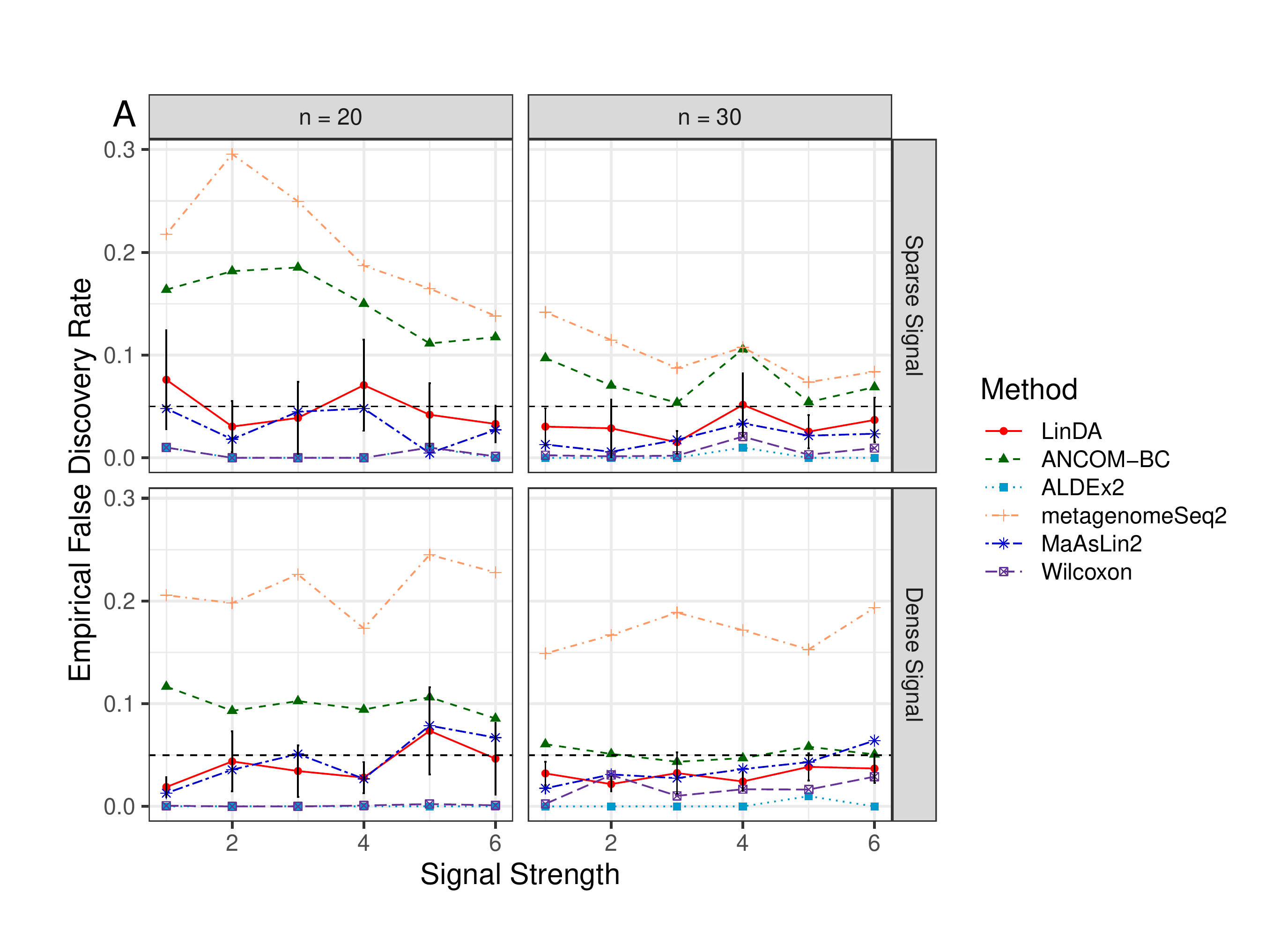}
	\end{subfigure}
	\begin{subfigure}[b]{1\textwidth}
		\centering
		\includegraphics[scale=0.5]{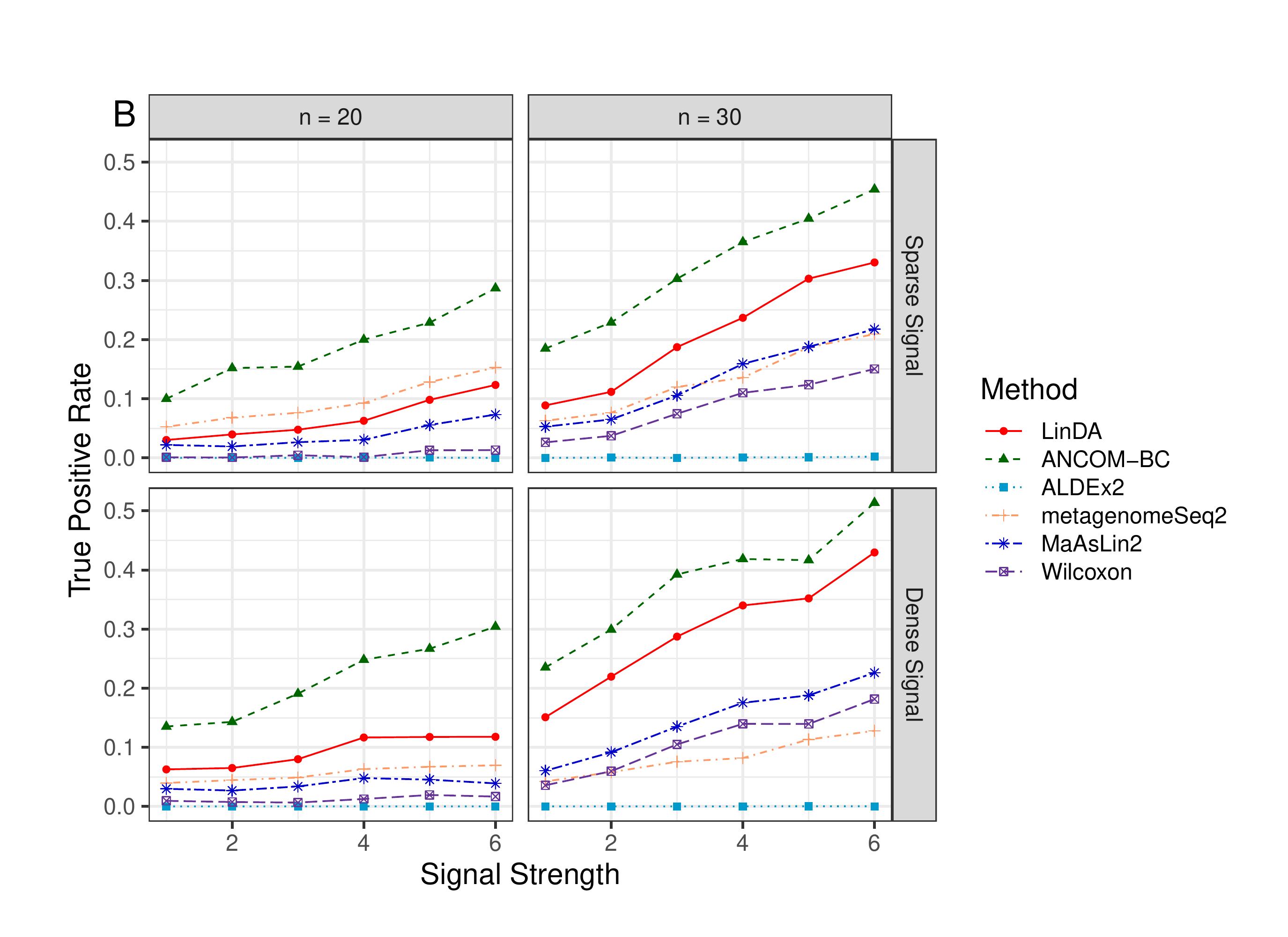}
	\end{subfigure}
	\caption{Performance comparison (S5C0: smaller $n$, a binary covariate). Empirical false discovery rate (A) and true positive rates (B) were averaged over 100 simulation runs. Error bars (A) represent the 95\% CIs of the method LinDA and the dashed horizontal line indicates the target FDR level of 0.05.}
	\label{fig-S5C0}
\end{figure}

\begin{figure}
	\begin{subfigure}[b]{1\textwidth}
		\centering
		\includegraphics[scale=0.5]{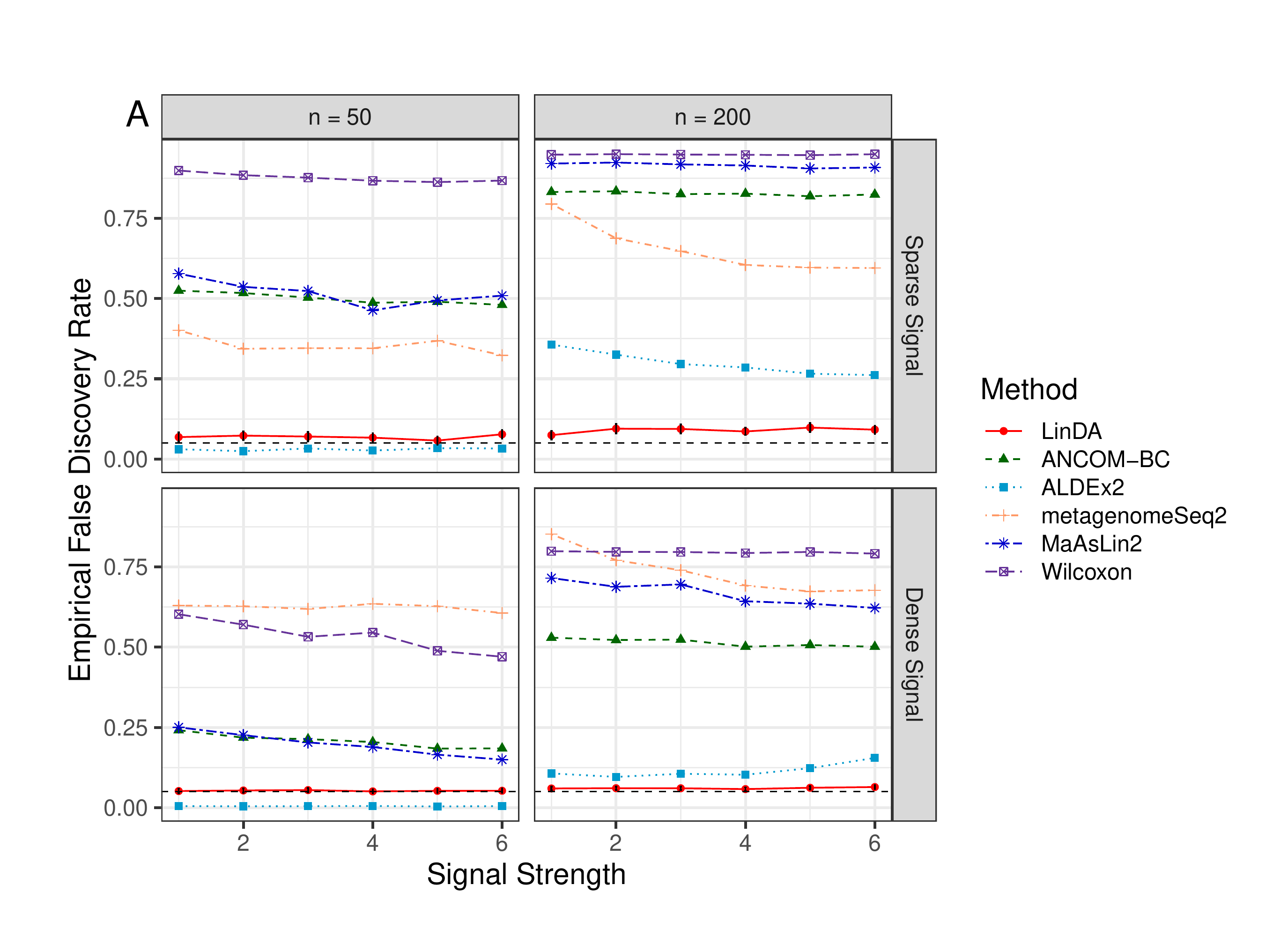}
	\end{subfigure}
	\begin{subfigure}[b]{1\textwidth}
		\centering
		\includegraphics[scale=0.5]{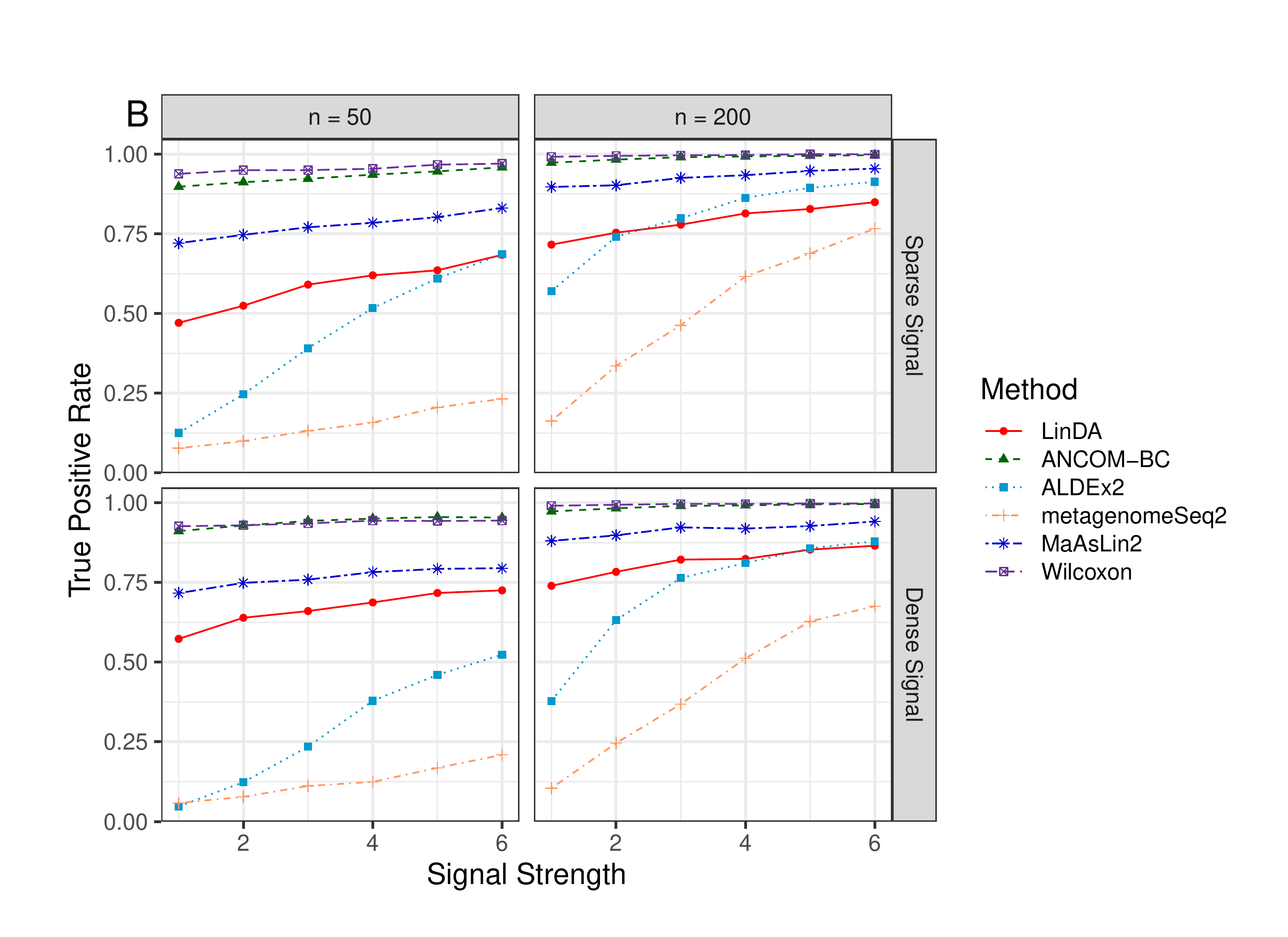}
	\end{subfigure}
	\caption{Performance comparison (S6C0: 10-fold difference in library size, a binary covariate). Empirical false discovery rate (A) and true positive rates (B) were averaged over 100 simulation runs. Error bars (A) represent the 95\% CIs of the method LinDA and the dashed horizontal line indicates the target FDR level of 0.05.}
	\label{fig-S6C0}
\end{figure}

\begin{figure}
	\begin{subfigure}[b]{1\textwidth}
		\centering
		\includegraphics[scale=0.5]{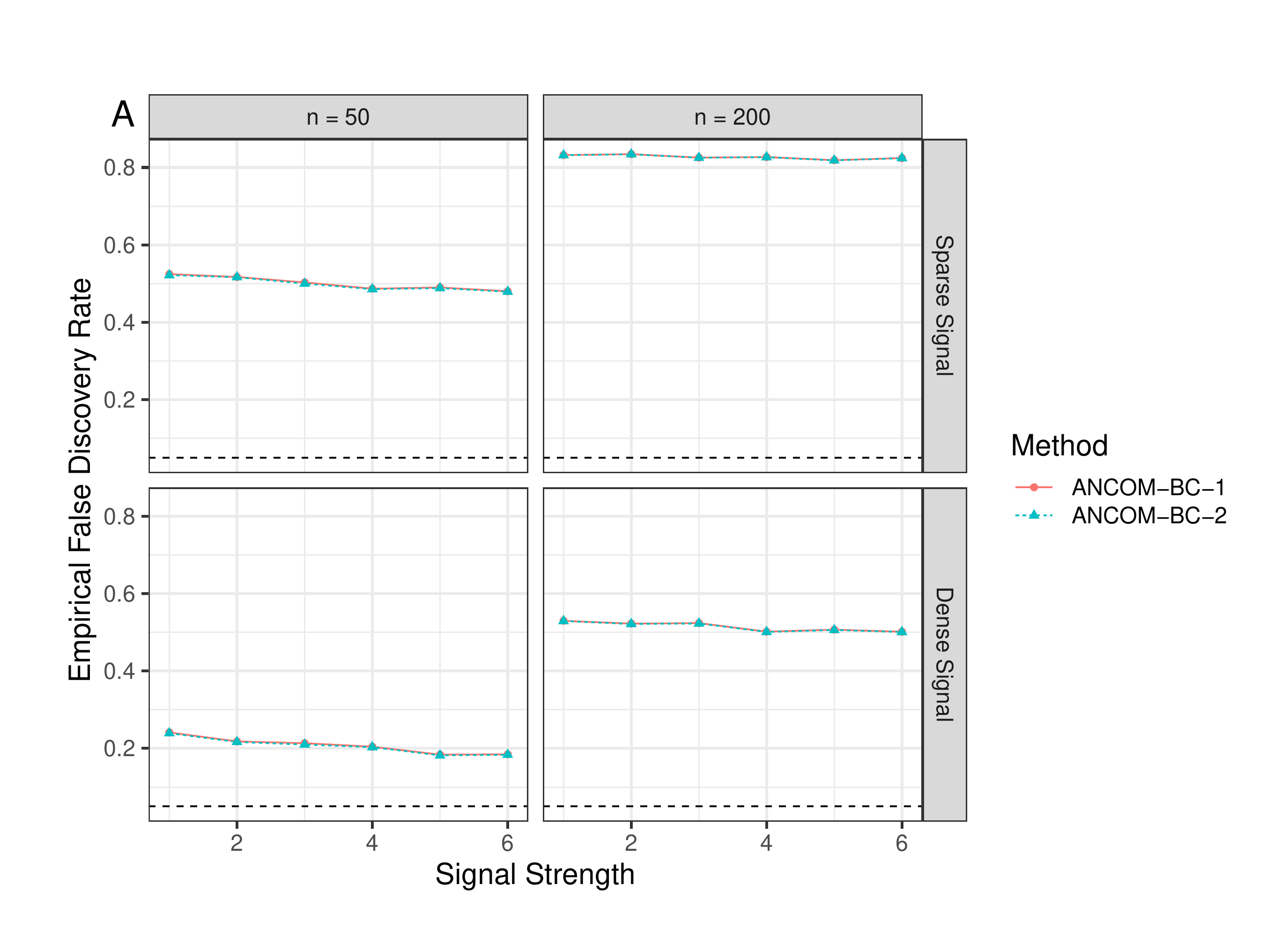}
	\end{subfigure}
	\begin{subfigure}[b]{1\textwidth}
		\centering
		\includegraphics[scale=0.5]{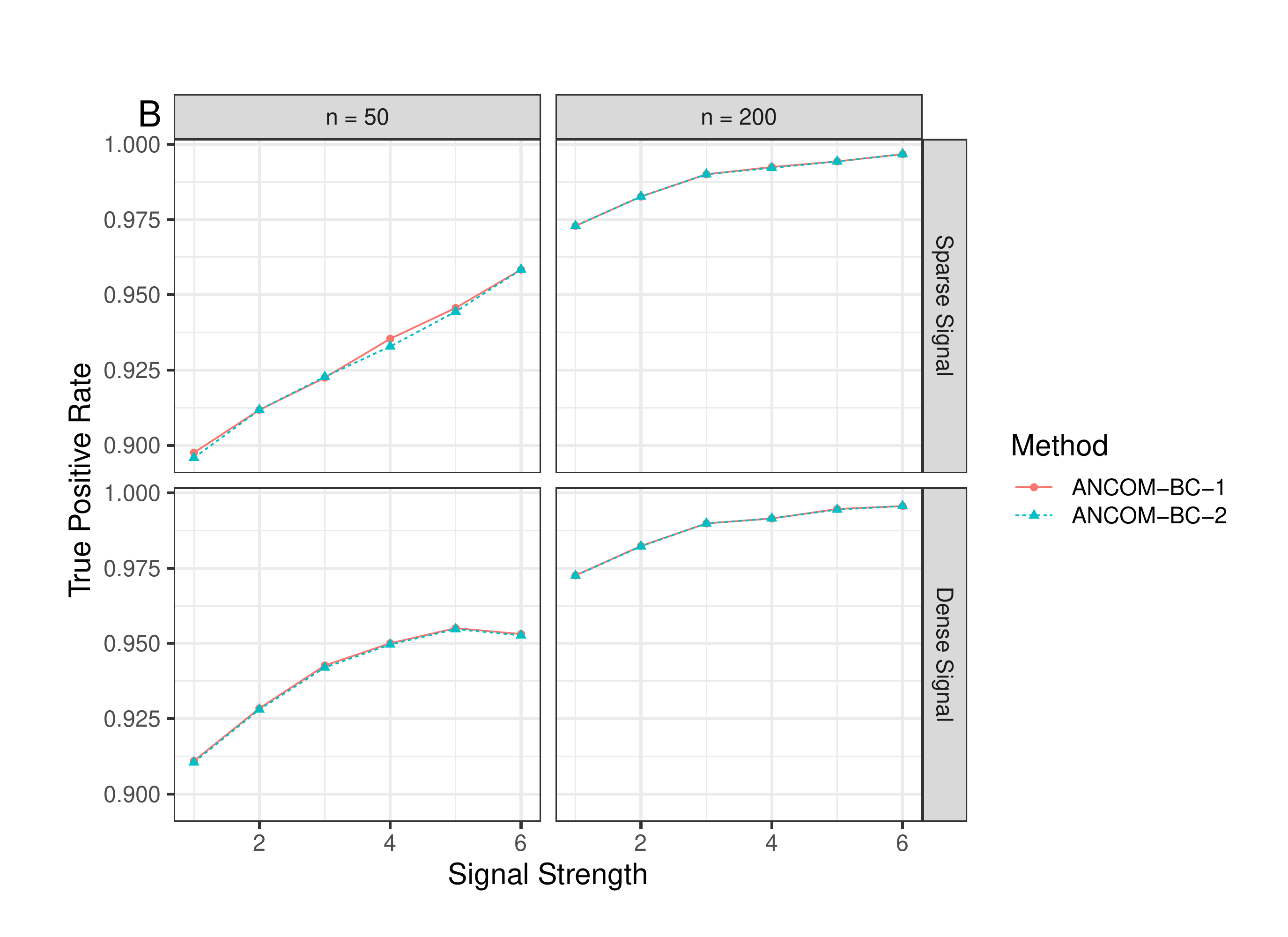}
	\end{subfigure}
	\caption{Performance of ANCOM-BC disabling (ANCOM-BC-1) and enabling (ANCOM-BC-2) zero treatment (S6C0: 10-fold difference in library size, a binary covariate). Empirical false discovery rate (A) and true positive rates (B) were averaged over 100 simulation runs. The dashed horizontal line (A) indicates the target FDR level of 0.05.}
	\label{fig-S6C0-ancombc}
\end{figure}

\begin{figure}
	\begin{subfigure}[b]{1\textwidth}
		\centering
		\includegraphics[scale=0.5]{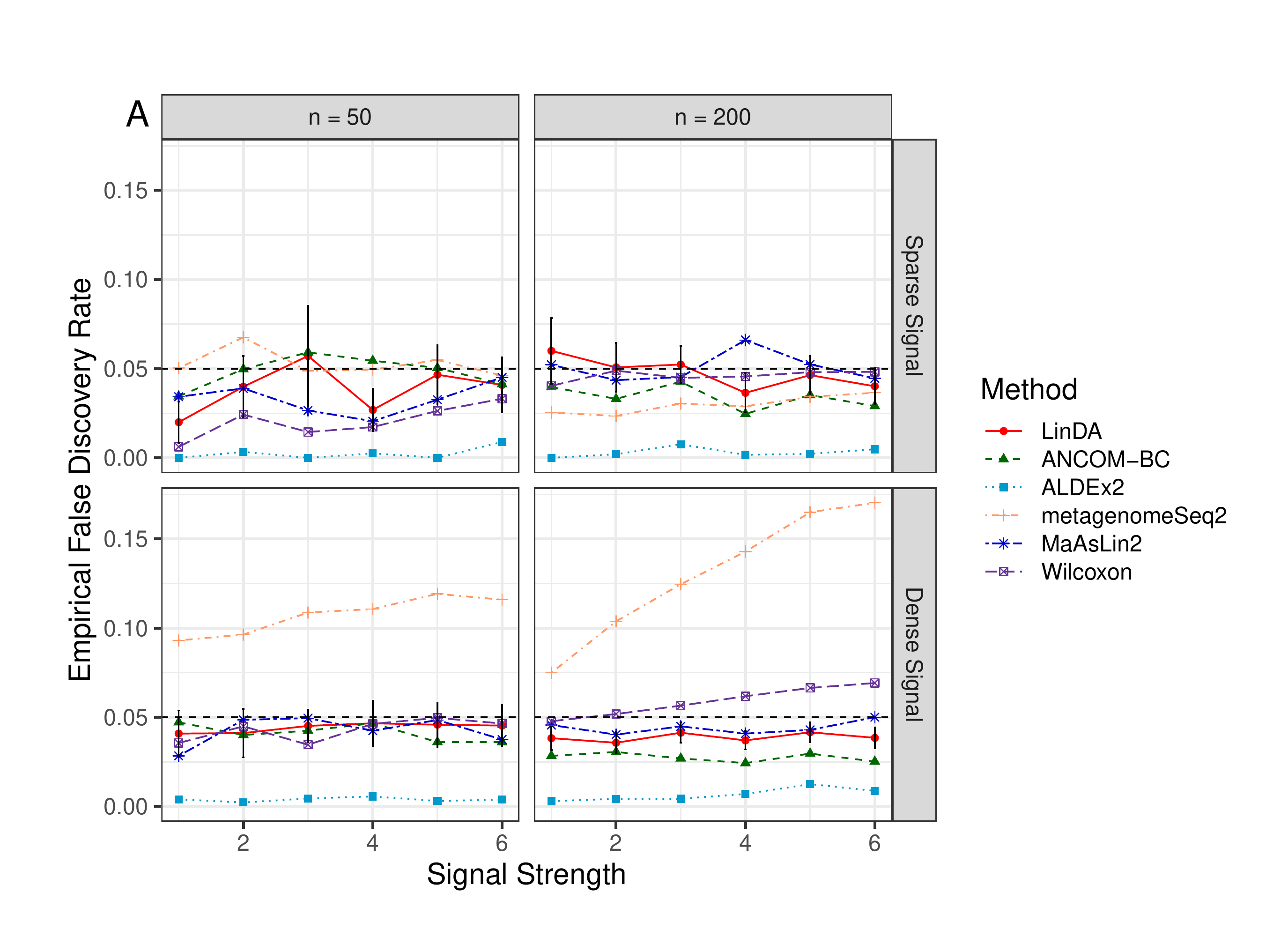}
	\end{subfigure}
	\begin{subfigure}[b]{1\textwidth}
		\centering
		\includegraphics[scale=0.5]{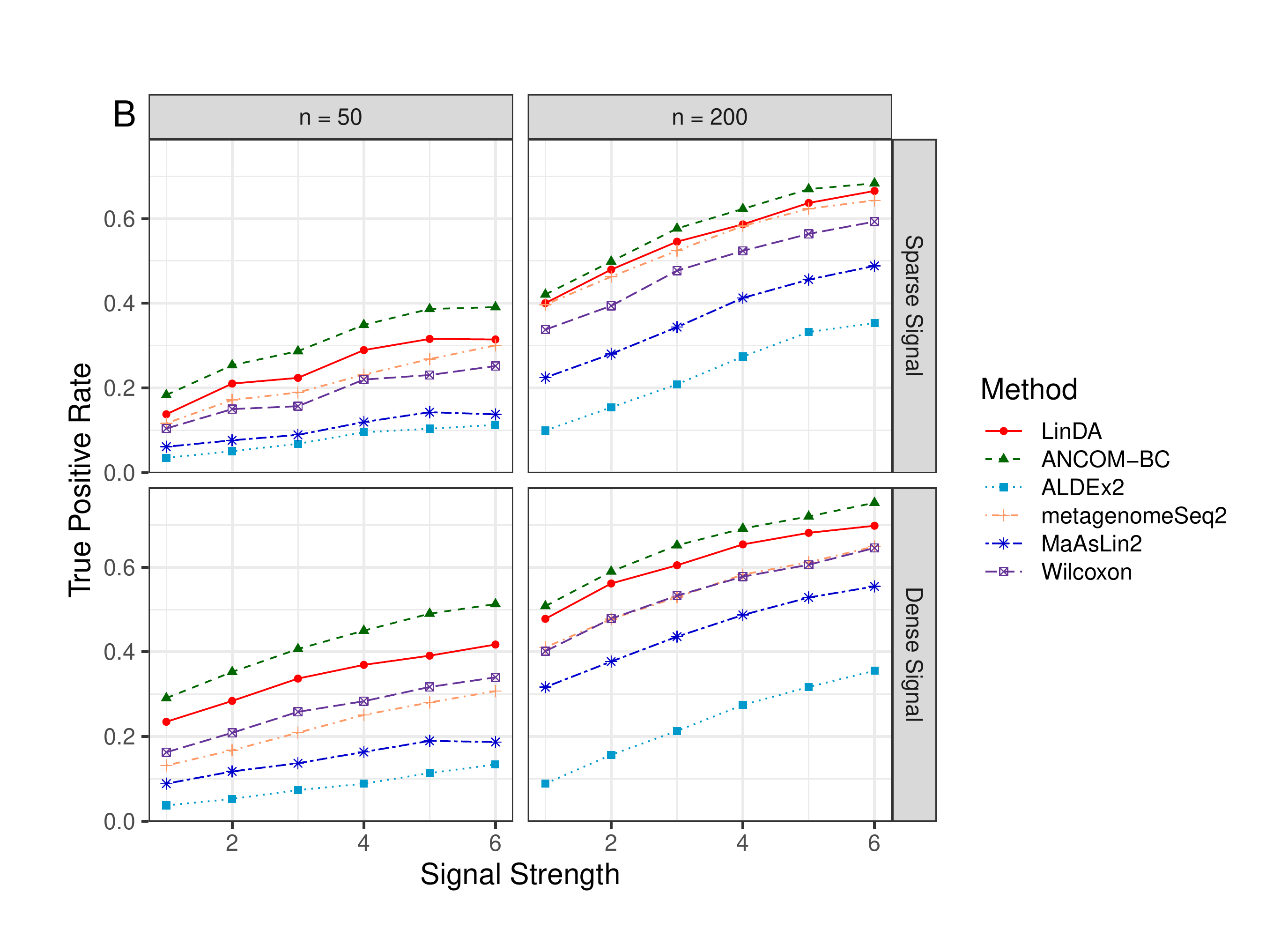}
	\end{subfigure}
	\caption{Performance comparison (S7C0: negative binomial abundance distribution, a binary covariate). Empirical false discovery rate (A) and true positive rates (B) were averaged over 100 simulation runs. Error bars (A) represent the 95\% CIs of the method LinDA and the dashed horizontal line indicates the target FDR level of 0.05.}
	\label{fig-S7C0}
\end{figure}

\begin{figure}
	\begin{subfigure}[b]{1\textwidth}
		\centering
		\includegraphics[scale=0.5]{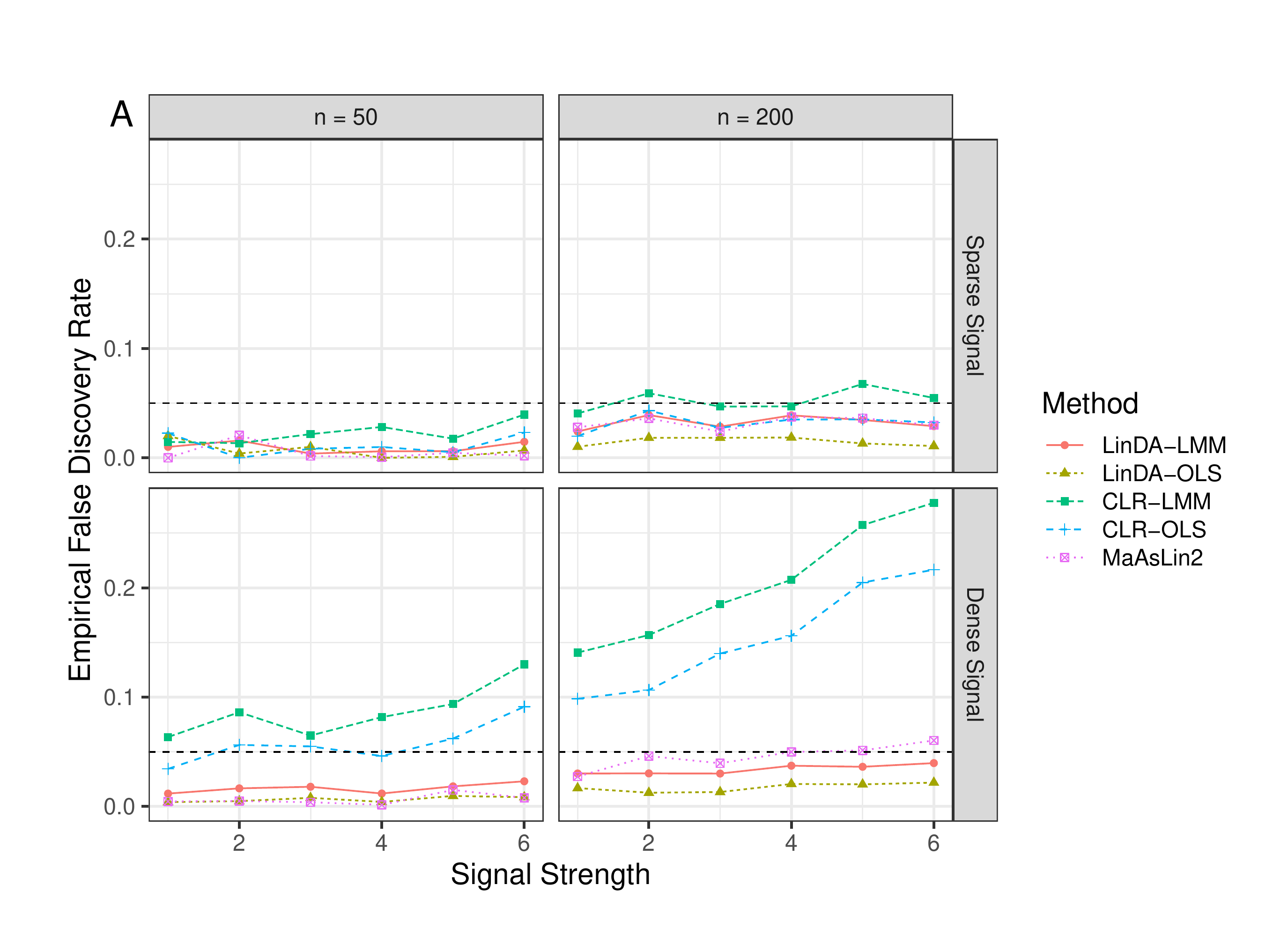}
	\end{subfigure}
	\begin{subfigure}[b]{1\textwidth}
		\centering
		\includegraphics[scale=0.5]{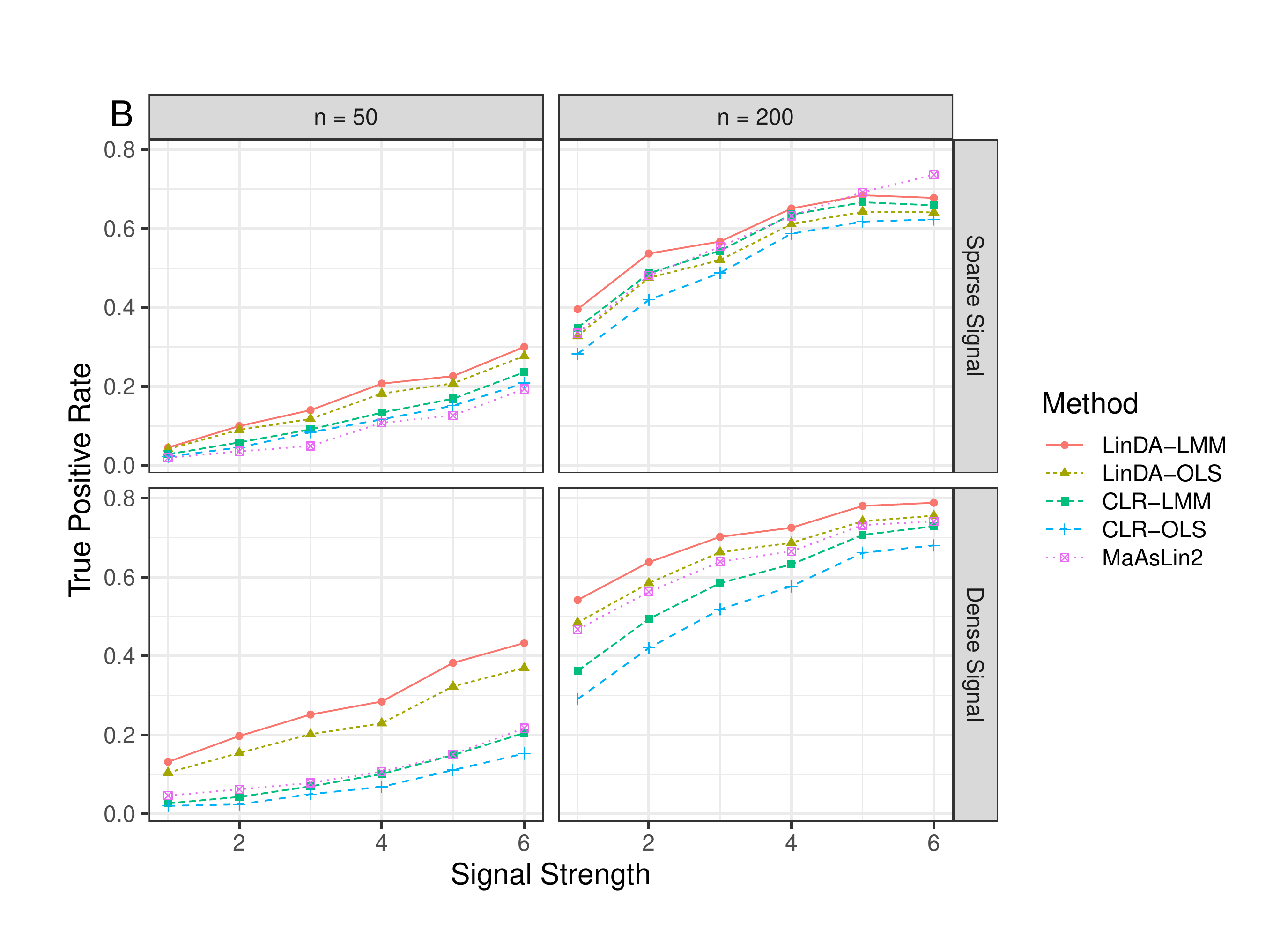}
	\end{subfigure}
	\caption{Performance comparison (S8.1C0: pre-treatment and post-treatment comparison, a binary covariate). Empirical false discovery rate (A) and true positive rates (B) were averaged over 100 simulation runs. The dashed horizontal line (A) indicates the target FDR level of 0.05.}
	\label{fig-S81C0}
\end{figure}

\begin{figure}
	\begin{subfigure}[b]{1\textwidth}
		\centering
		\includegraphics[scale=0.5]{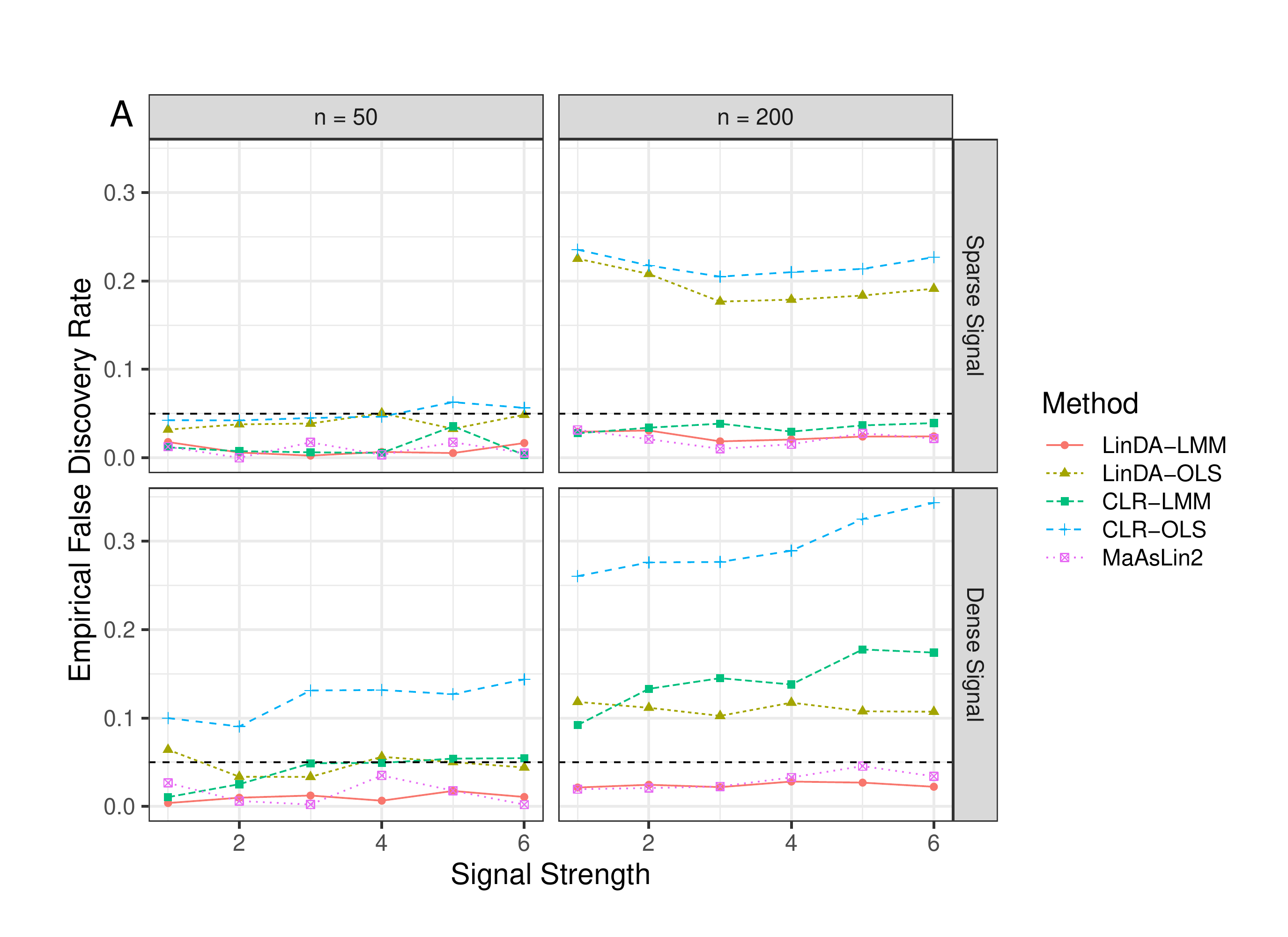}
	\end{subfigure}
	\begin{subfigure}[b]{1\textwidth}
		\centering
		\includegraphics[scale=0.5]{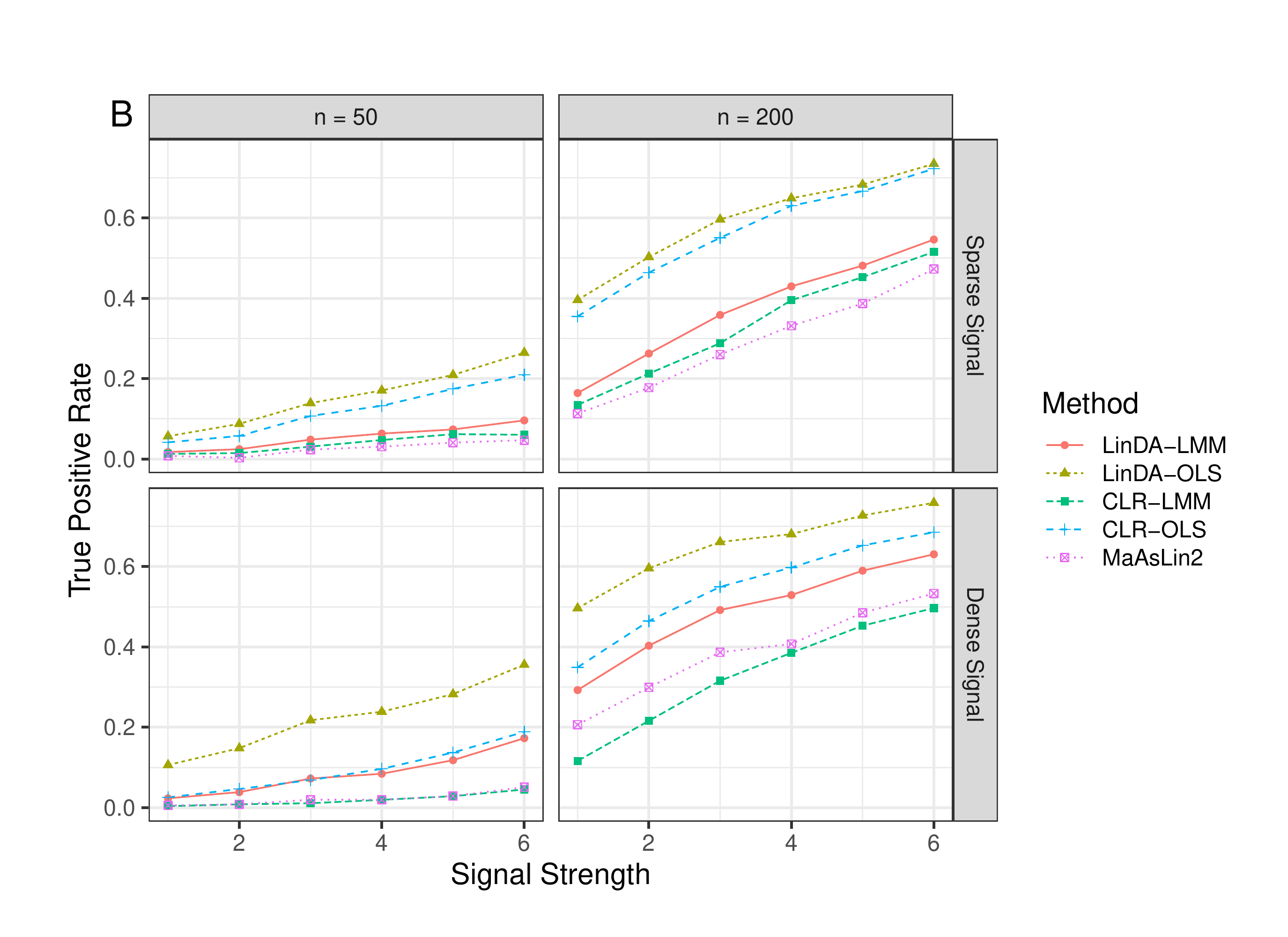}
	\end{subfigure}
	\caption{Performance comparison (S8.2C0: replicate sampling, a binary covariate). Empirical false discovery rate (A) and true positive rates (B) were averaged over 100 simulation runs. The dashed horizontal line (A) indicates the target FDR level of 0.05.}
	\label{fig-S82C0}
\end{figure}

\begin{figure}
	\begin{subfigure}[b]{1\textwidth}
		\centering
		\includegraphics[scale=0.5]{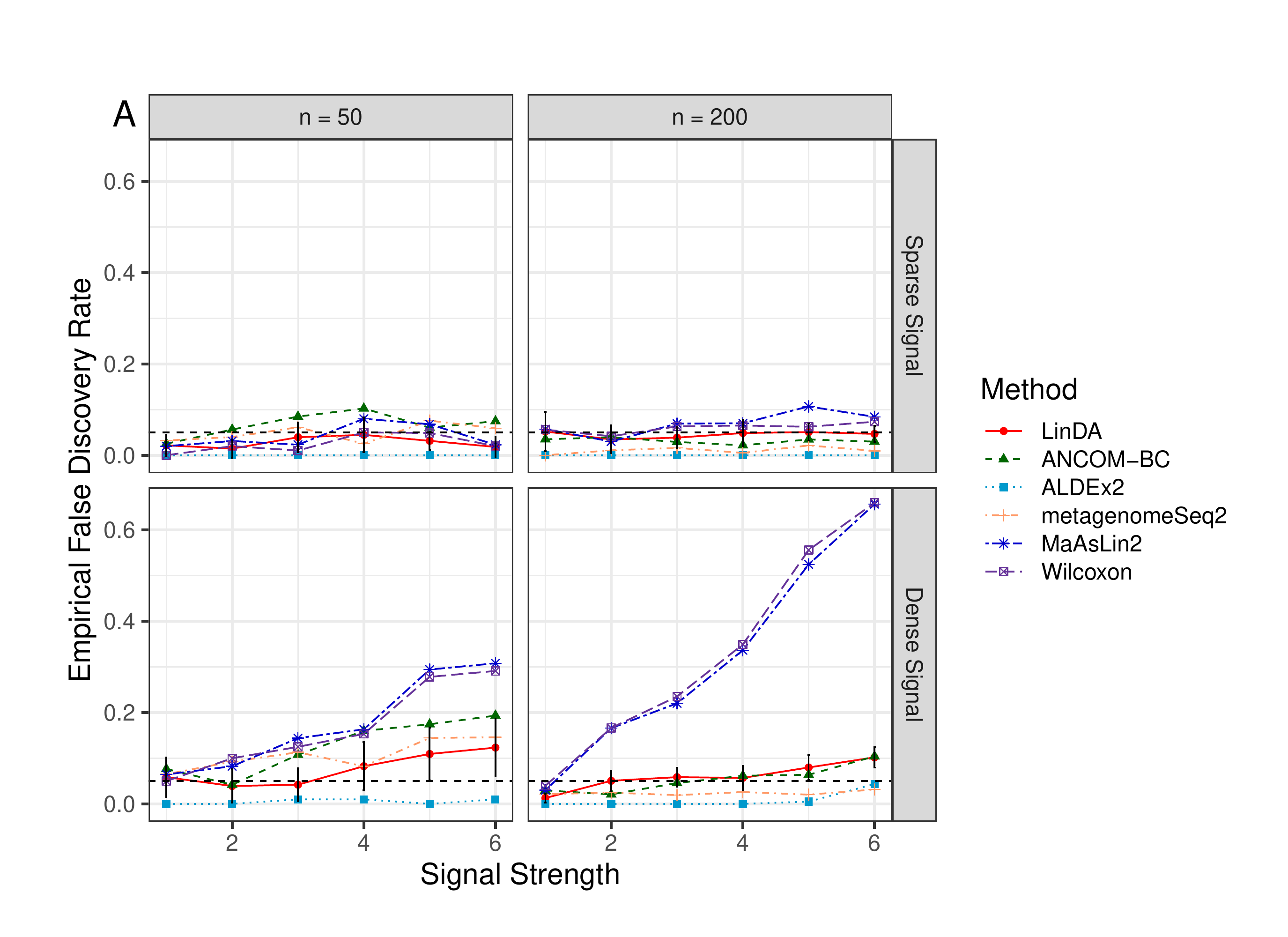}
	\end{subfigure}
	\begin{subfigure}[b]{1\textwidth}
		\centering
		\includegraphics[scale=0.5]{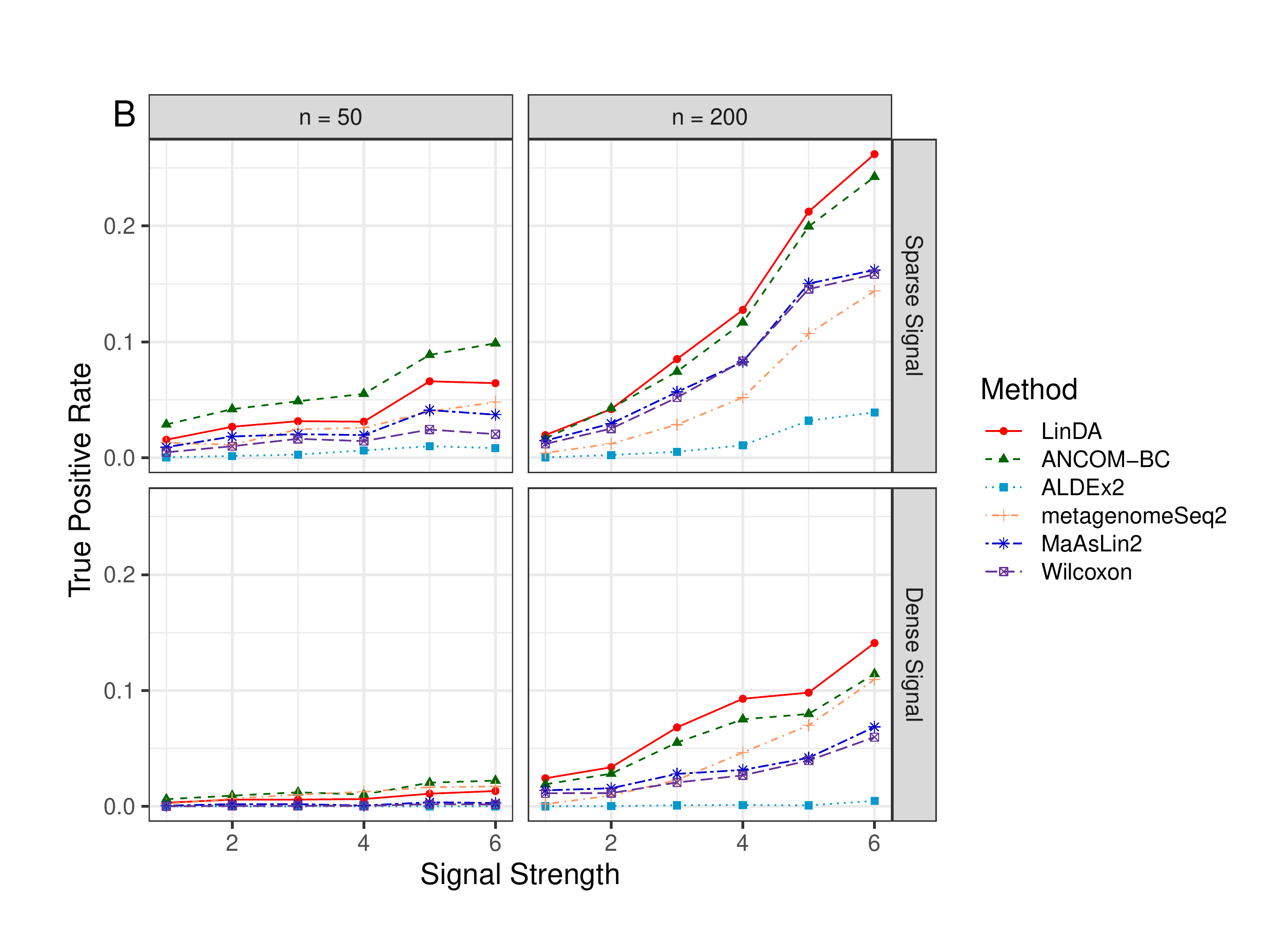}
	\end{subfigure}
	\caption{Performance comparison (S0C0 with strong compositional effects). Empirical false discovery rate (A) and true positive rates (B) were averaged over 100 simulation runs. Error bars (A) represent the 95\% CIs of the method LinDA and the dashed horizontal line indicates the target FDR level of 0.05.}
	\label{fig-S0C0-strong}
\end{figure}

\section{Additional results of real data applications}\label{sec-supp-real}
Fig. \ref{fig-real-CDI}--\ref{fig-real-SMOKE} show the effect size plots and volcano plots for the four datasets (CDI, IBD, RA, and SMOKE) respectively.

\begin{figure}
	\begin{subfigure}[b]{1\textwidth}
		\subcaption*{A}
		\centering
		\includegraphics[scale=0.45]{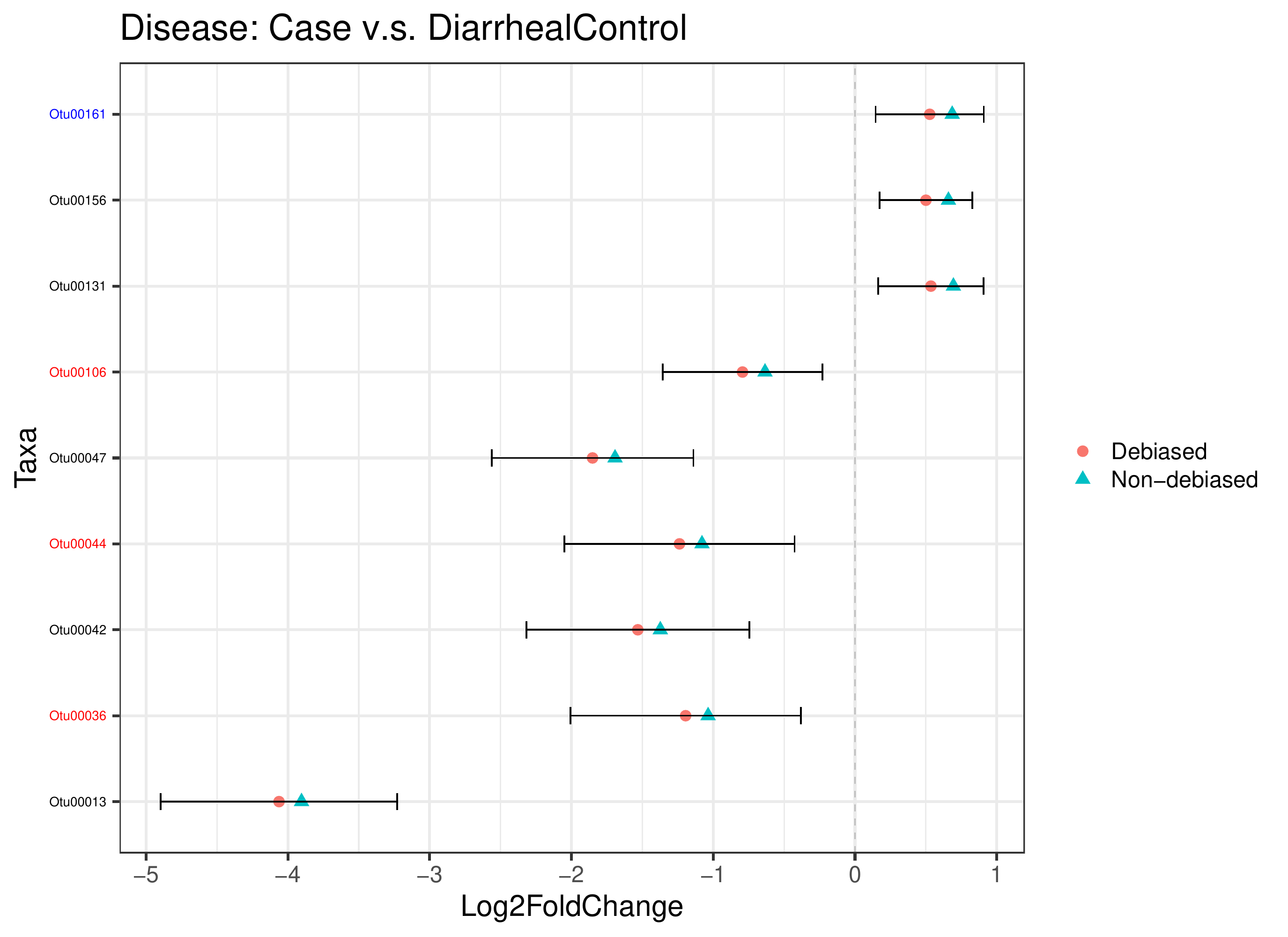}
	\end{subfigure}
	\begin{subfigure}[b]{1\textwidth}
		\subcaption*{B}
		\centering
		\includegraphics[scale=0.45]{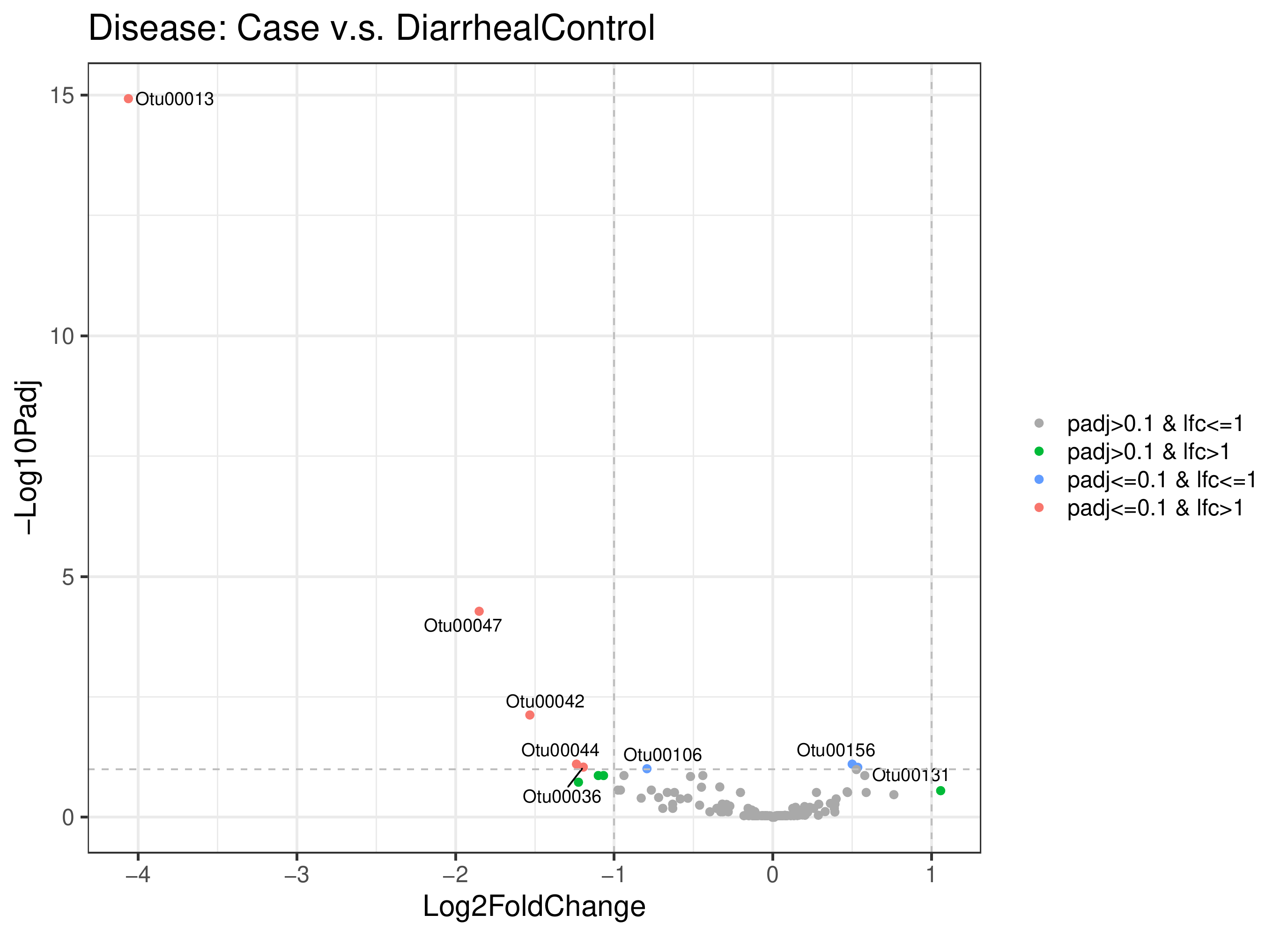}
	\end{subfigure}
	\caption{Effect size plot (A) of differential taxa at FDR level of 0.1 and volcano plot (B) for the CDI dataset. The ``Debiased" points represent the bias-corrected regression coefficients, and ``Non-debiased" points represent the original (biased) regression coefficients. The error bars represent the 95\% CIs of the ``Debiased" points. The taxa in black are detected by LinDA, taxa in red are detected solely by LinDA, and the taxa in blue are missed by LinDA but detected by one or more of the other methods (A).}
	\label{fig-real-CDI}
\end{figure}

\begin{figure}
	\begin{subfigure}[b]{1\textwidth}
		\subcaption*{A}
		\centering
		\includegraphics[scale=0.45]{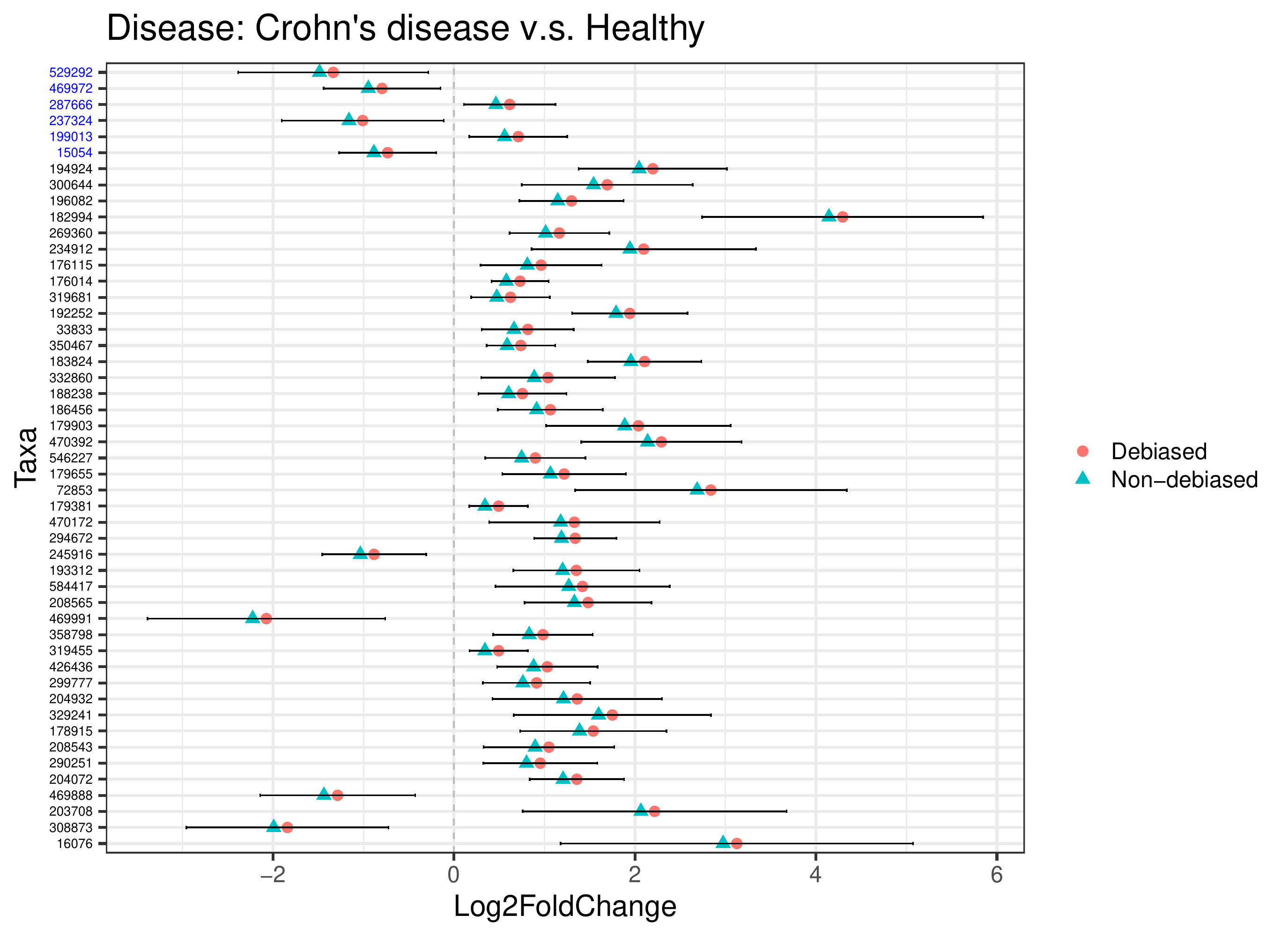}
	\end{subfigure}
	\begin{subfigure}[b]{1\textwidth}
		\subcaption*{B}
		\centering
		\includegraphics[scale=0.45]{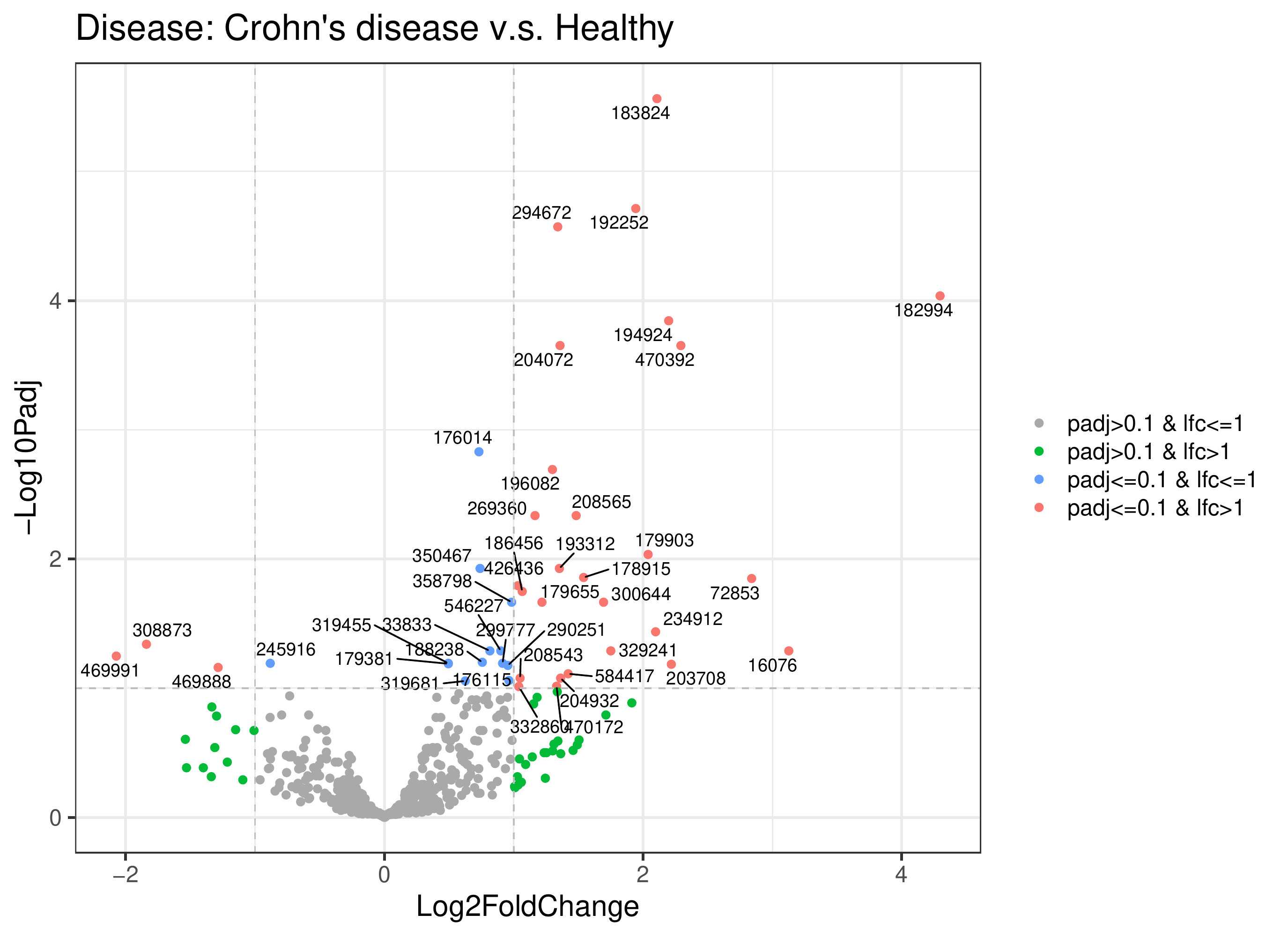}
	\end{subfigure}
	\caption{Effect size plot (A) of differential taxa at FDR level of 0.1 and volcano plot (B) for the IBD dataset. The ``Debiased" points represent the bias-corrected regression coefficients, and ``Non-debiased" points represent the original (biased) regression coefficients. The error bars represent the 95\% CIs of the ``Debiased" points. The taxa in black are detected by LinDA, taxa in red are detected solely by LinDA, and taxa in blue are missed by LinDA but detected by two or more of the other methods (A). }
	\label{fig-real-IBD}
\end{figure}

\begin{figure}
	\begin{subfigure}[b]{1\textwidth}
		\subcaption*{A}
		\centering
		\includegraphics[scale=0.45]{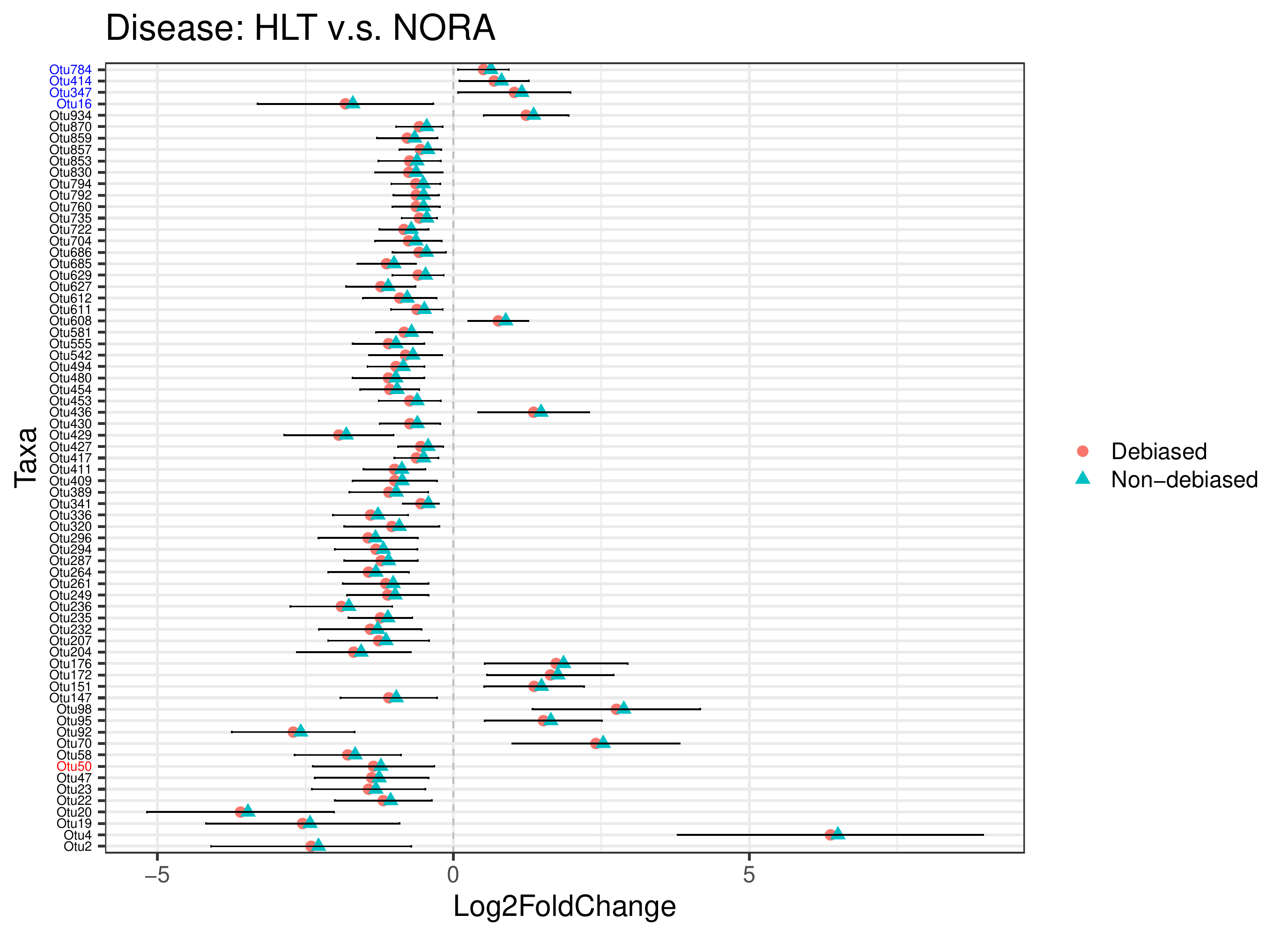}
	\end{subfigure}
	\begin{subfigure}[b]{1\textwidth}
		\subcaption*{B}
		\centering
		\includegraphics[scale=0.45]{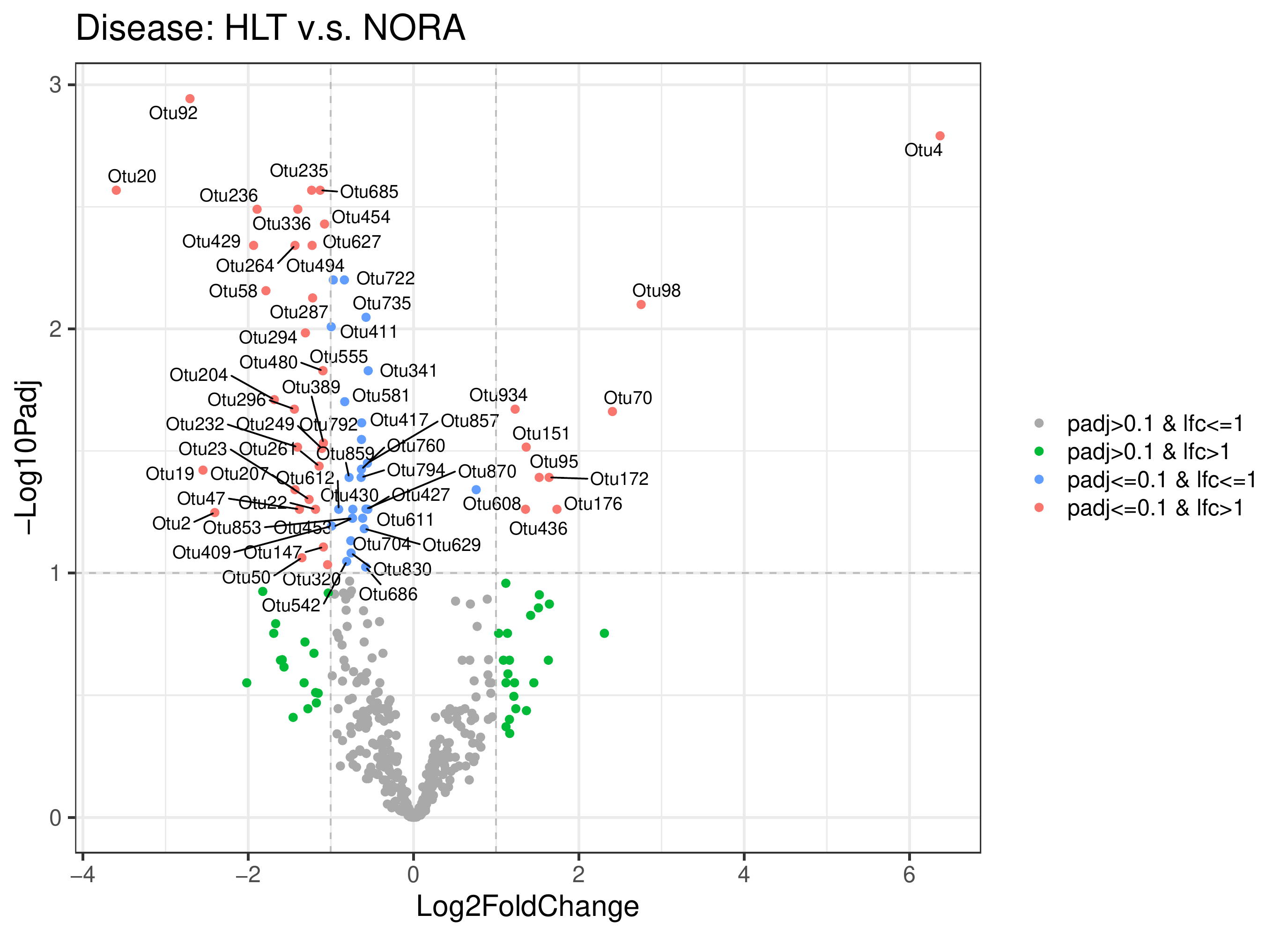}
	\end{subfigure}
	\caption{Effect size plot (A) of differential taxa at FDR level of 0.1 and volcano plot (B) for the RA dataset. The 
		``Debiased" points represent the bias-corrected regression coefficients, and ``Non-debiased" points represent the original (biased) regression coefficients. The error bars represent the 95\% CIs of the ``Debiased" points. The taxa in black are detected by LinDA, taxa in red are detected solely by LinDA, and taxa in blue are missed by LinDA but detected by two or more of the other methods (A).}
	\label{fig-real-RA}
\end{figure}

\begin{figure}
	\begin{subfigure}[b]{1\textwidth}
		\subcaption*{A}
		\centering
		\includegraphics[scale=0.45]{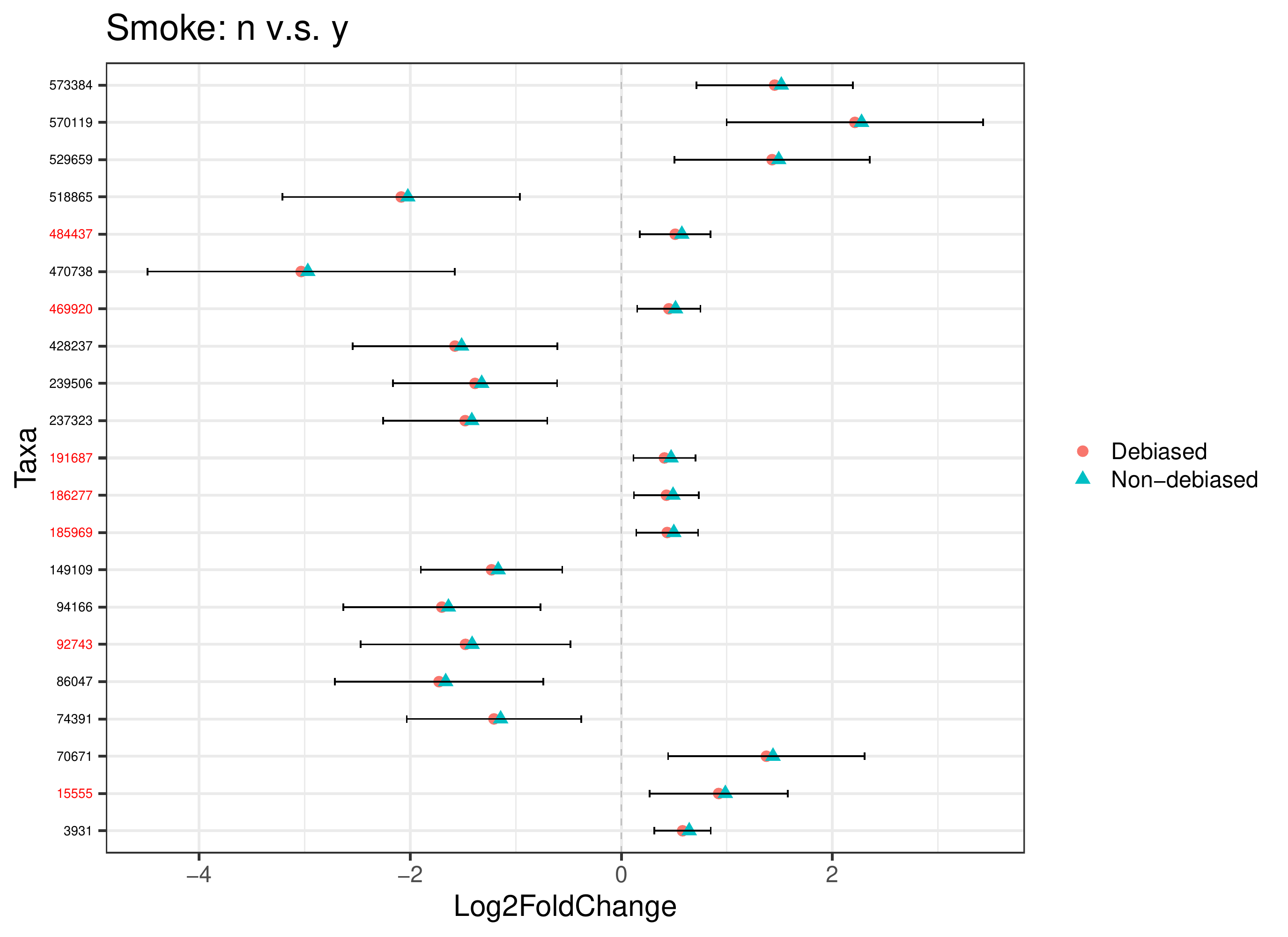}
	\end{subfigure}
	\begin{subfigure}[b]{1\textwidth}
		\subcaption*{B}
		\centering
		\includegraphics[scale=0.45]{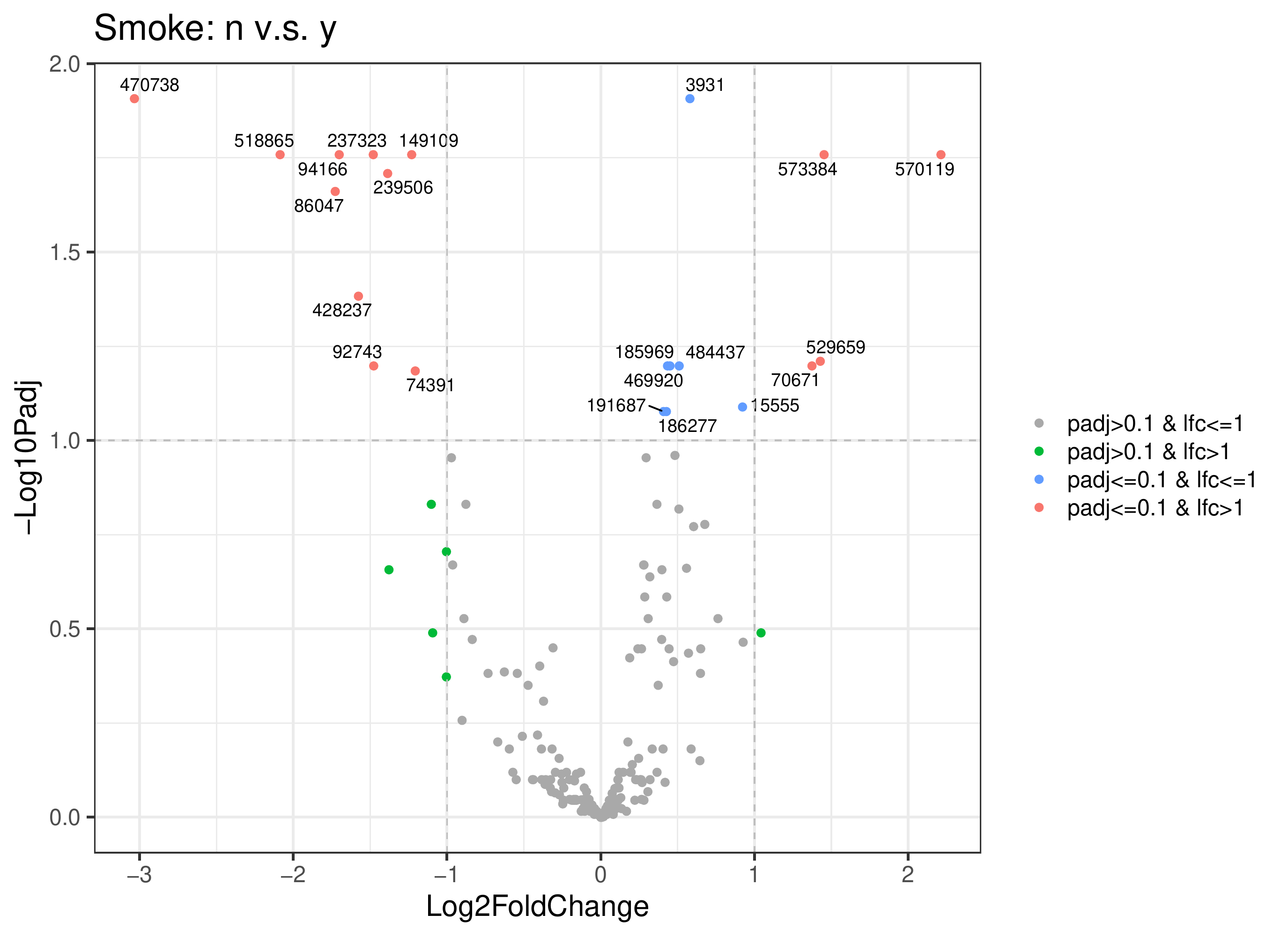}
	\end{subfigure}
	\caption{Effect size plot (A) of differential taxa detected by LinDA at FDR level of 0.1 and volcano plot (B) for the SMOKE dataset. The ``Debiased" points represent the bias-corrected regression coefficients, and ``Non-debiased" points represent the original (biased) regression coefficients. The error bars represent the 95\% CIs of the ``Debiased" points. The taxa in black are detected by LinDA, taxa in red are detected by LinDA but missed by MaAsLin2, and no taxa are detected by MaAsLin2 but missed by LinDA (A).}
	\label{fig-real-SMOKE}
\end{figure}

\section{Full comparisons of numerical studies}\label{sec-supp-simu-full}
Fig. \ref{fig-S0C0All}--\ref{fig-S0C0StrongAll} present the full result of all methods under different simulation settings.

\begin{figure}
	\begin{subfigure}[b]{1\textwidth}
		\centering
		\includegraphics[scale=0.5]{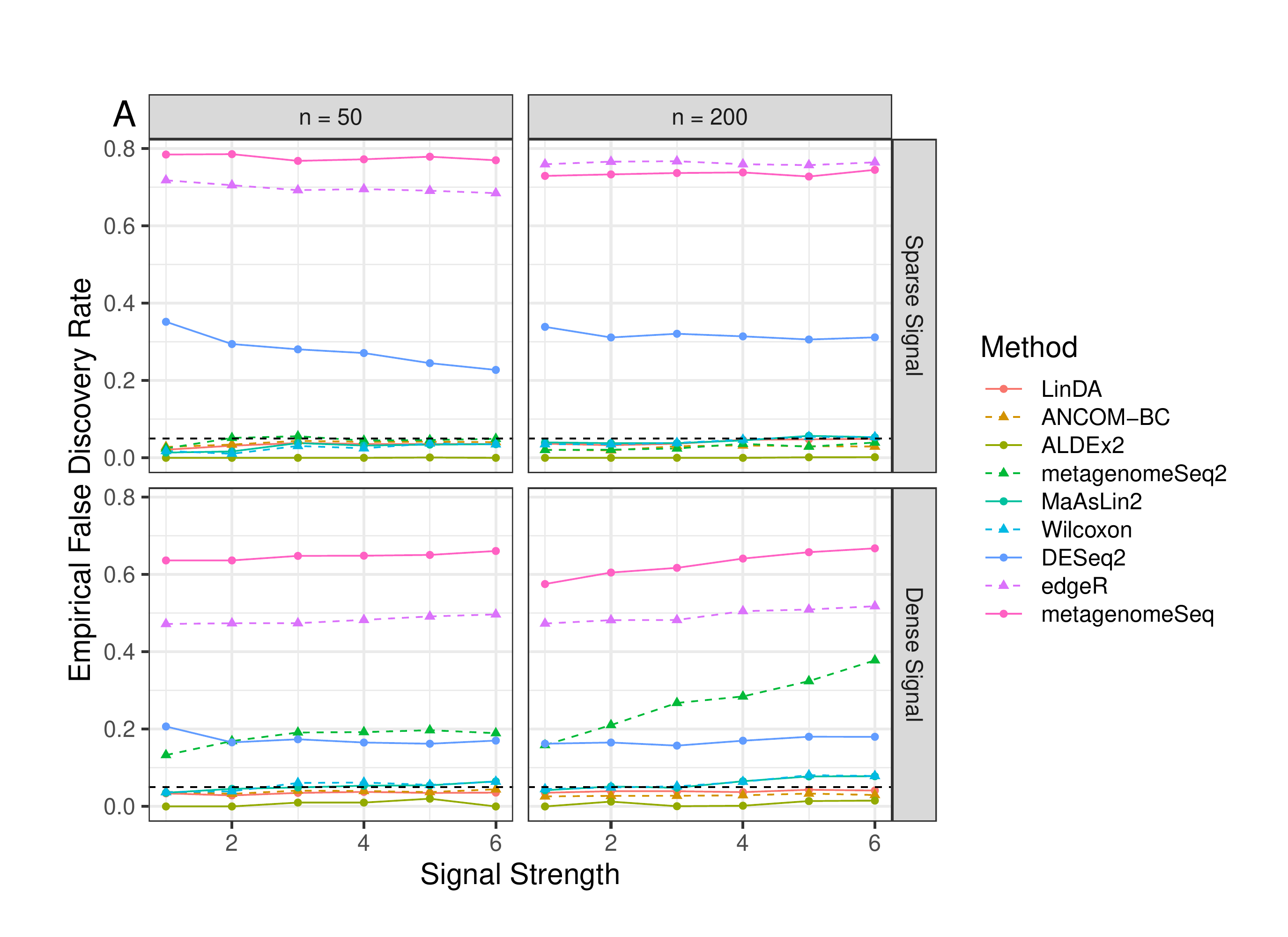}
	\end{subfigure}
	\begin{subfigure}[b]{1\textwidth}
		\centering
		\includegraphics[scale=0.5]{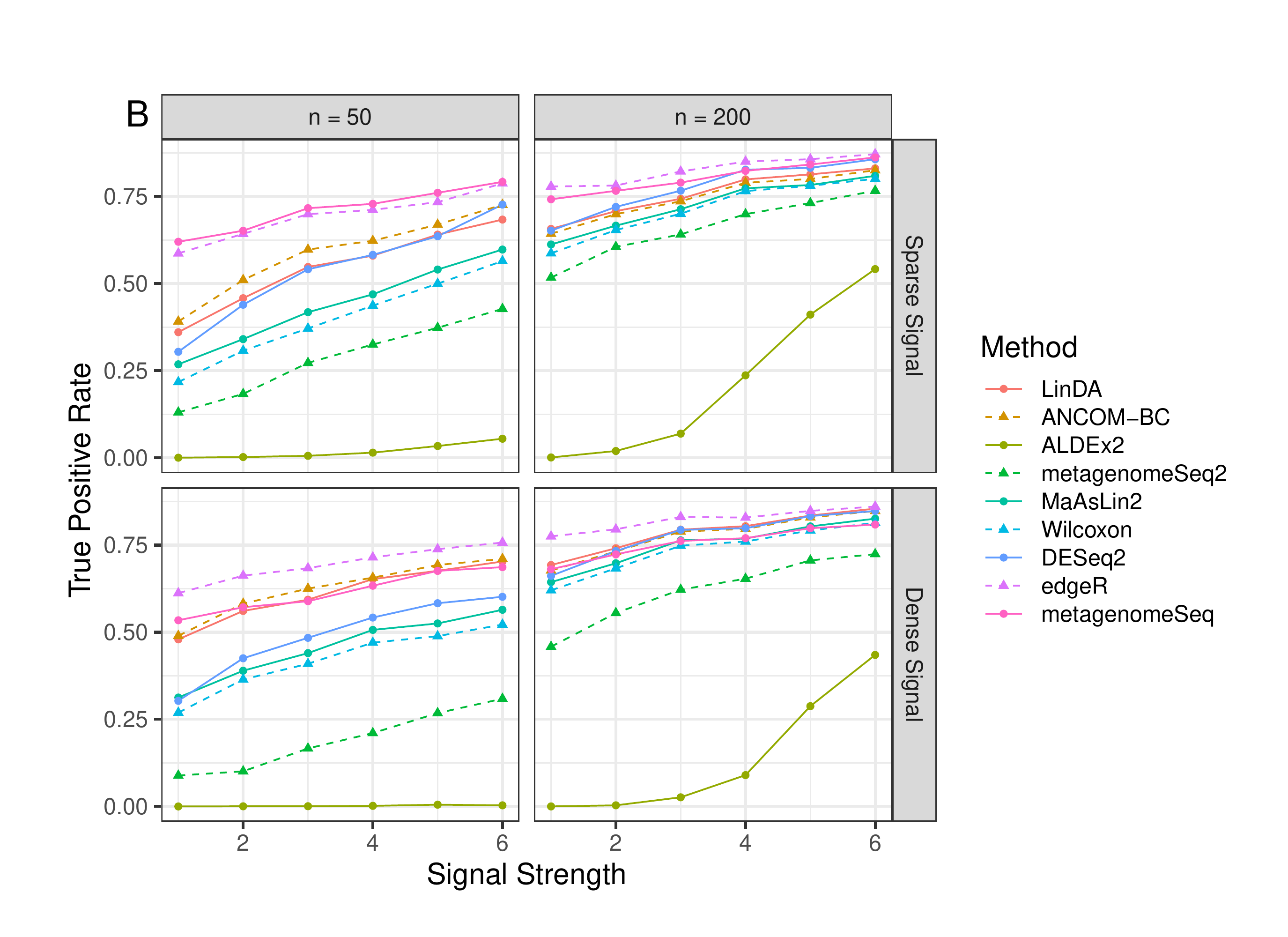}
	\end{subfigure}
	\caption{Full performance comparison (S0C0: log normal abundance distribution, a binary covariate). Empirical false discovery rate (A) and true positive rates (B) were averaged over 100 simulation runs. The dashed horizontal line (A) indicates the target FDR level of 0.05.}
	\label{fig-S0C0All}
\end{figure}

\begin{figure}
	\begin{subfigure}[b]{1\textwidth}
		\centering
		\includegraphics[scale=0.5]{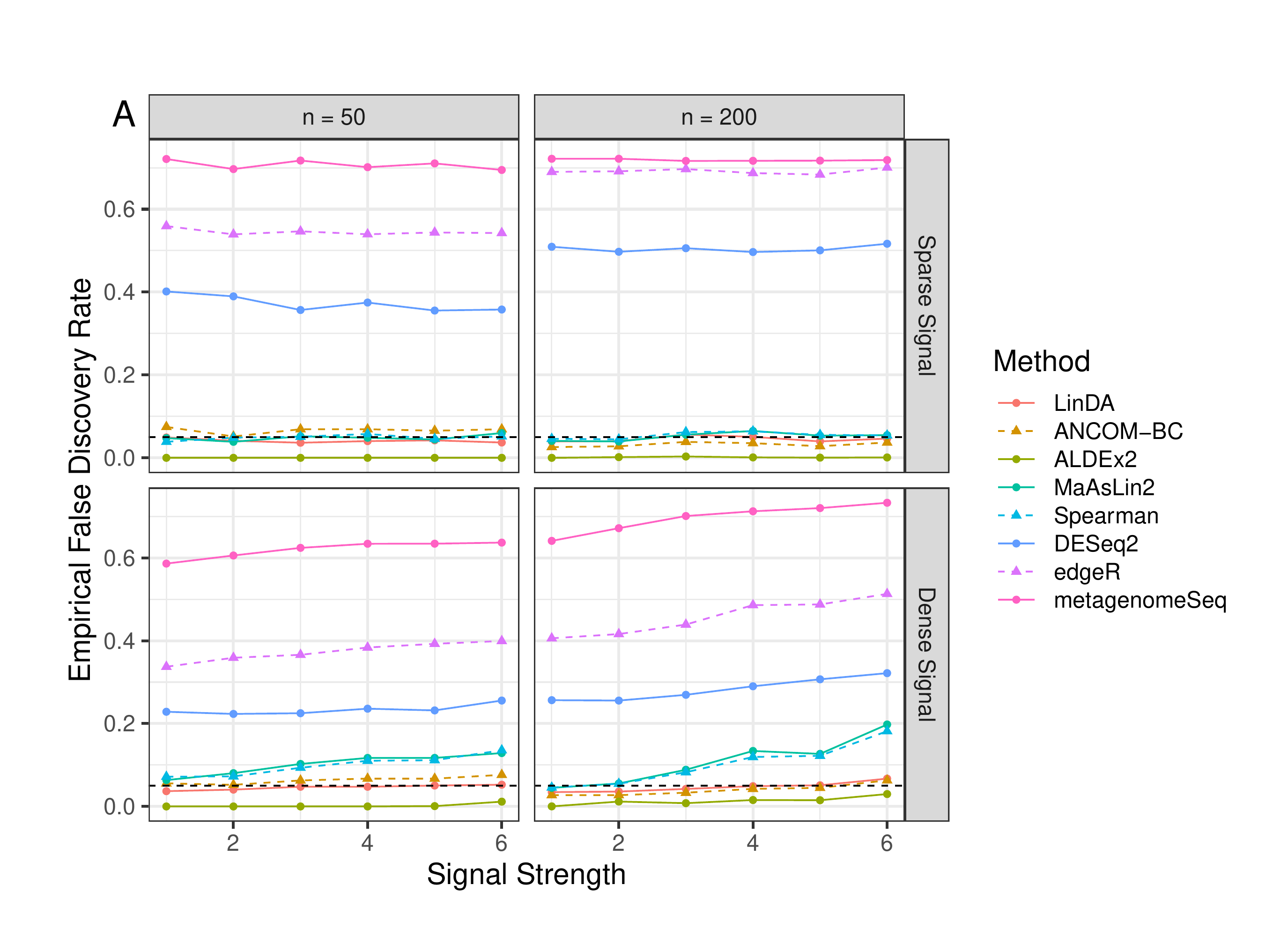}
	\end{subfigure}
	\begin{subfigure}[b]{1\textwidth}
		\centering
		\includegraphics[scale=0.5]{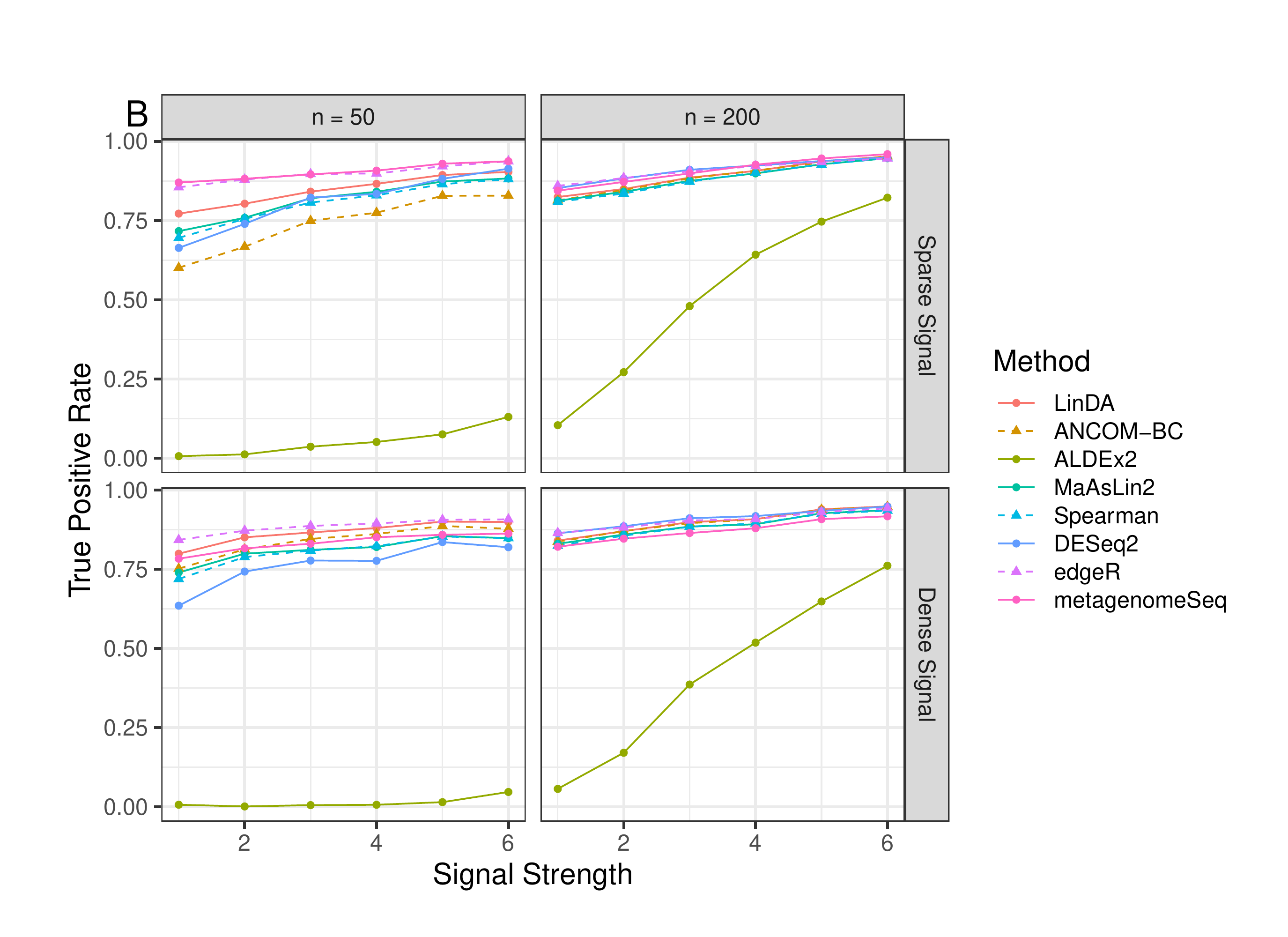}
	\end{subfigure}
	\caption{Full performance comparison (S0C1: log normal abundance distribution, a continuous covariate). Empirical false discovery rate (A) and true positive rates (B) were averaged over 100 simulation runs. The dashed horizontal line (A) indicates the target FDR level of 0.05.}
	\label{fig-S0C1All}
\end{figure}

\begin{figure}
	\begin{subfigure}[b]{1\textwidth}
		\centering
		\includegraphics[scale=0.5]{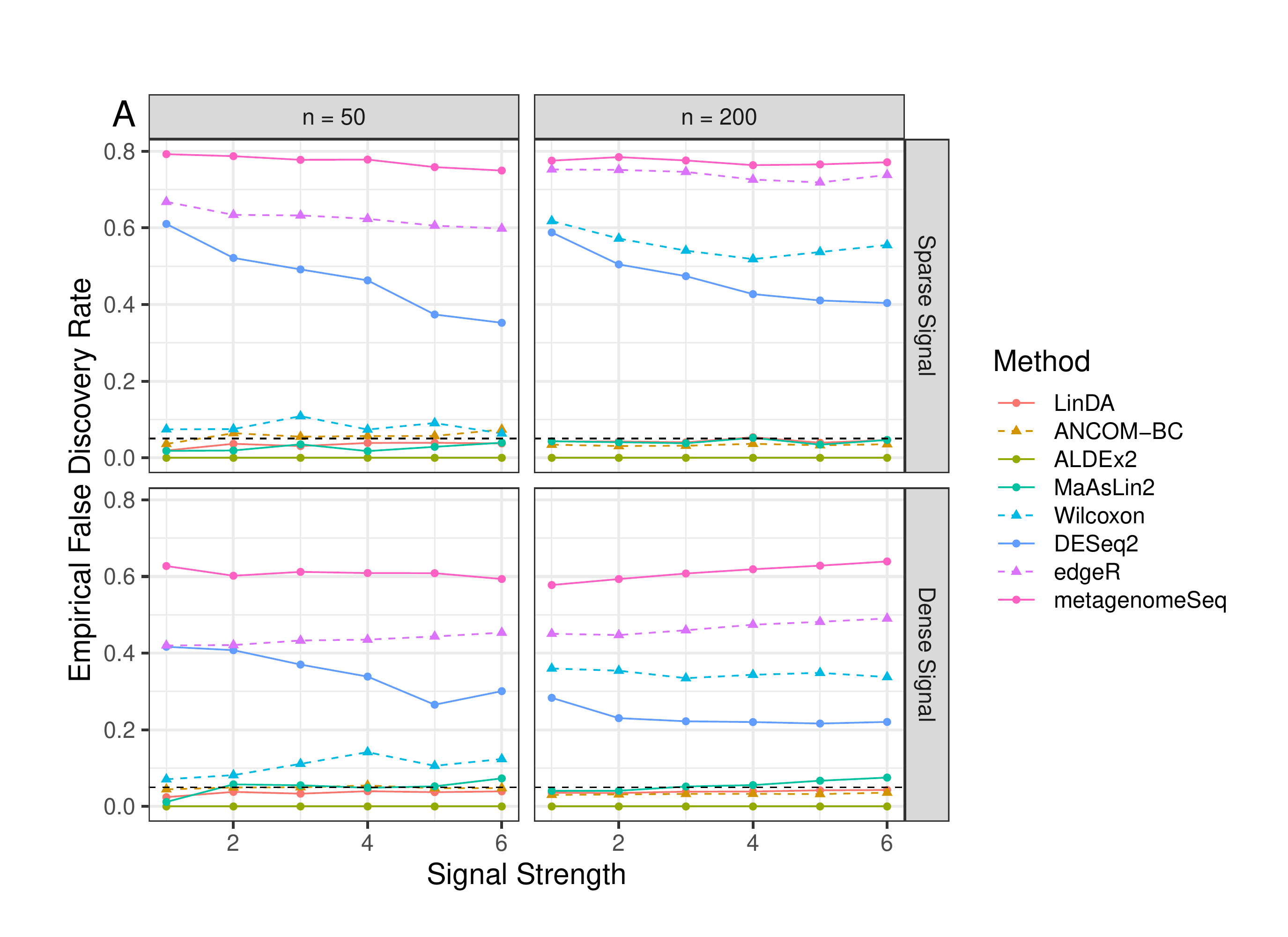}
	\end{subfigure}
	\begin{subfigure}[b]{1\textwidth}
		\centering
		\includegraphics[scale=0.5]{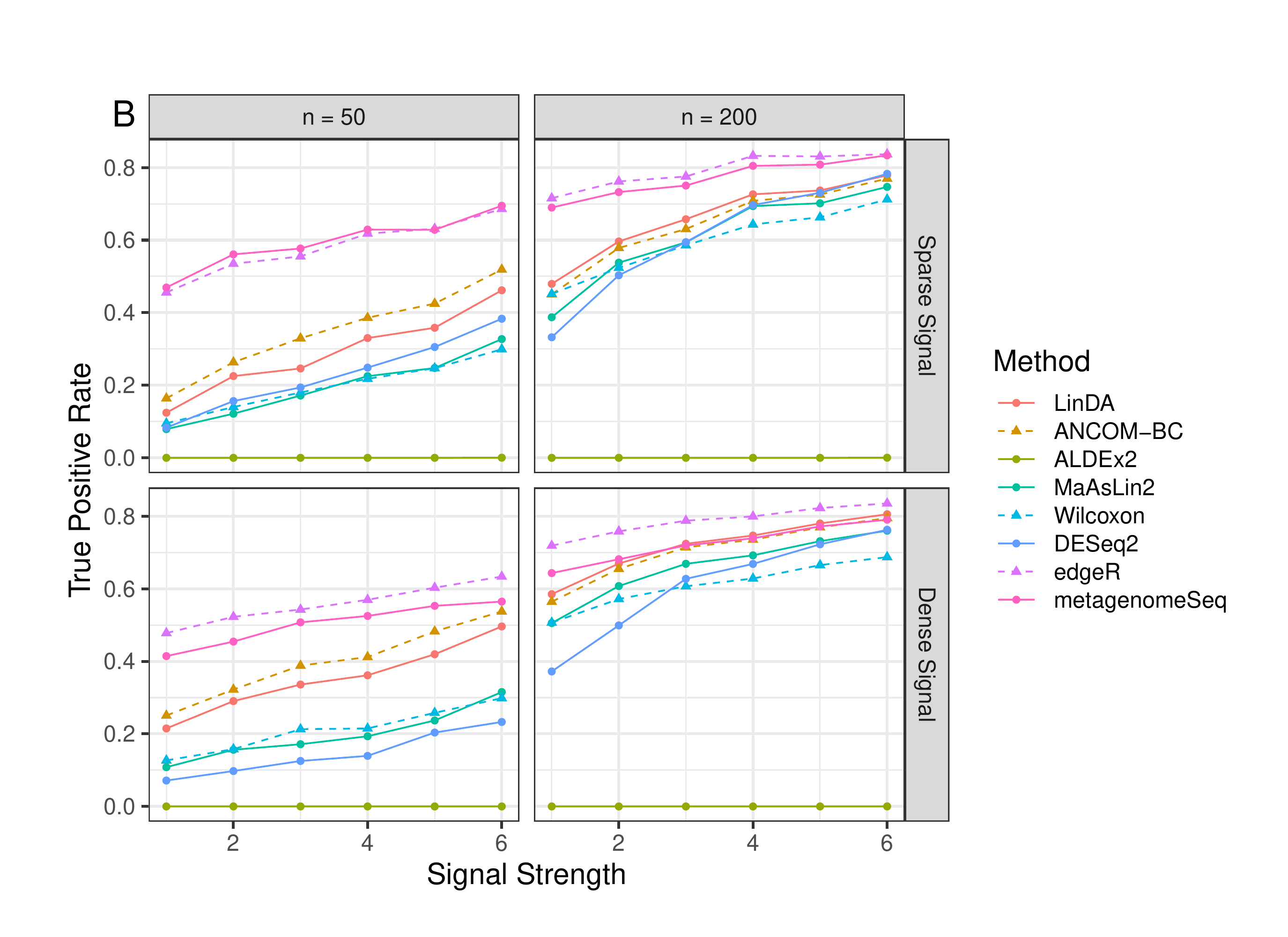}
	\end{subfigure}
	\caption{Full performance comparison (S0C2: log normal abundance distribution, a binary variable of interest and two confounders). Empirical false discovery rate (A) and true positive rates (B) were averaged over 100 simulation runs. The dashed horizontal line (A) indicates the target FDR level of 0.05.}
	\label{fig-S0C2All}
\end{figure}

\begin{figure}
	\begin{subfigure}[b]{1\textwidth}
		\centering
		\includegraphics[scale=0.5]{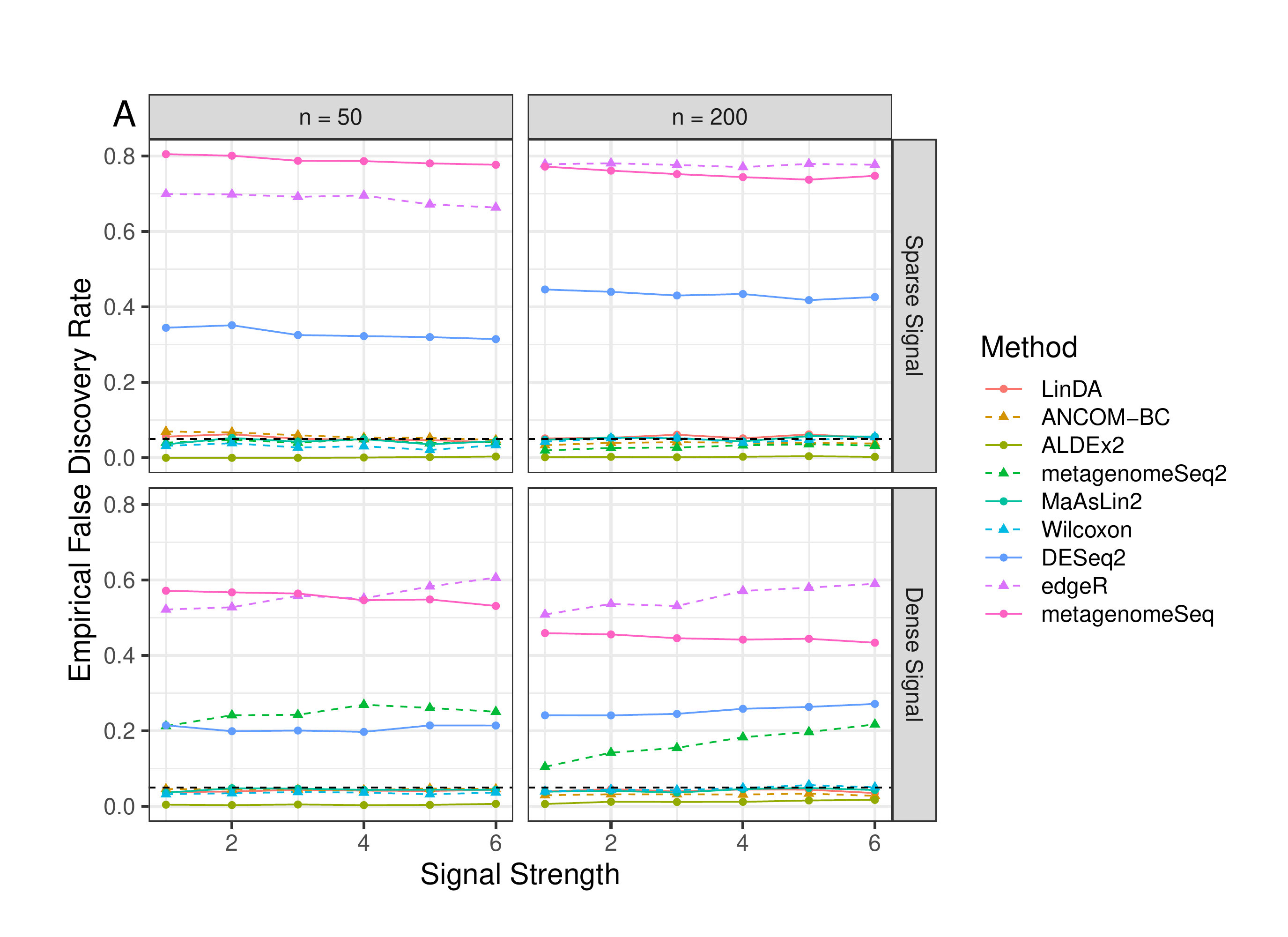}
	\end{subfigure}
	\begin{subfigure}[b]{1\textwidth}
		\centering
		\includegraphics[scale=0.5]{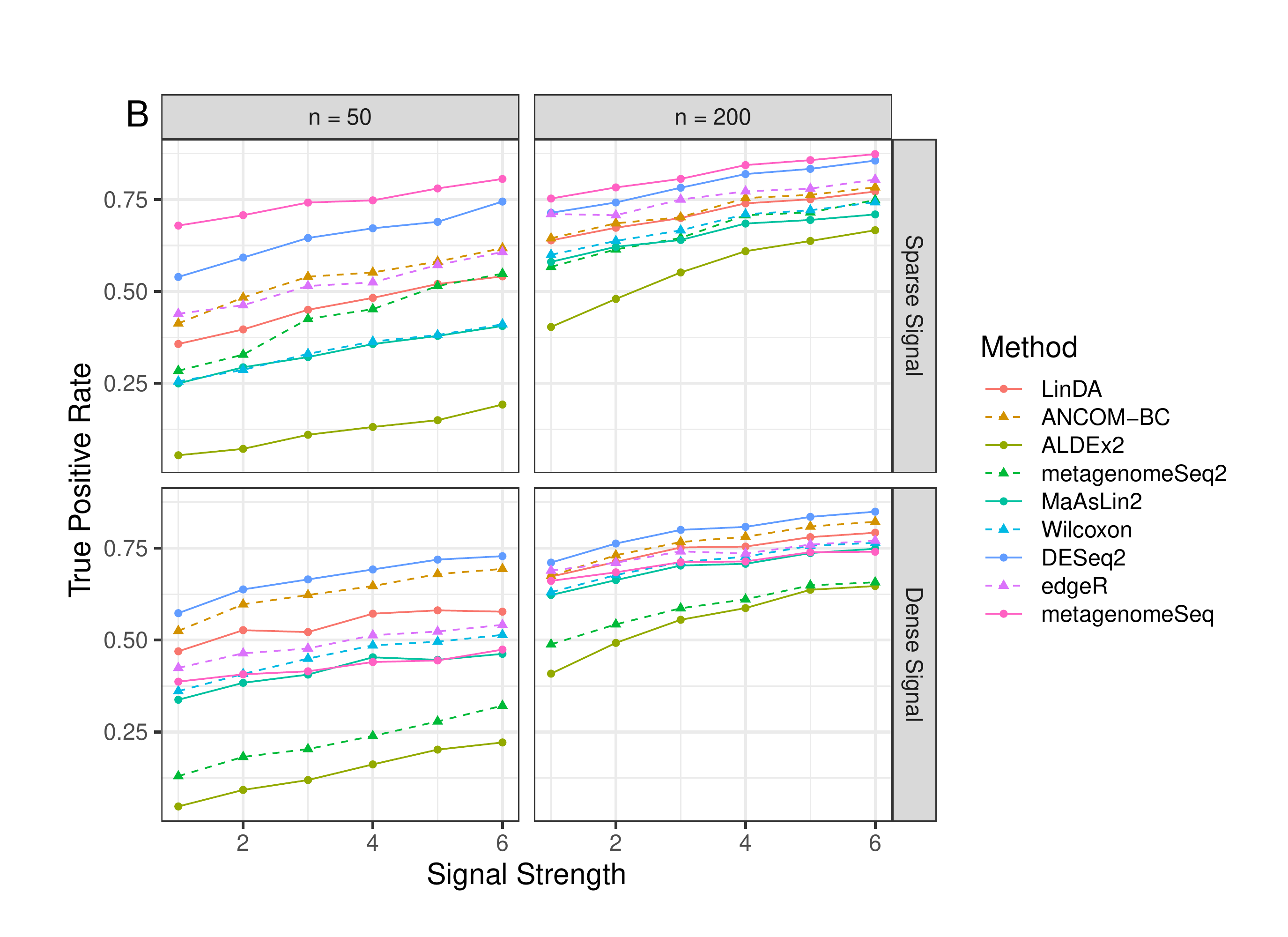}
	\end{subfigure}
	\caption{Full performance comparison (S1C0: zero inflated absolute abundances, a binary covariate). Empirical false discovery rate (A) and true positive rates (B) were averaged over 100 simulation runs. The dashed horizontal line (A) indicates the target FDR level of 0.05.}
	\label{fig-S1C0All}
\end{figure}

\begin{figure}
	\begin{subfigure}[b]{1\textwidth}
		\centering
		\includegraphics[scale=0.5]{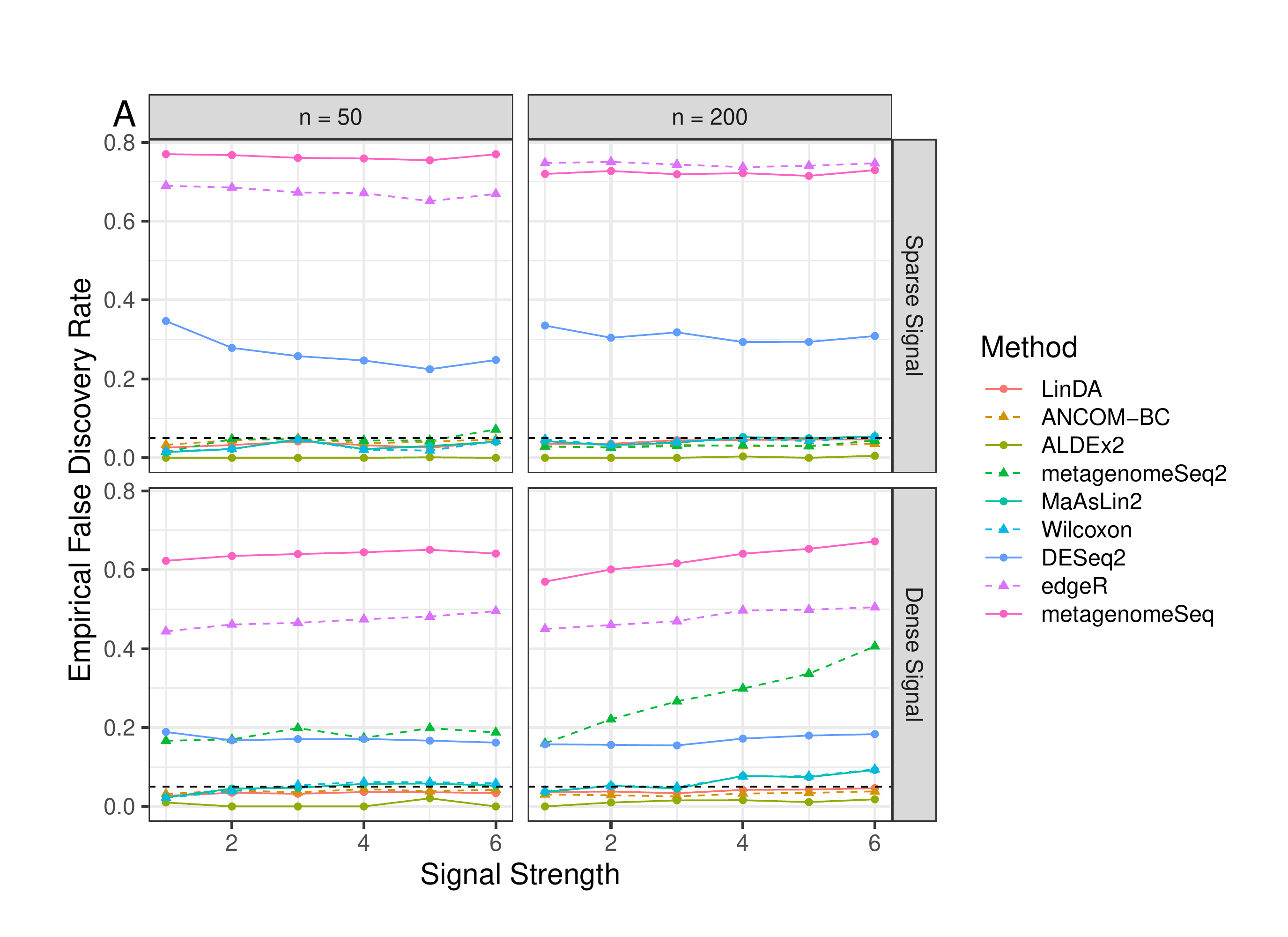}
	\end{subfigure}
	\begin{subfigure}[b]{1\textwidth}
		\centering
		\includegraphics[scale=0.5]{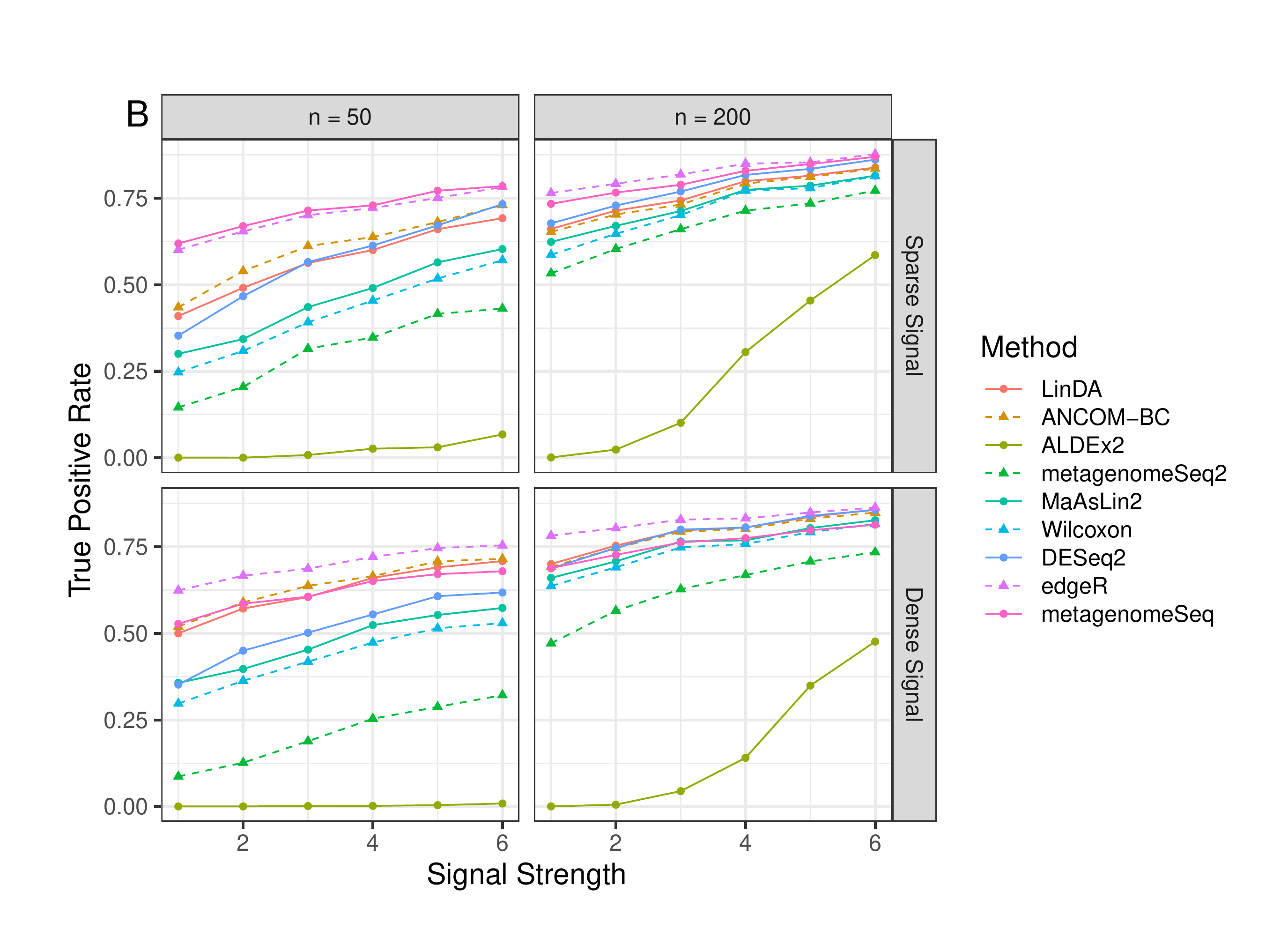}
	\end{subfigure}
	\caption{Full performance comparison (S2C0: correlated absolute abundances, a binary covariate). Empirical false discovery rate (A) and true positive rates (B) were averaged over 100 simulation runs. The dashed horizontal line (A) indicates the target FDR level of 0.05.}
	\label{fig-S2C0All}
\end{figure}

\begin{figure}
	\begin{subfigure}[b]{1\textwidth}
		\centering
		\includegraphics[scale=0.5]{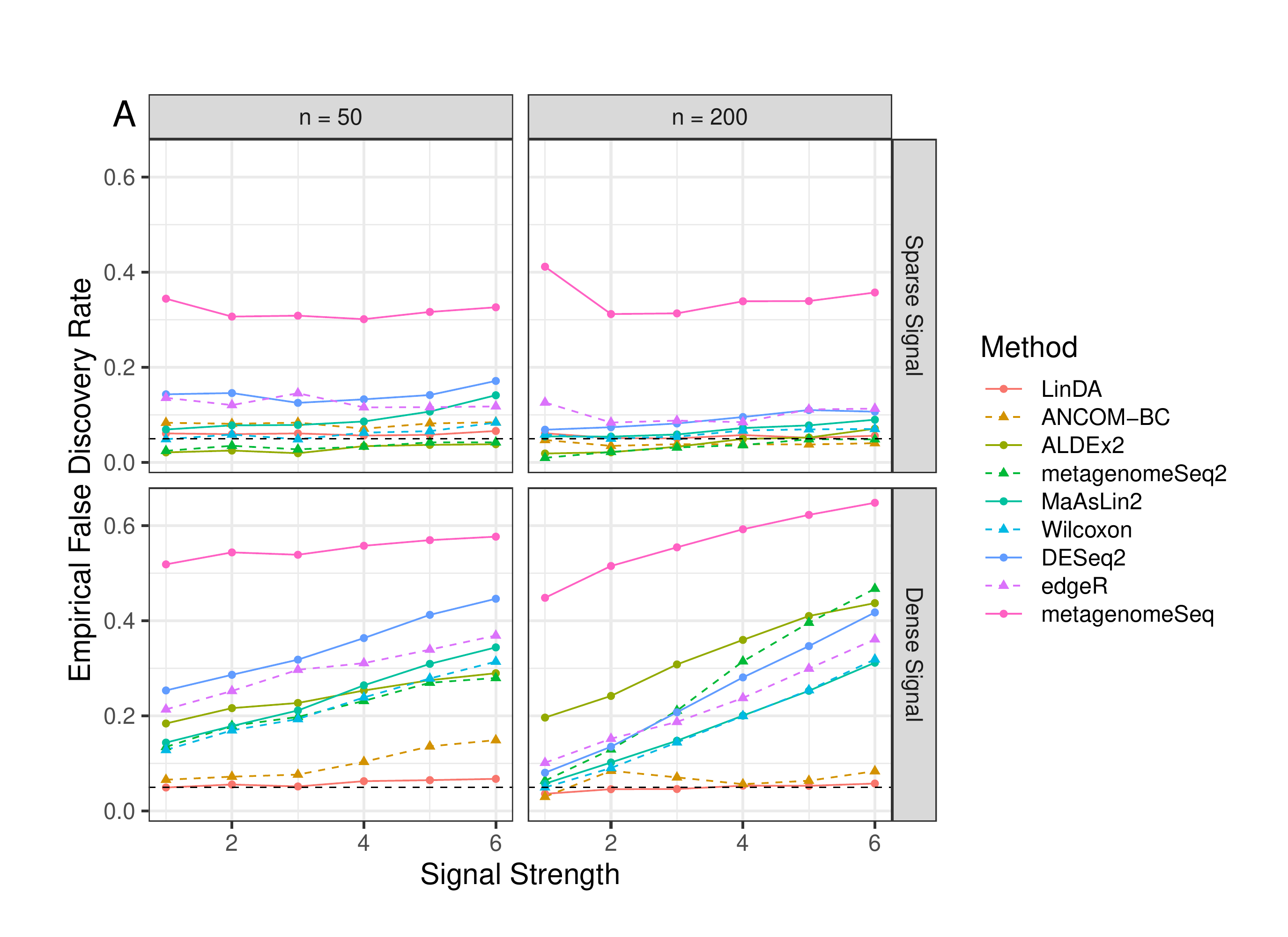}
	\end{subfigure}
	\begin{subfigure}[b]{1\textwidth}
		\centering
		\includegraphics[scale=0.5]{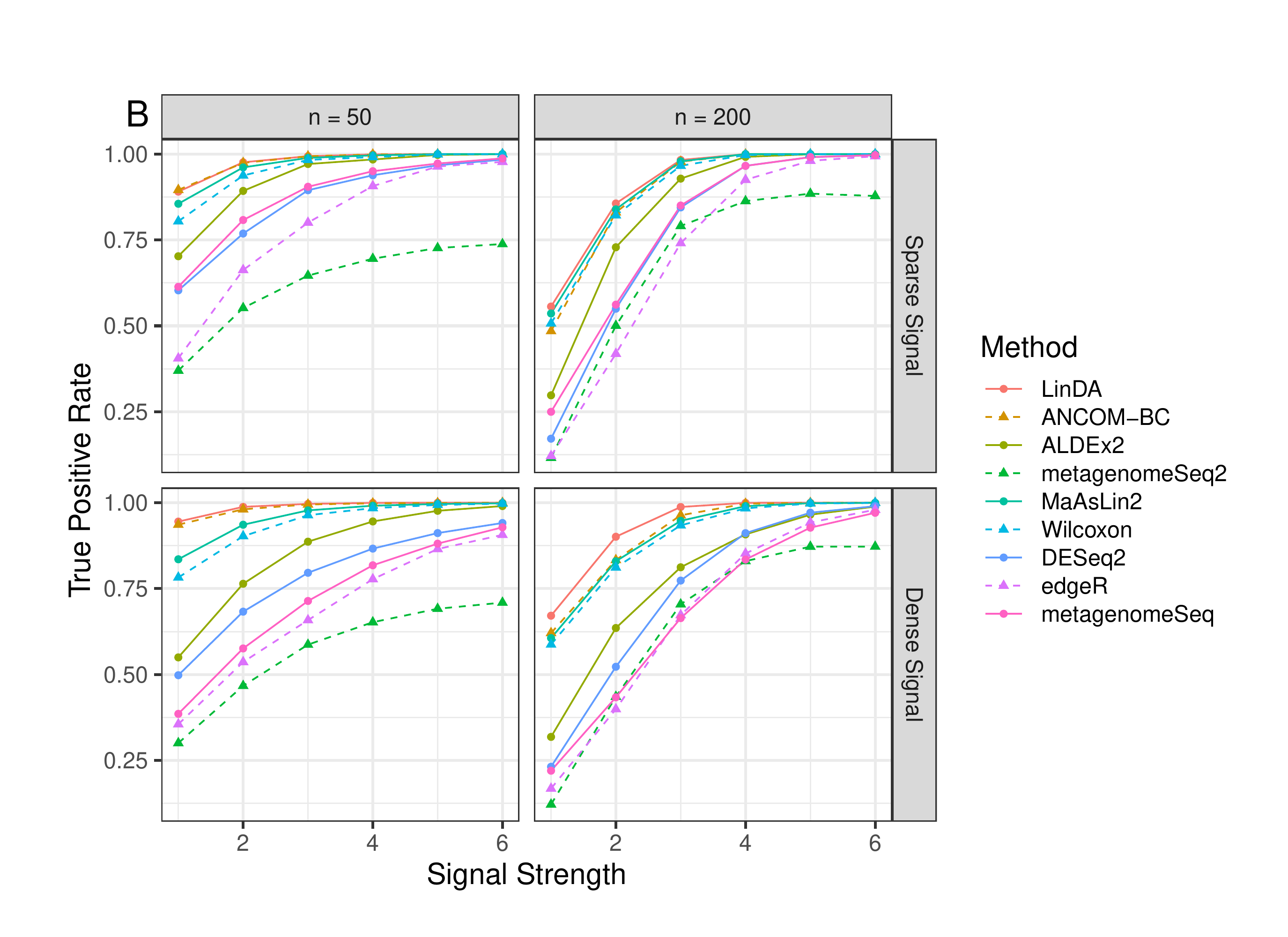}
	\end{subfigure}
	\caption{Full performance comparison (S3C0: gamma abundance distribution, a binary covariate). Empirical false discovery rate (A) and true positive rates (B) were averaged over 100 simulation runs. The dashed horizontal line (A) indicates the target FDR level of 0.05.}
	\label{fig-S3C0All}
\end{figure}

\begin{figure}
	\begin{subfigure}[b]{1\textwidth}
		\centering
		\includegraphics[scale=0.5]{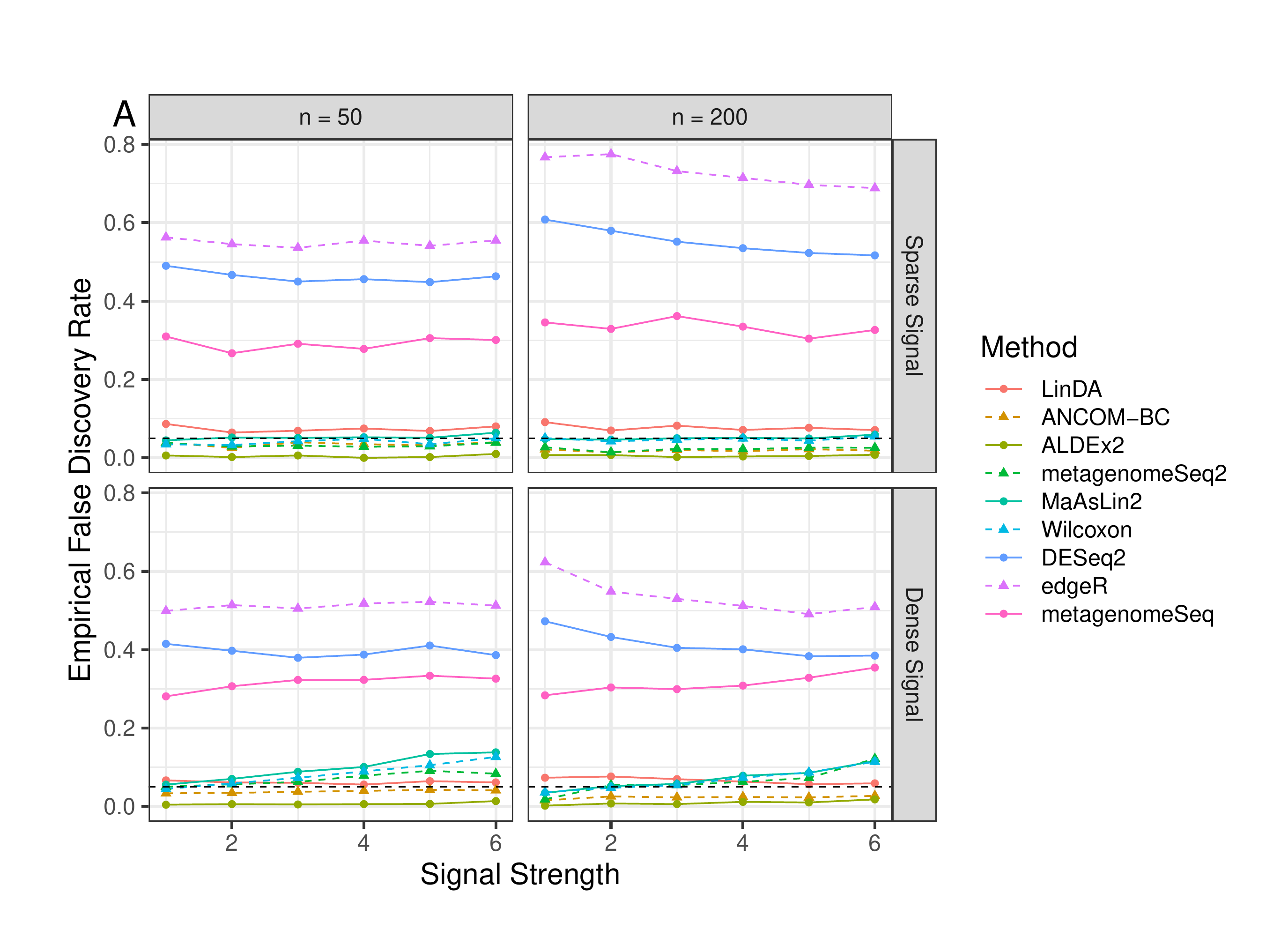}
	\end{subfigure}
	\begin{subfigure}[b]{1\textwidth}
		\centering
		\includegraphics[scale=0.5]{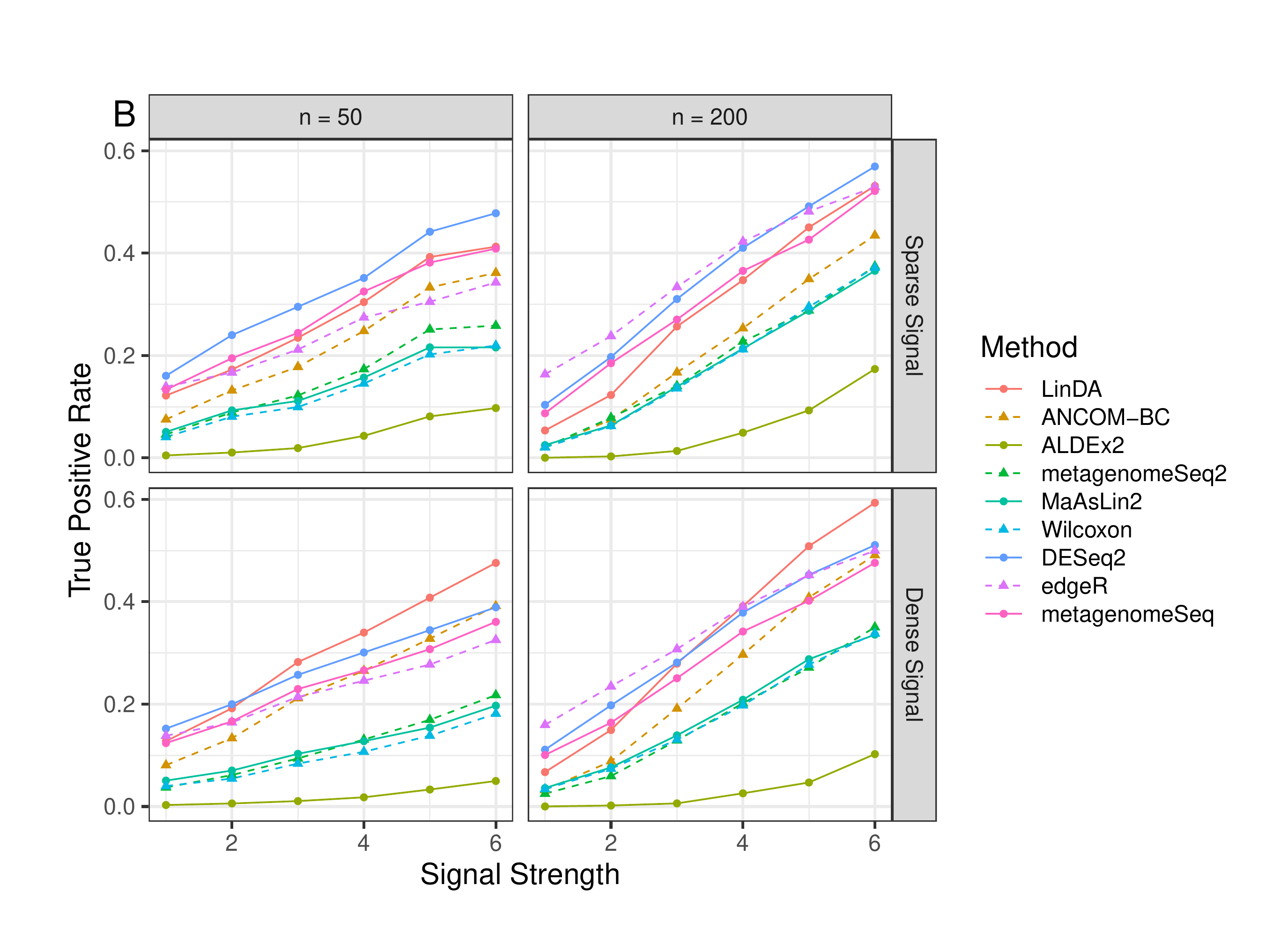}
	\end{subfigure}
	\caption{Full performance comparison (S4C0: smaller $m$, a binary covariate). Empirical false discovery rate (A) and true positive rates (B) were averaged over 1000 simulation runs. The dashed horizontal line (A) indicates the target FDR level of 0.05.}
	\label{fig-S4C0All}
\end{figure}

\begin{figure}
	\begin{subfigure}[b]{1\textwidth}
		\centering
		\includegraphics[scale=0.5]{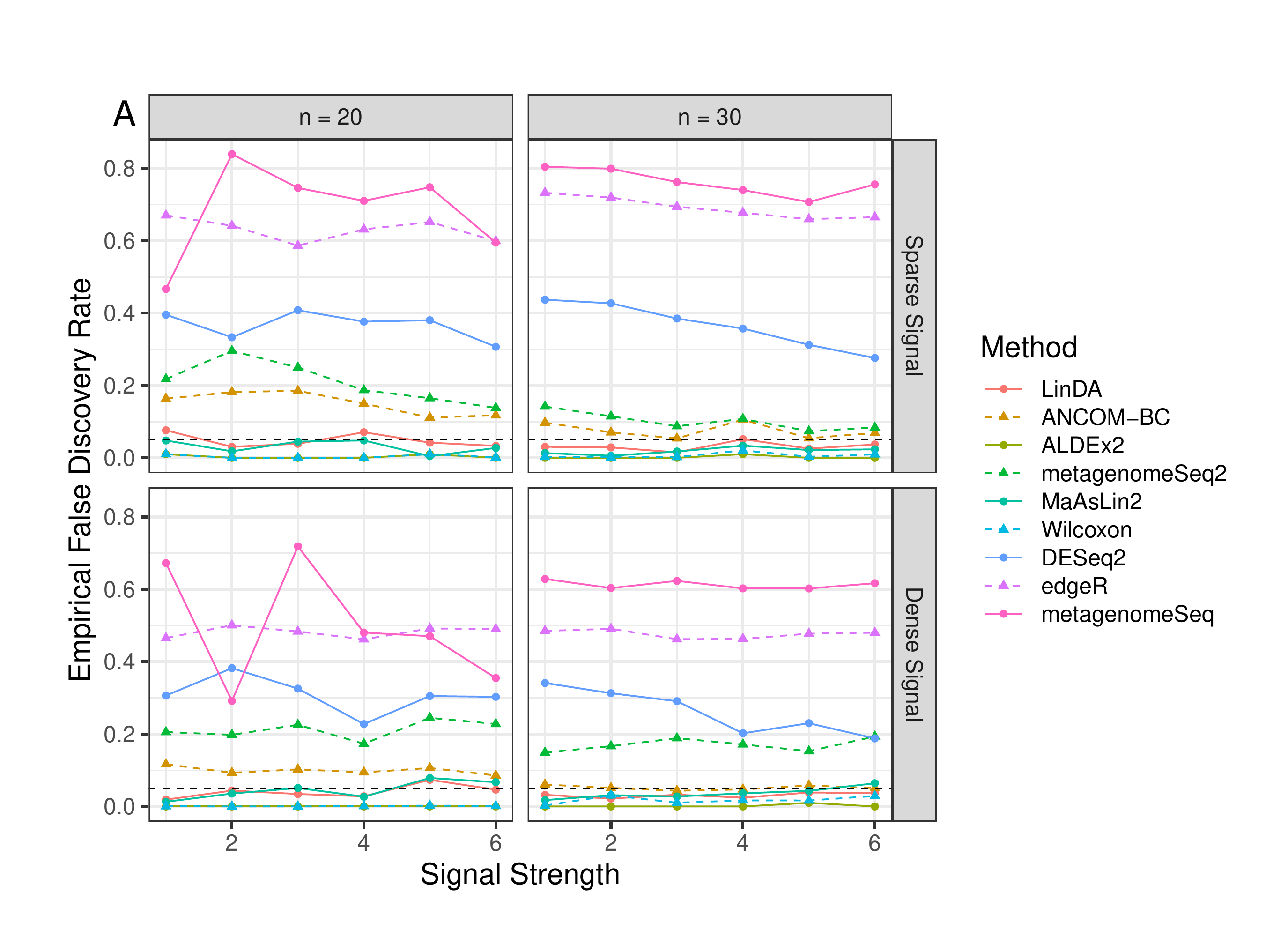}
	\end{subfigure}
	\begin{subfigure}[b]{1\textwidth}
		\centering
		\includegraphics[scale=0.5]{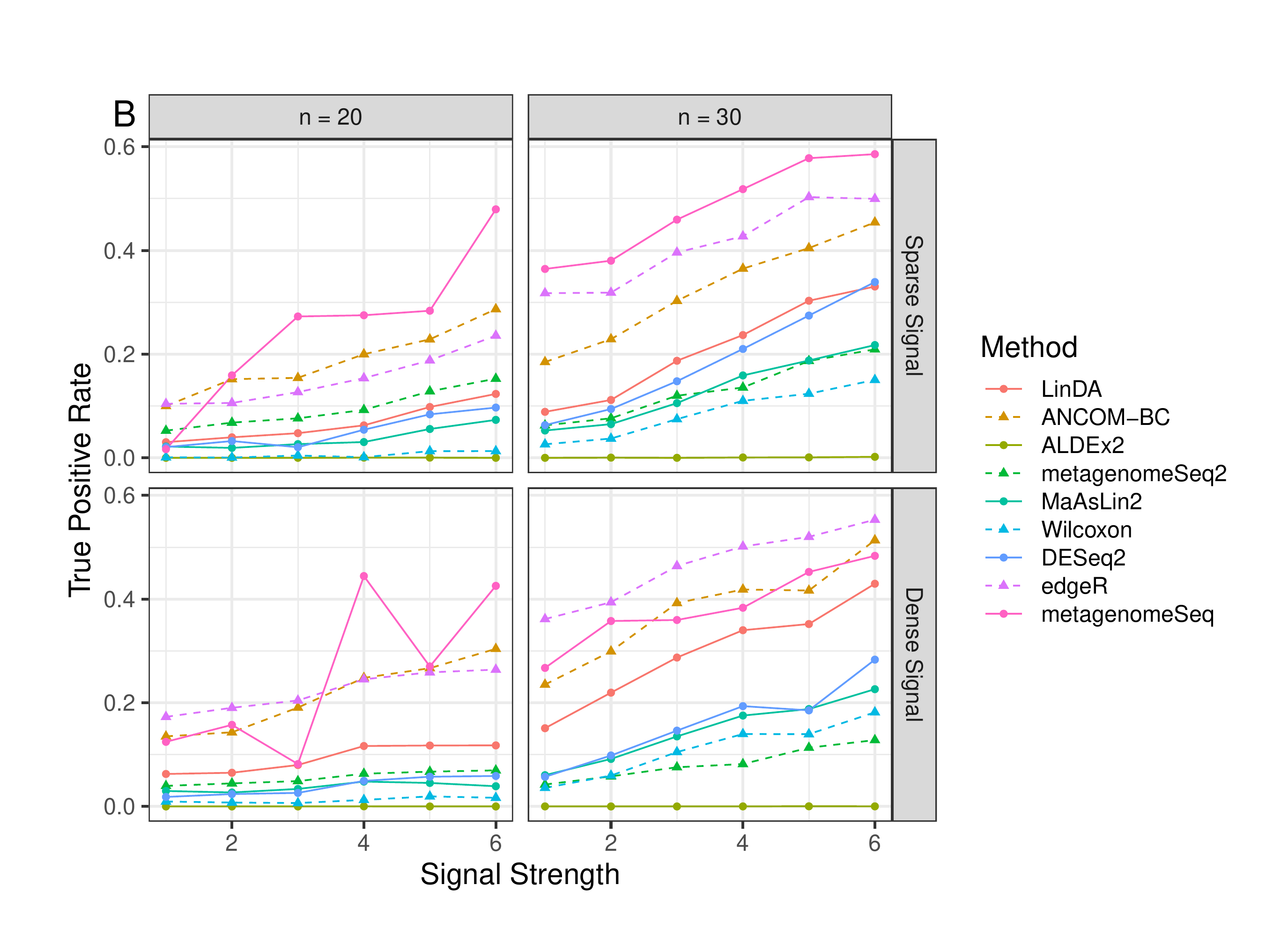}
	\end{subfigure}
	\caption{Full performance comparison (S5C0: smaller $n$, a binary covariate). Empirical false discovery rate (A) and true positive rates (B) were averaged over 100 simulation runs. The dashed horizontal line (A) indicates the target FDR level of 0.05.}
	\label{fig-S5C0All}
\end{figure}

\begin{figure}
	\begin{subfigure}[b]{1\textwidth}
		\centering
		\includegraphics[scale=0.5]{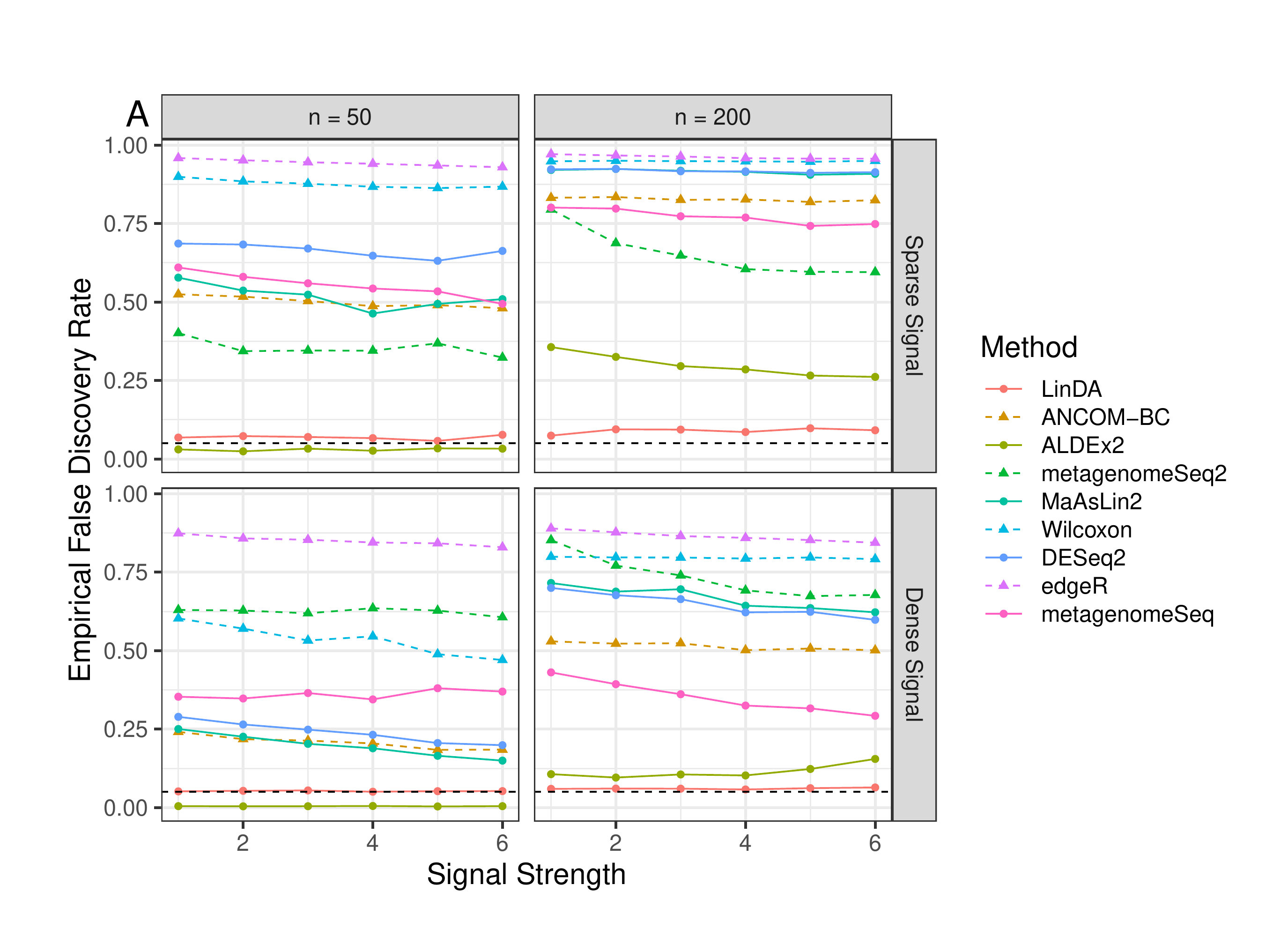}
	\end{subfigure}
	\begin{subfigure}[b]{1\textwidth}
		\centering
		\includegraphics[scale=0.5]{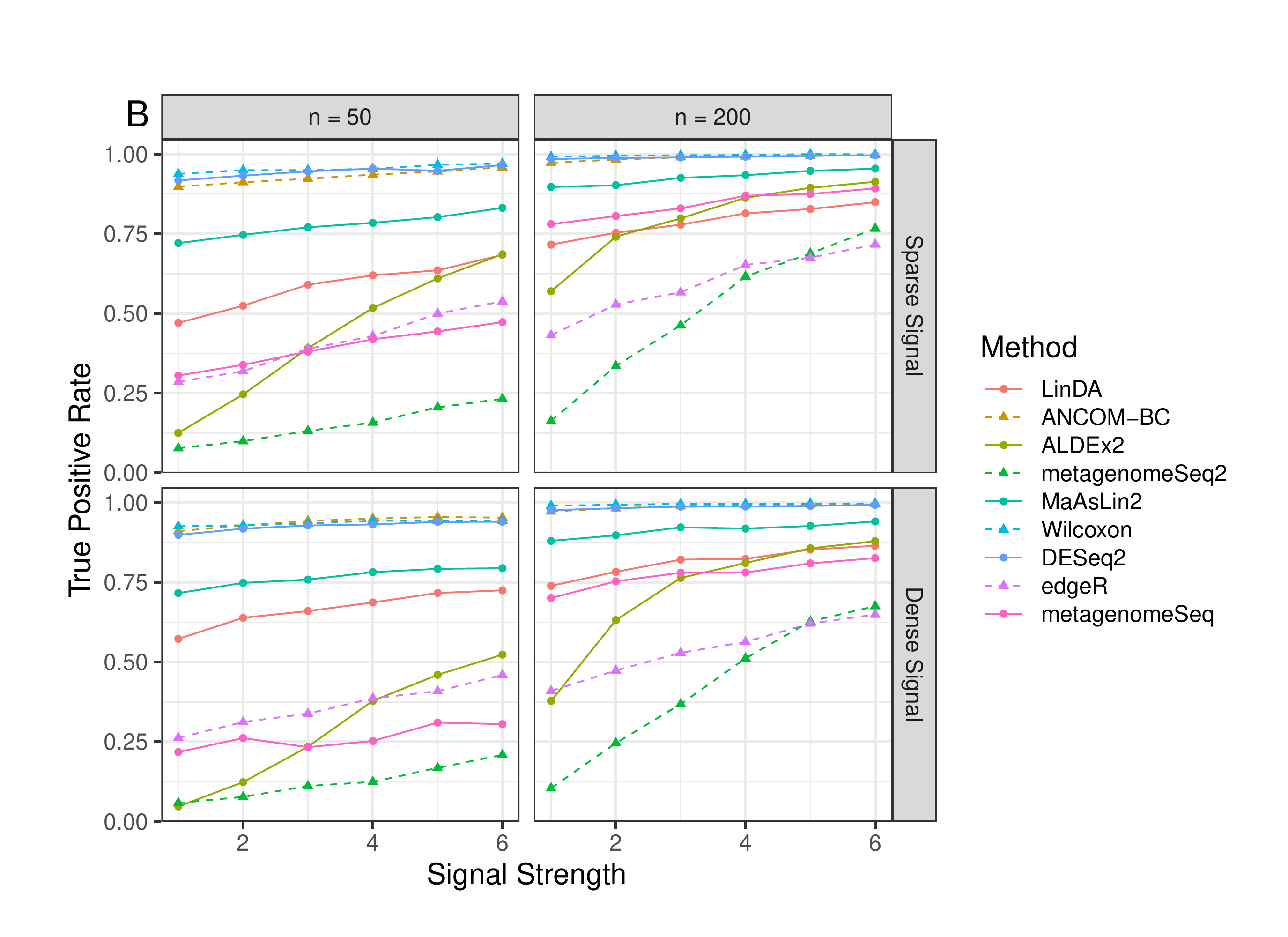}
	\end{subfigure}
	\caption{Full performance comparison (S6C0: 10-fold difference in library size, a binary covariate). Empirical false discovery rate (A) and true positive rates (B) were averaged over 100 simulation runs. The dashed horizontal line (A) indicates the target FDR level of 0.05.}
	\label{fig-S6C0All}
\end{figure}

\begin{figure}
	\begin{subfigure}[b]{1\textwidth}
		\centering
		\includegraphics[scale=0.5]{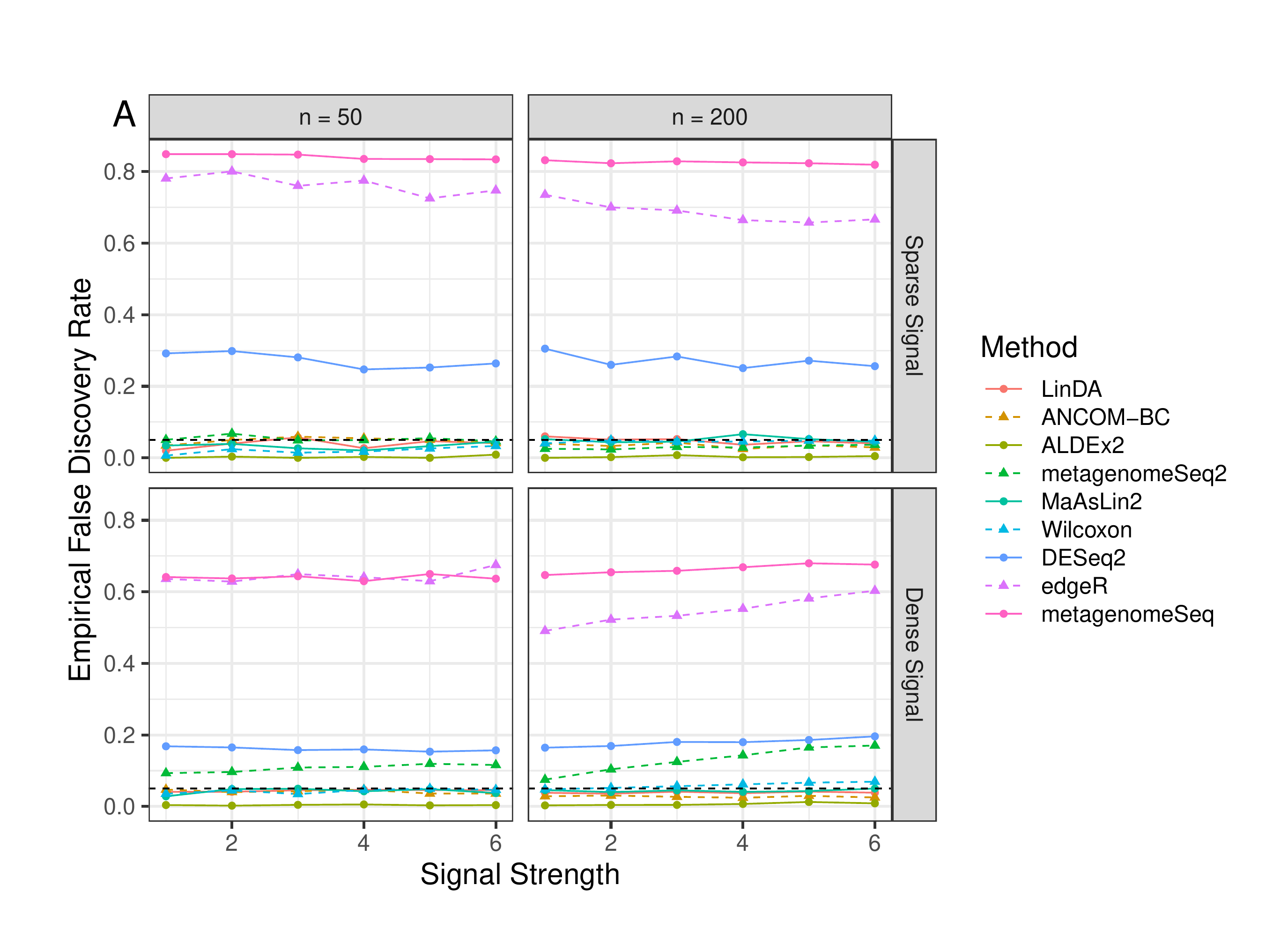}
	\end{subfigure}
	\begin{subfigure}[b]{1\textwidth}
		\centering
		\includegraphics[scale=0.5]{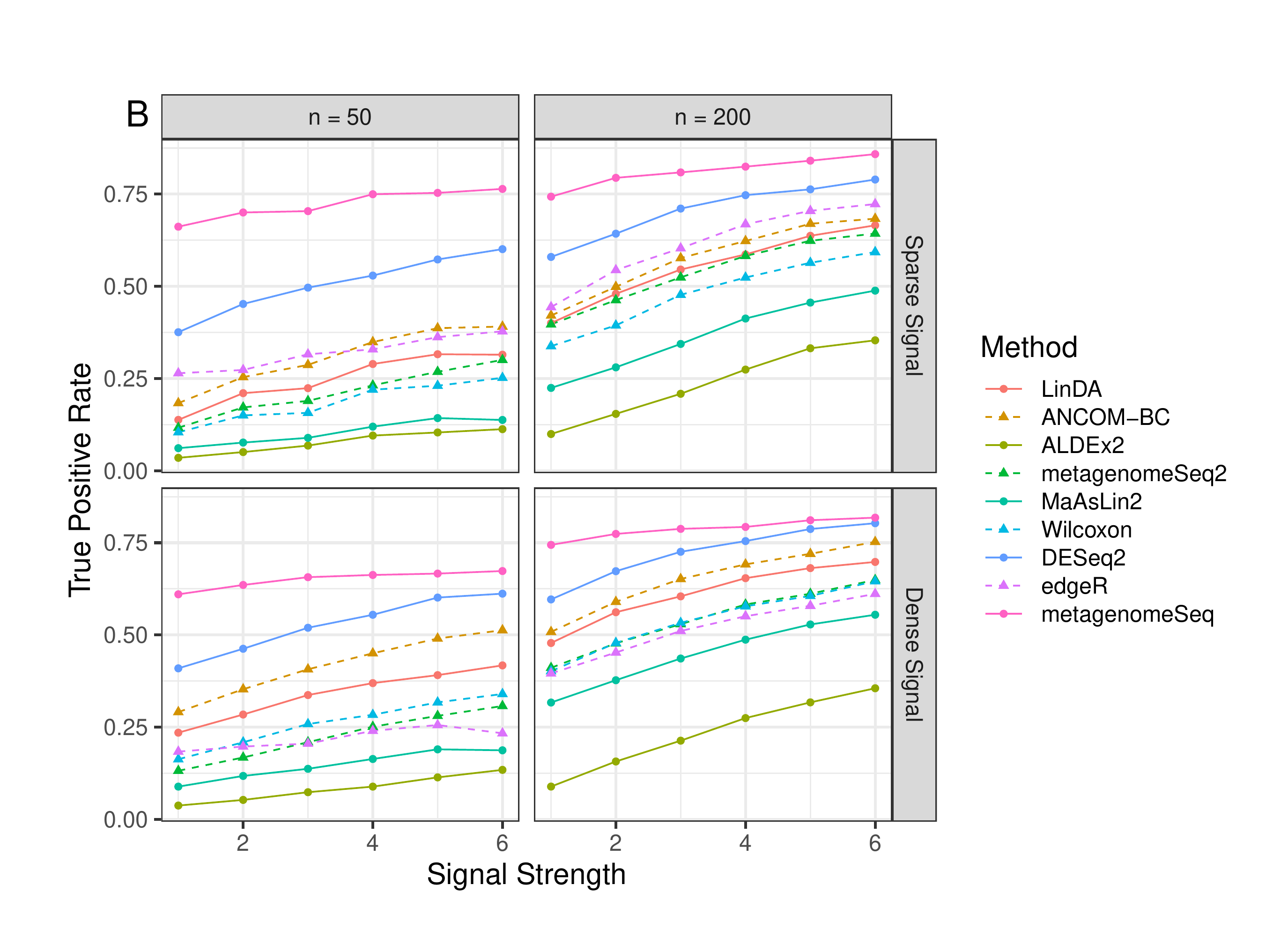}
	\end{subfigure}
	\caption{Full performance comparison (S7C0: negative binomial abundance distribution, a binary covariate). Empirical false discovery rate (A) and true positive rates (B) were averaged over 100 simulation runs. The dashed horizontal line (A) indicates the target FDR level of 0.05.}
	\label{fig-S7C0All}
\end{figure}

\begin{figure}
	\begin{subfigure}[b]{1\textwidth}
		\centering
		\includegraphics[scale=0.5]{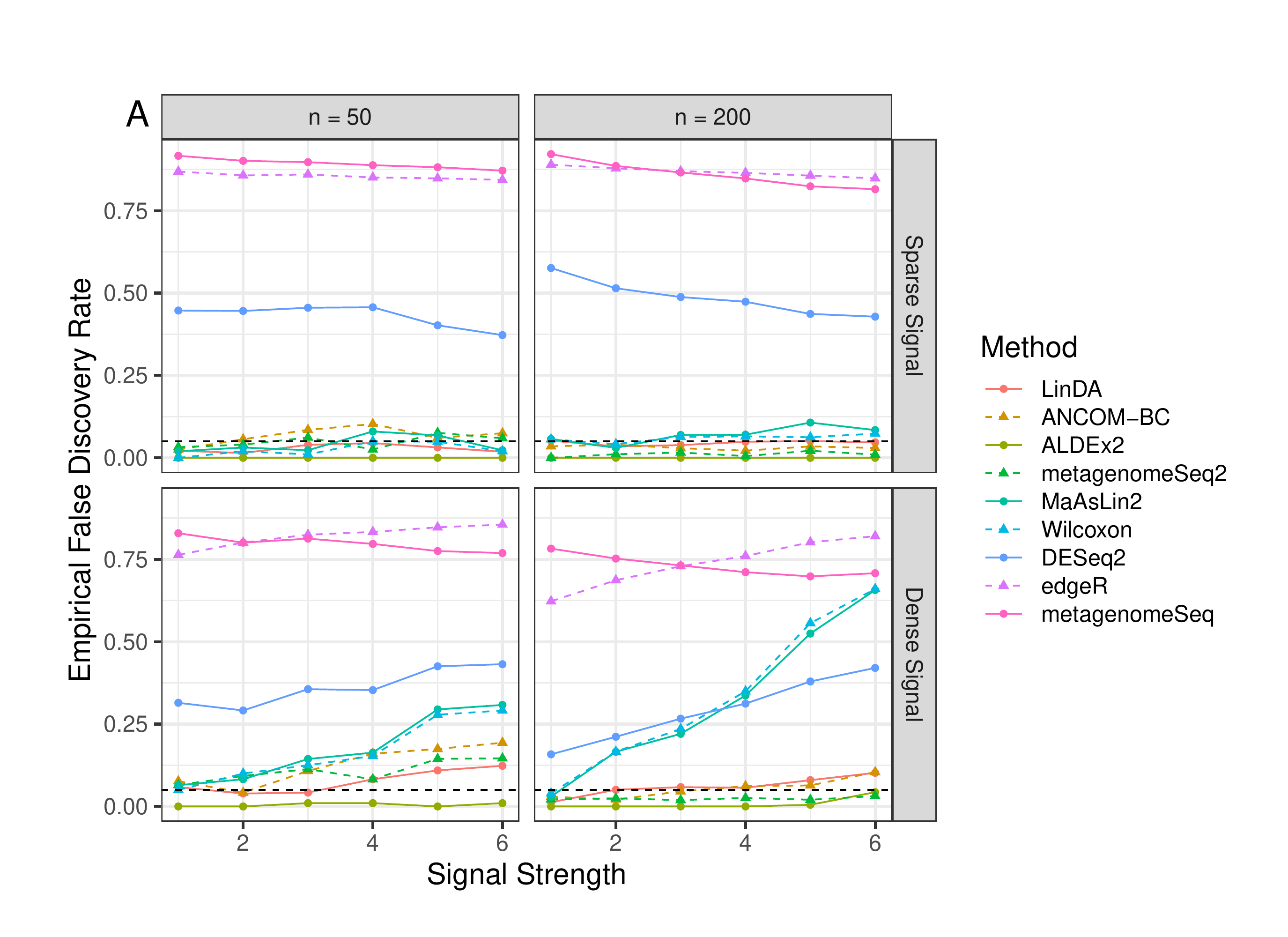}
	\end{subfigure}
	\begin{subfigure}[b]{1\textwidth}
		\centering
		\includegraphics[scale=0.5]{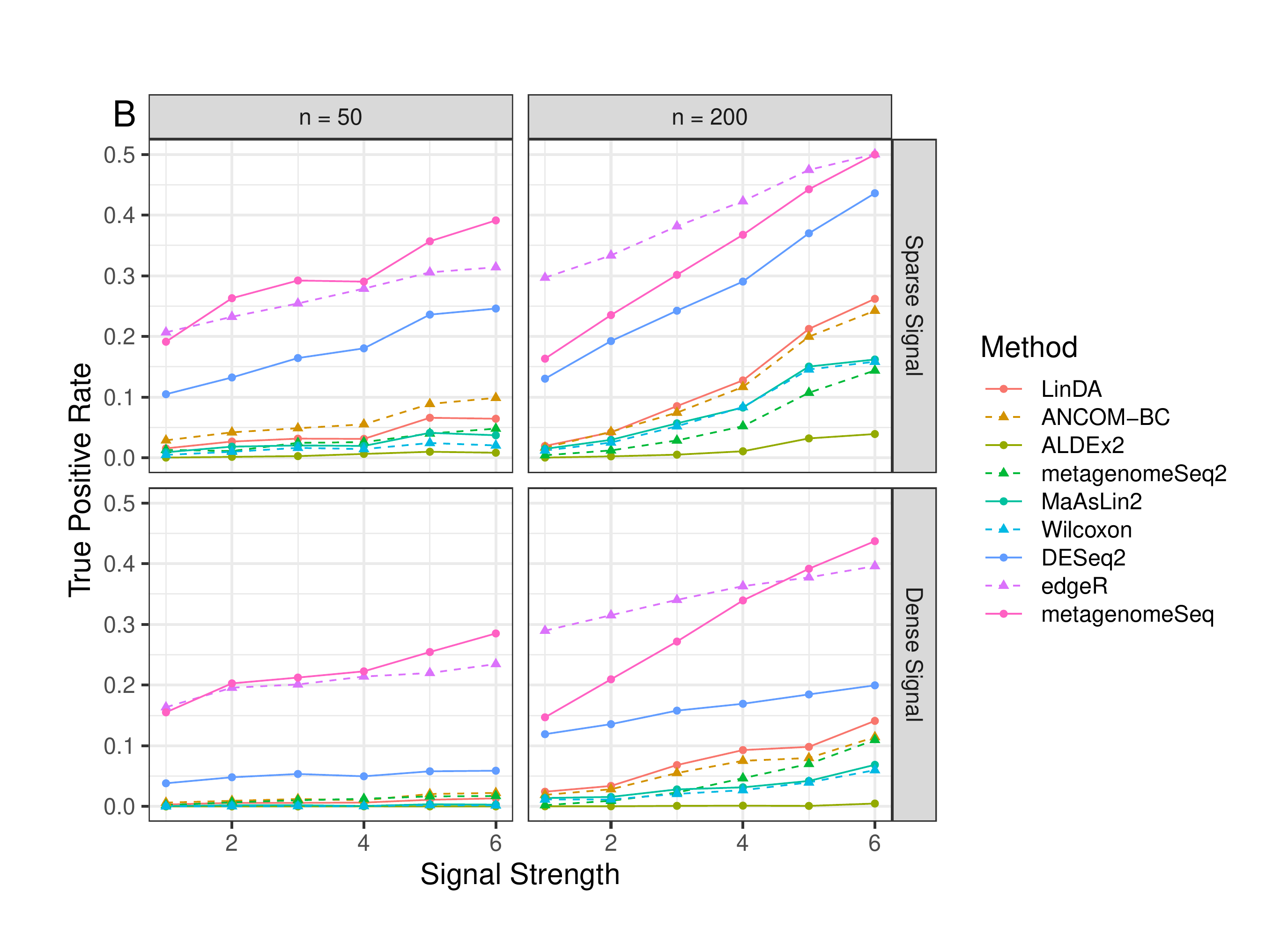}
	\end{subfigure}
	\caption{Full performance comparison (S0C0 with strong compositional effects). Empirical false discovery rate (A) and true positive rates (B) were averaged over 100 simulation runs. The dashed horizontal line (A) indicates the target FDR level of 0.05.}
	\label{fig-S0C0StrongAll}
\end{figure}


\begin{thebibliography}{7}
	
\bibitem[{Fan \& Pedersen(2021)}]{Fan:2021cf}
Fan, Y. and Pedersen, O. (2021).
Gut microbiota in human metabolic health and disease.
\textit{Nature Reviews Microbiology} \textbf{19}, 55--71.

	\bibitem[Valdes et al.(2018)]{Valdes:2018fx}
Valdes, A. M., Walter,  J., Segal, E., and Spector, T.~M. (2018).
Role of the gut microbiota in nutrition and health.
\textit{BMJ} \textbf{361}, k2179.

\bibitem[{Edgar(2013)}]{Edgar:2013}
Edgar, R. C. (2013).
UPARSE: highly accurate OTU sequences from microbial amplicon reads.
\textit{Nature Methods} \textbf{10}, 996--998.

	\bibitem[{Callahan et al.(2016)}]{Callahan:2016}
Callahan, B. J., McMurdie, P. J., Rosen, M. J., Han, A. W., Johnson, A. J. A., and Holmes, S. P. (2016).
DADA2: High-resolution sample inference from Illumina amplicon data.
\textit{Nature Methods} \textbf{13}, 581--583.

	\bibitem[{Segata et al.(2012)}]{Segata:2012}
Segata, N., Waldron, L., Ballarini, A., Narasimhan, V., Jousson, O., and Huttenhower, C. (2012).
Metagenomic microbial community profiling using unique clade-specific marker genes.
\textit{Nature Methods} \textbf{9}, 811--814.

	\bibitem[{Gloor et al.(2017)}]{Gloor:2017}
Gloor, G. B., Macklaim, J. M., Pawlowsky-Glahn, V., and Egozcue, J. J (2017).
Microbiome datasets are compositional: and this is not optional.
\textit{Frontiers in Microbiology} \textbf{8}, 2224.

	\bibitem[Tsilimigras \& Fodor(2016)]{Tsilimigras:2016}
Tsilimigras, M. C. B. and Fodor, A. A. (2016).
Compositional data analysis of the microbiome: fundamentals, tools, and challenges.
\textit{Annals of Epidemiology} \textbf{26}, 330--335.

	\bibitem[Morton et al.(2019)]{Morton:2019}
Morton, J. T., Marotz, C., Washburne, A., Silverman, J., Zaramela, L. S., Edlund, A., Zengler, K., and Knight, R. (2019).
Establishing microbial composition measurement standards with reference frames. 
\textit{Nature communications} \textbf{10}, 2719.

	\bibitem[Xiao et al.(2018:1)]{Xiao:20181}
Xiao, J., Chen, L., Yu, Y., Zhang, X., and Chen, J. (2018).
A phylogeny-regularized sparse regression model for predictive modeling of microbial community data.
\textit{Frontiers in Microbiology} \textbf{9}, 3112.

	\bibitem[Xiao et al.(2018:2)]{Xiao:20182}
Xiao, J., Chen, L., Johnson, S., Yu, Y., Zhang, X., and Chen, J. (2018).
Predictive modeling of microbiome data using a phylogeny-regularized generalized linear mixed model.
\textit{Frontiers in Microbiology} \textbf{9}, 1391.

	\bibitem[{Robinson \& Oshlack(2010)}]{Robinson:2010}
Robinson, M. D. and Oshlack, A. (2010).
A scaling normalization method for differential expression analysis of RNA-seq data.
\textit{Genome Biology} \textbf{11}, R25.

	\bibitem[{Anders \& Huber(2010)}]{Anders:2010}
Anders, S. and Huber, W. (2010).
Differential expression analysis for sequence count data.
\textit{Genome Biology} \textbf{11}, R106.

	\bibitem[Paulson et al.(2013)]{Paulson:2013}
Paulson, J. N., Stine, O. C., Bravo, H. C., and Pop, M. (2013).
Differential abundance analysis for microbial marker-gene surveys.
\textit{Nature Methods} \textbf{10}, 1200--1202.

	\bibitem[{Chen et al.(2018)}]{Chen:2018}
Chen, L., Reeve, J., Zhang, L., Huang, S., Wang, X., and Chen, J. (2018).
GMPR: A robust normalization method for zero-inflated count data with application to microbiome sequencing data.
\textit{PeerJ} \textbf{6}, e4600.

	\bibitem[Robinson et al.(2010)]{Robinson:2010-1}
Robinson, M. D., McCarthy, D. J., and Smyth, G. K. (2010).
edgeR: a Bioconductor package for differential expression analysis of digital gene expression data.
\textit{Bioinformatics} \textbf{26}, 139--140.

	\bibitem[{Love et al.(2014)}]{Love:2014}
Love, M. I., Huber, W., and Anders, S. (2014).
Moderated estimation of fold change and dispersion for RNA-seq data with DESeq2.
\textit{Genome Biology} \textbf{15}, 550.

	\bibitem[{Chen et al.(2018)}]{Chenj:2018}
Chen, J., King, E., Deek, R., Wei, Z., Yu, Y., Grill, D., Ballman, K., and Stegle, O. (2018).
An omnibus test for differential distribution analysis of microbiome sequencing data.
\textit{Bioinformatics} \textbf{34}, 643--651.

	\bibitem[{Mallick et al.(2020)}]{Mallick:2020}
Mallick H., Rahnavard A., Mclver L. J., Ma S., Zhang Y., Nguyen L. H., Tickle T. L., Weingart G., Ren B., Schwager E. H., et al. (2021).
Multivariable association discovery in population-scale meta-omics studies. \textit{PLoS computational biology} \textbf{17},  e1009442.

	\bibitem[{Sohn et al.(2015)}]{Sohn:2015}
Sohn, M. B., Du, R., and An, L. (2015).
A robust approach for identifying differentially
abundant features in metagenomic samples.
\textit{Bioinformatics} \textbf{31}, 2269--2275.

	\bibitem[{Brill et al.(2020)}]{Brill:2020}
Brill, B., Amir, A., and Heller, R. (2020).
Testing for differential abundance in compositional counts data, with application to microbiome studies.
\textit{arXiv preprint} arXiv:1904.08937.

	\bibitem[{Aitchison(1986)}]{Aitchison:1986}
Aitchison J. (1986). \textit{The statistical analysis of compositional data}. Chapman and Hall.

	\bibitem[{Fernandes et al.(2014)}]{Fernandes:2014}
Fernandes, A. D., Reid, J. N., Macklaim, J. M., McMurrough, T. A., Edgell, D. R., and Gloor, G. B. (2014).
Unifying the analysis of high-throughput
sequencing datasets: characterizing RNA-seq,
16S rRNA gene sequencing and selective
growth experiments by compositional data
analysis.
\textit{Microbiome} \textbf{2}, 15.

	\bibitem[{Mandal et al.(2015)}]{Mandal:2015}
Mandal, S., Treuren, W. V., White, R. A., Eggesbø, M., Knight, R., and Peddada, S. D. (2015).
Analysis of composition of microbiomes: a novel method for studying microbial composition.
\textit{Microbial Ecology in Health \& Disease} \textbf{26}, 27663.

	\bibitem[{Lin \& Peddada(2020)}]{Lin:2020}
Lin, H. and Peddada, S. D. (2020).
Analysis of compositions of microbiomes with bias correction.
\textit{Nature Communications} \textbf{11}, 3514.

	\bibitem[Weiss et al.(2017)]{Weiss:2017}
Weiss, S., Xu, Z. Z., Peddada, S., Amir, A., Bittinger, K., Gonzalez, A., Lozupone, C., Zaneveld, J. R., V$\acute{\text{a}}$zquez-Baeza, Y., Birmingham, A., et al. (2017).
Normalization and microbial differential abundance strategies depend upon data characteristics.
\textit{Microbiome} \textbf{5}, 27.

	\bibitem[{Hawinkel et al.(2019)}]{Hawinkel:2019}
Hawinkel, S., Mattiello, F., Bijnens, L., and Thas, O. (2019).
A broken promise: microbiome differential abundance
methods do not control the false discovery rate.
\textit{Briefings in Bioinformatics} \textbf{20}, 210--221.

	\bibitem[{Faust et al.(2015)}]{Faust:2015}
Faust, K., Lahti, L., Gonze, D., Vos, W. M. D., and Raes, J. (2015).
Metagenomics meets time series analysis: unraveling microbial community dynamics.
\textit{Current Opinion in Microbiology} \textbf{25}, 56--66.

	\bibitem[{Lewis et al.(2015)}]{Lewis:2015}
Lewis, J. D., Chen, E. Z., Baldassano, R. N., Otley, A. R., Griffiths, A. M., Lee, D., Bittinger, K., Bailey, A., Friedman, E. S., Hoffmann, C., et al. (2015).
Inflammation, antibiotics, and diet as environmental stressors
of the gut microbiome in pediatric crohn’s Disease.
\textit{Cell Host \& Microbe} \textbf{18}, 489--500.

	\bibitem[{Schubert et al.(2014)}]{Schubert:2014}
Schubert, A. M., Rogers, M.~A.~M., Ring, C., Mogle, J., Petrosino, J.~P., Young, V. B., Aronoff, D. M., and Schloss, P. D. (2014).
Microbiome data distinguish patients with clostridium difficile
infection and non-C. difficile-associated diarrhea from healthy
controls.
\textit{mBio} \textbf{5},  e01021-14.

	\bibitem[{Morgan et al.(2012)}]{Morgan:2012}
Morgan, X. C., Tickle, T. L., Sokol, H., Gevers, D., Devaney, K. L., Ward, D. V., Reyes, J. A., Shah, S. A., LeLeiko, N., Snapper, S. B., et al. (2012).
Dysfunction of the intestinal microbiome in
inflammatory bowel disease and treatment.
\textit{Genome Biology} \textbf{13}, R79.

	\bibitem[Scher et al.(2013)]{Scher:2013}
Scher, J. U., Sczesnak, A., Longman, R. S., Segata, N., Ubeda, C., Bielski, C., Rostron, T., Cerundolo, V., Pamer, E. G., Abramson, S. B., et al. (2013). 
Expansion of intestinal \textit{Prevotella copri} correlates with
enhanced susceptibility to arthritis.
\textit{eLife} \textbf{2}, e01202.

	\bibitem[{Charlson et al.(2010)}]{Charlson:2010}
Charlson, E. S., Chen, J., Custers-Allen, R., Bittinger, K., Li, H., Sinha, R., Hwang, J., Bushman, F. D., and Collman, R. G. (2010).
Disordered microbial communities in the upper respiratory tract of cigarette smokers.
\textit{PloS One} \textbf{5}, e15216.

	\bibitem[{Gonzalez(2018)}]{Gonzalez:2018}
Gonzalez, A., Navas-Molina, J. A., Kosciolek, T., McDonald, D., V$\acute{\text{a}}$zquez-Baeza, Y., Ack- ermann, G., DeReus, J., Janssen, S., Swafford, A. D., Orchanian, S. B., et al. (2018).
Qiita: rapid, web-enabled microbiome meta-analysis.
\textit{Nature Methods} \textbf{15}, 796--798.

	\bibitem[{Lex et al.(2014)}]{Lex:2014}
Lex, A., Gehlenborg, N., Strobelt, H., Vuillemot, R., and Pfister, H. (2014). UpSet: visualization of intersecting sets. \textit{IEEE Transactions on Visualization and Computer Graphics} \textbf{20}, 1983--1992.

	\bibitem[{Bates et al.(2015)}]{Bates:2015}
Bates, D., M$\ddot{\text{a}}$chler, M., Bolker, B. M., and Walker, S. C. (2015).
Fitting linear mixed-effects models using lme4.
\textit{Journal of Statistical Software} \textbf{67}.

	\bibitem[{Carpenter \& Kenward(2012)}]{Carpenter:2012}
Carpenter, J. and Kenward, M. (2012). 
\textit{Multiple imputation and its application}.
John Wiley \& Sons.

	\bibitem[Quinn et al.(2018)]{Quinn:2018}
Quinn, T. P., Erb, I., Richardson, M. F., and Crowley, T. M. (2018).
Understanding sequencing data as compositions: an outlook and review.
\textit{Bioinformatics} \textbf{34}, 2870--2878.

	\bibitem[{Chen et al.(2012)}]{Chenj:2012}
Chen, J., Bittinger, K., Charlson, E. S., Hoffmann, C., Lewis, J., Wu, G. D., Collman, R. G., Bushman, F. D., and Li, H. (2012.) Associating microbiome composition with environmental covariates using generalized UniFrac distances. 
\textit{Bioinformatics}, \textbf{28} 2106--2113.

	\bibitem[{Chen \& Zhang(2021)}]{Chenj:2021}
Chen, J. and Zhang, X. (2022). 
D-MANOVA: fast distance-based multivariate analysis of variance for large-scale microbiome association studies. 
\textit{Bioinformatics} \textbf{38},  286--288.


	\bibitem[{Thorsen et al.(2016)}]{Thorsen:2016}
Thorsen, J., Brejnrod, A., Mortensen, M., Rasmussen, M. A., Stokholm, J., Al-Soud, W. A., Sørensen, S., Bisgaard, H., and Waage, J. (2016).
Large-scale benchmarking reveals false discoveries and count transformation sensitivity in 16S rRNA gene amplicon data analysis methods used in microbiome studies.
\textit{Microbiome} \textbf{4}, 62.

	\bibitem[Zhou \& Gallins(2019)]{Zhou:2019}
Zhou, Y. H. and Gallins, P. (2019). 
A review and tutorial of machine learning methods for microbiome host trait prediction. 
\textit{Frontiers in Genetics} \textbf{10}, 579.

	\bibitem[{McDonald et al.(2018)}]{McDonald:2018}
McDonald, D., Hyde, E., Debelius, J. W., Morton, J. T., Gonzalez, A., Ackermann, G., Aksenov, A. A., Behsaz, B., Brennan, C., Chen, Y., et al. (2018).
American Gut: an Open Platform for Citizen Science Microbiome Research.
\textit{mSystems} \textbf{3}, e00031-18.

	\bibitem[Parzen(1962)]{Parzen:1962}
Parzen, E. (1962).
On estimation of a probability density function and mode.
\textit{Annals of Mathematical Statistics} \textbf{33}, 1065--1076.

	\bibitem[Storey et al.(2004)]{Storey:2004}
Storey, J. D., Taylor, J. E., and Siegmund, D. (2004).
Strong control, conservative point estimation and simultaneous conservative consistency of false discovery rates: a unified approach.
\textit{Journal of the Royal Statistical Society, Series B} \textbf{66}, 187--205.

	\bibitem[Wu et al.(2011)]{Wu:2011}
Wu, G. D., Chen, J., Hoffmann, C., Bittinger, K., Chen, Y.-Y., Keilbaugh, S. A., Bewtra, M., Knights, D., Walters, W. A., Knight, R., et al. (2011).
Linking long-term dietary patterns with gut microbial enterotypes.
\textit{Science} \textbf{334}, 105-108.

	\bibitem[{Silverman et al.(2020)}]{Silverman:2020}
Silverman, J. D., Roche, K., Mukherjee, S., and David, L. A. (2020).
Naught all zeros in sequence count data are the same.
\textit{Computational and Structural Biotechnology Journal} \textbf{18}, 2789--2798.

	\bibitem[{Kaul et al.(2017)}]{Kaul:2017}
Kaul, A., Mandal, S., Davidov, O., and Peddada, S. D. (2017).
Analysis of microbiome data in the presence of excess zeros.
\textit{Frontiers in Microbiology} \textbf{8}, 2114.

	\bibitem[{Kurtz et al.(2015)}]{Kurtz:2015}
Kurtz, Z. D., M$\ddot{\text{u}}$ller, C. L., Miraldi, E. R., Littman, D. R., Blaser, M. J., and Bonneau, R. A. (2015).
Sparse and compositionally robust inference of microbial ecological networks.
\textit{PLoS Computational Biology} \textbf{11}, e1004226.

	\bibitem[{Connolly(2014)}]{Connolly:2014}
Connolly, S. R., MacNeil, M. A., Caley, M. J., Knowlton, N., Cripps, E., Hisano, M., Thibaut, L. M., Bhattacharya, B. D., Benedetti-Cecchi, L., Brainard, R. E., et al. (2014).
Commonness and rarity in the marine biosphere.
\textit{Proceedings of the National Academy of Sciences} \textbf{111}, 8524–-8529.

	\bibitem[{Zhou et al.(2020)}]{Zhou:2020}
Zhou, H., Zhang, X., and Chen, J. (2021).
Covariate adaptive family-wise error rate control for
genome-wide association studies.
\textit{Biometrika} \textbf{108}, 915--931.

	\bibitem[{Vershynin(2018)}]{Vershynin:2018}
Vershynin, R. (2018).
\textit {High-dimensional probability: An introduction with applications in data science.} 
Cambridge University Press.

	\bibitem[{Wainwright(2019)}]{Wainwright:2019}
Wainwright, M. J. (2019).
\textit{High-dimensional statistics: A non-asymptotic viewpoint.} 
Cambridge University Press.

	\bibitem[{Cao et al.(2021)}]{Cao:2021}
Cao, H., Chen, J., and Zhang, X. (2021).
Optimal false discovery rate control for large scale multiple testing with auxiliary information.
\textit{The Annals of Statistics}, forthcoming. \texttt{https://web.stat.tamu.edu/~zhangxiany/Order-FDR.pdf}.
\end{thebibliography}
\end{document}